\documentclass[12pt]{article}
\usepackage{amssymb}
\usepackage{amsmath,bbm}
\usepackage{amsfonts}
\usepackage{geometry}
\usepackage{setspace}
\usepackage[font={small}]{caption}
\usepackage{chbibref}
\usepackage{float}
\usepackage{color}
\usepackage{appendix}
\usepackage{natbib}
\usepackage{hyperref}
\usepackage{enumerate}
\usepackage{graphicx}
\usepackage{caption}
\usepackage{subcaption}
\usepackage{afterpage}
\usepackage{pdflscape}
\usepackage{booktabs}
\usepackage{multirow}

\hypersetup{
    colorlinks,
    citecolor=blue,
    filecolor=black,
    linkcolor=blue,
    urlcolor=black
}
\newtheorem{theorem}{Theorem}

\newtheorem{assumption}{Assumption}

\newtheorem{example}{Example}
\newtheorem{lemma}{Lemma}

\newenvironment{proof}[1][Proof]{\noindent \textbf{#1.} }{\  \rule{0.5em}{0.5em}}

\setbibref{References}{\sffamily}
\allowdisplaybreaks[1]

\begin{document}

\title{\bf Bounds on Average Effects  in \\ Discrete Choice Panel Data Models\thanks{
We thank Victor Aguirregabiria, 
 Iv\'{a}n Fern\'{a}ndez-Val,
Hide Ichimura, Christophe Gaillac, Jiaying Gu, Bo Honor\'{e}, Louise Laage, Laura Liu, Elena Manresa, Whitney Newey, Claudia Noack, Alexandre Poirier, the Editor and three anonymous Referees, and participants at various seminars and workshops for helpful comments and suggestions.
This research was
supported by   the Economic and Social Research Council through the ESRC Centre for
Microdata Methods and Practice (grant numbers RES-589-28-0001, RES-589-28-0002 and ES/P008909/1), and by  the
European Research Council grants ERC-2014-CoG-646917-ROMIA and
ERC-2018-CoG-819086-PANEDA. 
The authors would like to acknowledge the use of the University of Oxford Advanced Research Computing (ARC). }}
\author{\setcounter{footnote}{2}Cavit Pakel\thanks{%
Department of Economics, University of Oxford. Email: \texttt{cavit.pakel@economics.ox.ac.uk} } \and Martin Weidner%
\thanks{%
Department of Economics, University College London. 
Email: \texttt{m.weidner@ucl.ac.uk} } }
\date{September 2026}

\maketitle
\thispagestyle{empty}
\vspace{-30pt}
\begin{abstract}

\noindent
In discrete choice panel data, estimation of average effects is crucial for quantifying the effect of covariates, and for policy evaluation and counterfactual analysis. However, in short panels with individual-specific effects, challenges arise due to partial identification and the incidental parameter problem. In particular, estimating the sharp identified set of average effects becomes impractical when covariates have large support sets, such as when they are continuous. This paper proposes a method for estimating outer bounds on the identified set of average effects, which are easy to construct, converge at the parametric rate, and remain computationally feasible even for moderately large samples. Asymptotically valid confidence intervals are also provided.

\end{abstract}



\newpage

\section{Introduction}

Panel data models with individual-specific effects make it possible to control for unobserved heterogeneity and confounding due to omitted variables that are constant over time. Nonlinear models are required to correctly describe discrete outcomes, and the main complication in such nonlinear panel models is the unknown distribution of unobserved heterogeneity, which constitutes an infinite-dimensional parameter. The fixed effects approach leaves this distribution unspecified, eliminating misspecification concerns (as opposed to the correlated random effects approach which models this distribution parametrically). However, lack of sufficient time-series variation in short panels means that this unknown distribution remains set-identified. An important consequence of this is a general lack of point-identification of average effects. While it is theoretically possible to recover the sharp identified set for average effects, in empirically relevant panel dimensions this often becomes an infeasible task due to a curse of dimensionality. This is a serious issue because average effects are typically the ultimate object of interest, especially from the policy perspective. In this paper, considering a general semiparametric setting, we propose alternative outer bounds which are simple to obtain and remain free of the curse of dimensionality in empirically relevant settings.

Formally, let $Y_i=(Y_{i1},\ldots,Y_{iT})$ be the vector of observed outcomes for individual $i=1,\ldots,n$, where $T$ is the number of time periods and $n$ is the number of cross-sectional units. Throughout, we assume that $n\to\infty$ but $T$ remains fixed. The semiparametric panel models we consider in this paper describe the distribution of $Y_i$ conditional on a vector of observed conditioning variables $Z_i$ as
\begin{align}
   f_{Y|Z}(y_i|z_i) \, &= \, \int_{\cal A} \,
             f(y_i \, | \,z_i,a_i;\beta) \,
              \pi(a_i|z_i) \; d a_i .	\label{ModelDist}              
\end{align}
Here, $f(y_i \, | \,z_i,a_i;\beta)$ is the distribution of $Y_i$ conditional on $Z_i$ and the (vector of) unobserved individual effects $A_i$, and it is assumed to be known up to the finite dimensional parameter $\beta$. The distribution of $A_i$ conditional on $Z_i$, given by $\pi(a_i|z_i)$, is left unrestricted. Both $\beta$ and $\pi = \pi(a_i|z_i)$ are unknown. The vector of conditioning variables usually consists of observed covariates $(X_{i1},\ldots,X_{iT})$ and/or initial conditions $(Y_{i0},Y_{i,-1},\ldots)$. 
Given the true distribution of $Y_i$ conditional on $Z_i$, the identified
set for the model parameters consists of all pairs $(\beta,\pi)$ that satisfy \eqref{ModelDist}.

In empirical research, the 
ultimate
object of interest is 
generally
an average effect of the form 
\begin{equation}
	\overline m := \mathbb{E} \, m(Z_i,A_i,\beta), \label{APE}
\end{equation}
where $m(Z_i,A_i,\beta)$ is some function of interest. The exact choice of $m(\cdot,\cdot,\cdot)$ may vary from application to application, leading to different definitions of $\overline{m}$; see, among others, \citet{Chamberlain(84)}, \citet{BlundellPowell(03),BlundellPowell(04)}, \citet{AltonjiMatzkin(05)}, \citet{Wooldridge(05),Wooldridge(05chp)}, \citet{BesterHansen(09)}, \citet{GrahamPowell(12)}, \citet{HoderleinWhite(12)}.\footnote{A different quantity of interest, which we will not consider, is the quantile structural function of \citet{ImbensNewey(09)}.} \citet{AbrevayaHsu21} 
discuss the different average effects used in the literature.

The average effect in \eqref{APE} can be rewritten as
	\begin{align*}
		\overline m = \int_{\cal Z} \int_{\cal A} \, m\left(z_i,a_i,\beta \right) \, \pi(a_i|z_i) \, f_Z(z_i) \, d a_i d z_i,
	\end{align*}
which
clearly 
depends on $(\beta,\pi)$.
In discrete choice models those model parameters (and in particular $\pi$) are usually only partially identified, implying that $\overline m$ is also 
typically only partially identified.

\citet{HonoreTamer06} and \citet{CFHN13} provide methods for obtaining the identified set when covariates are discrete. More recently, there has been an increased interest in the identification and estimation of average effects in various settings; see, e.g., \citet{AC21}, \citet{DDL21}, \citet{LiuPoirierShiu21}, \citet{BM22}, \citet{BMS22}, and \citet{DGK24}.\footnote{Lack of point-identification of $\pi(a_i|z_i)$ does not invariably lead to set-identification of $\overline m$. An interesting contribution in this vein is by \citet{AC21} who obtain point-identification of the average effect with respect to the lagged dependent variable in a dynamic logit model. However, such case-specific results usually remain an exception. A different route is to obtain point-identification of average effects under additional restrictions on the data generating process as in \citet{LiuPoirierShiu21}. In contrast to these approaches, our aim is to provide a method which applies to an arbitrary function $m(Z_i,A_i,\beta)$ in a generic semiparametric framework.
}

Unfortunately, obtaining the sharp identified set is 
often practically infeasible for sample sizes typically encountered in applications, due to a curse of dimensionality. 
This is because obtaining the sharp identified set for $\overline m$ typically requires estimates of the conditional probabilities $f_{Y|Z}(y_i|z_i)$.
Since $Z_i$ usually contains a 
time-vector of (multiple) covariates, 
the curse of dimensionality is obvious
for continuous covariates. However, even with discrete covariates the number of conditional probabilities 
that would need to be estimated 
is usually large. Suppose, for example, $Y_{it}, X_{it} \in \{0,1\}$, and that $Z_i=(X_{i1},\ldots,X_{iT})$. This implies $2^{2T}$ different conditional probabilities $f_{Y|Z}(y_i|z_i)$; for, say, $T=5$ this yields $1,024$ conditional probabilities.
Estimation of objects of such numbers would require a much larger cross-sectional sample size than available in the majority of applications.\footnote{
Some general inference frameworks, like the ones described in
 \cite{chen2011sensitivity}, 
 are in principle applicable to models of the form \eqref{ModelDist}
 and \eqref{APE}. However, we are not aware of any framework
 that addresses the main statistical challenge that we are facing -- namely that $Z_i$ is high-dimensional and 
every realization of $Z_i$ is unique in  standard panel applications.
}

Motivated by this issue, we propose alternative bounds on the average effect $\overline m$ which can be feasibly obtained in realistic data settings. Our proposal is based on  finding appropriate functions $L(Z_i,Y_i,\beta)$ and $U(Z_i,Y_i,\beta)$ such that
\begin{equation*}
    \mathbb E L(Z_i,Y_i,\beta) 
    \leq 
    \mathbb E m(Z_i,A_i, \beta) 
    \leq 
    \mathbb E U(Z_i,Y_i,\beta).
\end{equation*} 
We show that suitable functions $L(\cdot,\cdot,\cdot)$ and $U(\cdot,\cdot,\cdot)$ 
can be obtained by solving an appropriate linear program for each realized value of $Z_i$. Asymptotically valid lower and upper bounds are then given by 
\begin{align*}
    \frac{1}{n} \sum_{i=1}^n L(Z_i,Y_i,\widehat{\beta})
    \qquad
    \text{and}
    \qquad 
    \frac{1}{n} \sum_{i=1}^n U(Z_i,Y_i,\widehat{\beta}),
\end{align*}
respectively, for some appropriate estimator $\widehat \beta$, assuming
that $\beta$ is point-identified.
We prove the validity of the proposed bounds and provide asymptotically valid inference methods on $\overline m$. Our approach allows for discrete, as well as continuous covariates. We also provide computationally feasible methods for obtaining the suggested bounds. Importantly, these do not require searching over the space of possible distributions for $\pi(a_i|z_i)$, but only over the domain of $A_i$ itself. Consequently, implementation of our method is computationally straightforward and fast.

Our proposal differs from the existing literature in several ways. Firstly, we do not propose a different approach to obtaining the sharp identified set for $\overline m$; rather, we obtain outer bounds on this set. This has the virtue of  avoiding the curse of dimensionality associated with the conditioning variable $Z_i$. Indeed, our outer bounds can be feasibly obtained at standard sample sizes even if the vector of conditioning variables $Z_i$ is continuous, or high-dimensional, or takes on many different values within the sample. Secondly, given our general semiparametric setting, the proposed method can easily be applied to different models (and functions $m(Z_i,A_i,\beta)$) of interest, such as the static logit, dynamic logit or the more complicated random coefficient logit models.

\citet{DDL21} propose an alternative method to achieve inference
on $\overline m$. Their paper initially focuses on
inference on the sharp identified set, but they also 
consider ``outer bounds''  (different from ours) that avoid non-parametric estimation of intermediate objects, similar in spirit to 
our results here. However, their approach currently only applies to static logit and ordered logit models (and for several choices of average effects),
while in this paper we consider  general models of the form
\eqref{ModelDist} (and more general average effects of the form
\eqref{APE}).

Throughout the paper, we consider the case
where  $\beta$ is point-identified. However, our approach can easily be extended to models with partially identified $\beta$, and we suggest two different extensions in the Supplementary Appendix. 
We however also note that
methods for point-estimation of $\beta$ are
well-established in the literature, and these methods are regularly used by applied researchers. Indeed, for essentially every type  of discrete outcome variable (e.g. binary, count data, ordered choice, multinomial choice, \ldots) there exist appropriate specifications for $f(y_i \, | \,z_i,a_i;\beta)$ that allow point identification and $\sqrt{n}$-consistent estimation of $\beta$ by the conditional likelihood method. In static models, this approach relies on the availability of a sufficient statistic for $A_i$ (conditional on $Z_i$), which is satisfied in exponential-family models.\footnote{To provide a non-exhaustive list of examples, see, e.g., \citet{Rasch(61)}, \citet{Andersen(70)}, \citet{Chamberlain(80)}, \citet{Chamberlain(85)} for binary choice logit;  \citet{lancaster2000incidental}, \citet{BluGriWin2002} for count data Poisson; and \citet{das_panel_1999}, \citet{Baetschmann2015}, \citet{Muris2017} for ordered choice logit models (using binarization).} In dynamic panel models, 
there also exist
specifications 
where
estimation of $\beta$ via the generalized method of moments is possible.\footnote{See, for example, \citet{honore2020dynamic}, \citet{kitazawa2021transformations} for dynamic binary choice logit; \citet{BluGriWin2002} for dynamic count data Poisson; and
\citet{honore2021dynamicOrdered} for dynamic ordered choice. \citet{HonoreKyriazidou(00)} and \citet{bartolucci2010dynamic} also consider estimation of $\beta$ in dynamic binary choice panel models.} More generally, the functional differencing method of \cite{bonhomme2012functional} can be viewed as a unifying framework for point-estimation of $\beta$ in both static and dynamic panel models 
of the form \eqref{ModelDist}.

Notice also that there are interesting models that do not require estimation of any common parameters $\beta$.
A prominent example is the binary choice random coefficient model, which allows for richer forms of heterogeneity than the classical fixed effects specification; see Example~\ref{ex:rclogit} below and our discussion in Section~\ref{sect:setcomparison}.
An alternative approach to such models uses finite discrete mixtures \citep{BC07, BC10, BC14}, for which \citet{BC13} establish identification conditions in terms of the number of time periods and mixture components.\footnote{The conditions in \citet{BC13} are sufficient rather than necessary, and the number of identifiable components depends on the model and on the number of time periods $T$. In particular, more components can typically be identified than these conditions alone suggest: for example, for $T=7$, identification of up to $11$ components follows from \citet{BC14}, and of up to $10$ components from Proposition~2 of \citet{KS09}.} Our framework accommodates both continuous and discrete specifications for the distribution of unobserved heterogeneity.

The rest of the paper is organized as follows: 
The main idea of our approach is introduced in Section~\ref{sect:bounds}.
Section~\ref{sect:boundconstruction} presents the general construction of our bounds, including the linear programs used to obtain them.
Section~\ref{sec:illustration} provides further discussion of the bounds, including an illustrative example and comparisons to the sharp identified set.
Section~\ref{sect:betahat} addresses inference when common parameters must be estimated, providing two approaches for constructing asymptotically valid confidence intervals.
Sections~\ref{sect:simulations} and~\ref{sect:empiricalanalysis} present simulation evidence and an empirical application to female labor force participation, respectively. 
Section~\ref{sect:conclusion} concludes. 
Proofs and additional results are provided in the Supplementary Appendix.

\section{Bounds on average effects}\label{sect:bounds}

We observe discrete outcomes $Y_i \in {\cal Y}$, and conditioning variables $Z_i \in {\cal Z}$ for a cross-sectional sample of units $i=1,\ldots,n$. Unobserved heterogeneity is modeled through an unobserved latent variable $A_i \in {\cal A}$. The probability of observing $Y_i = y$ conditional on $Z_i=z$ and $A_i = a$ is given by  $f\left(y\, |\, z,a; \beta_0 \right)$ where $\beta_0 \in  \mathcal{B} \subset \mathbb{R}^{\dim \beta}$ and $f  :  \mathcal{Y} \times  \mathcal{Z}  \times \mathcal{A} \times \mathcal{B}  \rightarrow [0,1]$ is a known function.  The joint distribution of the conditioning variables $Z_i$ and $A_i$ is left unspecified. We focus on  panel data models, where $Y_i = (Y_{i1},\ldots,Y_{iT})$ is a  vector of  outcomes $Y_{it} \in {\cal Y}_t$.
The vector of conditioning variables $Z_i$ can, for example, be equal to $X_i = (X_{i1},...,X_{iT})$ in static models, or to $Z_i=(X_i,Y_{i0})$ in dynamic models where $Y_{i0}$ is the initial condition from time period $t=0$.
We assume throughout that the covariates $X_i$ are strictly exogenous, meaning that $(X_{i1},\ldots,X_{iT})$ is independent of the model error terms $(\varepsilon_{i1},\ldots,\varepsilon_{iT})$ 
(as introduced below)
conditional on $A_i$.\footnote{This rules out predetermined (but not strictly exogenous) covariates such as lagged values of other endogenous variables. However, lagged values of the dependent variable $Y_{it}$ itself can be accommodated, as in our dynamic model examples, because these enter through the conditioning set $Z_i$ rather than as covariates~$X_i$.}
In dynamic models, we assume that the initial condition $Y_{i0}$ is observed.
Our goal is to provide inference methods on average effects of the form $\overline{m}:=\mathbb{E} \left[ m\left(Z_i,A_i,\beta_0 \right) \right]$
where $m : \mathcal{Z}  \times \mathcal{A} \times \mathcal{B} \rightarrow \mathbb{R}$ is a known function. 

To focus on the main features and intuition behind our proposed approach, in this section we abstract away from  estimation of  $\beta_0$ and assume that  it  is known. In Section \ref{sect:betahat} we will consider the case where $\beta_0$ is unknown but point-identified. The case of set-identified $\beta$ (along with a simulation analysis for the probit model) is considered in Sections \ref{sec:setidentifiedbeta} and \ref{sec:appendixC} in the Supplementary Appendix.
The random coefficient model in Example \ref{ex:rclogit} below
provides an interesting case where no estimation of $\beta_0$
is necessary, because the model does not feature any such common parameter. In that case, the results in this section are
already fully sufficient for inference on $\overline m$.

While our approach is general enough to accommodate different panel models of interest (including dynamic ones), for illustration purposes we focus on  two running examples.

\begin{example} \label{ex:statlogit} 
In a static binary choice model, outcomes are generated as $Y_{it} = 1\{ X_{it}\beta_0 + A_i \geq \varepsilon_{it} \}$, where $X_{it}$ is a $1\times K$ vector of covariates, ${\rm dim} \, \beta_0 = K$, and $\varepsilon_{it}$ is a logistic or standard normal random variable. Letting $X_{k,it}$ be the $k$-th covariate and $X_{-k,it}$ be a row matrix containing the remaining covariates, typical examples of $m(Z_{i},A_i,\beta_0)$ are
\begin{gather}
    \frac{1}{T} \sum_{t=1}^T \left [ P(Y_{it}=1|X_{k,it}=x_1,X_{-k,it},A_i,\beta_0) -   P(Y_{it}=1|X_{k,it}=x_2,X_{-k,it},A_i,\beta_0) \right ], \label{marg1} 
    \\
    \frac{1}{T} \sum_{t=1}^T \frac{\partial P(Y_{it}=1|X_{it}=x_{it},A_i,\beta_0)}{\partial x_{k,it}}, \label{marg2} 
\end{gather}
for discrete and continuous $X_{k,it}$, respectively, where $x_1, x_2 \in \mathbb{R}$. For binary and multinomial variables, examples are $(x_1=1,x_2=0)$ and $(x_1=x+1,x_2=x)$, for some $x$.
In \eqref{marg2}, $x_{it}$ could be 
the observed value of $X_{it}$ or its time average, or some other value of interest.
\end{example}

\begin{example}\label{ex:rclogit}
    Our second example is the random coefficient binary choice model, given by $Y_{it} = 1 \{ X_{it} \, A_{2,i} + A_{1,i} \geq \varepsilon_{it} \}$, where $A_{1,i} \in \mathbb{R}$, $A_{2,i}\in \mathbb{R}^{\dim X_{it}}$, and $\varepsilon_{it}$ can have the logistic or standard normal distribution. This allows for richer types of heterogeneity which cannot be captured by the classical fixed effects model (see, for example, \citealt{BC07, BC10, BC14}). For simplicity, we consider the static setting, but our approach remains valid if lagged dependent variables are included as regressors. Defining $A_i=(A_{1,i},A_{2,i})$, examples for $m(Z_i,A_i)$ in this case can be generated analogous to \eqref{marg1} and \eqref{marg2}. We will later consider the case of a single discrete covariate $X_{it} $ and focus on
    \begin{equation}
        \displaystyle \frac{1}{T} \sum_{t=1}^T
        \left[
        P(Y_{it}=1|X_{it}=1,A_{i}) -   P(Y_{it}=1|X_{it}=0,A_{i})
        \right]. \label{margrc}
    \end{equation}
\end{example}

\bigskip
\medskip

Our proposal for inference on $\overline m$ is based on the simple idea that suitable non-random functions $L, U : {\cal Z} \times {\cal Y} \times {\cal B} \rightarrow [b_{\min}, b_{\max}]$ which satisfy,
\begin{align}
       \label{eq:boundcondition}
\sum_{y\in \mathcal{Y}}  L\left( z, y, \beta \right)   f\left(y\, |\, z, a; \beta \right)   \leq \,
  m\left(z,a,\beta \right) 
\,  \leq \,\sum_{y\in \mathcal{Y}}  U\left( z, y, \beta \right)   f\left(y\, |\, z, a; \beta \right)   ,
\end{align}
can be used to obtain asymptotically valid bounds on $\overline{m}$. To see how, notice that when evaluated at $\beta_0$,  the condition in \eqref{eq:boundcondition} is equivalent to
\begin{align*}
    \mathbb{E} \left[L(Z_i,Y_i,\beta_0) \, \big|\, Z_i=z, A_i=a \right] 
    \leq 
    m(z,a,\beta_0) 
    \leq  \mathbb{E} \left[U(Z_i,Y_i,\beta_0) \, \big|\, Z_i=z, A_i=a \right],
\end{align*}
which, by the Law of Iterated Expectations, implies that
\begin{align}
    \mathbb{E}[L(Z_i,Y_i,\beta_0)] \leq \overline{m} \leq \mathbb{E}[U(Z_i,Y_i,\beta_0)].
    \label{eq:popbound}
\end{align}
This suggests that asymptotically valid bounds on $\overline m$ are given by 
\begin{align}
     \widehat L &:= \frac 1 n \sum_{i=1}^{n}  L(Z_i,Y_i, \beta_0 ) ,
     &
     \widehat U &:= \frac 1 n \sum_{i=1}^{n}  U(Z_i,Y_i,  \beta_0 ) .
     \label{DefineBoundEstimatesKNOWN}
\end{align}
To formally show this, we impose the following regularity conditions.

\begin{assumption}~
    \label{ass:MAIN}
       
        \begin{enumerate}[(i)]
            \item $(Y_i, Z_i, A_i)$ are independent and identically distributed across $i=1,\ldots,n$.
     
            \item   
            Correct model specification: 
            $P \left( Y_i=y \, \big| \, Z_i = z, \,A_i=a \right) \, = \, f\left(y\, |\, z, a; \beta_0 \right).$

            \item 
            There are known bounds $b_{\min}, b_{\max} \in \mathbb{R}$ such that $b_{\min} \leq  m\left(z,a,\beta  \right) \leq b_{\max}$, 
            for all $z \in {\cal Z}$, $a \in {\cal A}$ and $\beta \in {\cal B}$.
    
        \end{enumerate}
    
\end{assumption} 
Assumption \ref{ass:MAIN}(i) demands cross-sectional sampling.
Assumption \ref{ass:MAIN}(ii) imposes correct specification 
of our parametric model for $Y_i$ conditional on $Z_i$ and $A_i$. This assumption also implies that all covariates contained in $Z_i$ are strictly exogenous (as opposed to pre-determined), as mentioned before.
Assumption~\ref{ass:MAIN}(iii) requires uniform bounds on
the functions $m\left(z,a,\beta  \right)$ that define
the average effect of interest $\overline m$.
This holds for typical choices for $\overline{m}$ such as those in Examples \ref{ex:statlogit} and \ref{ex:rclogit}, and it can easily be confirmed for any given $m(z,a,\beta)$.\footnote{In principle a weaker condition such as $b_{\min} \leq  \mathbb E[m\left(Z_i,A_i,\beta_0  \right) |A_i=a] \leq b_{\max}$ might also be used here, or bounds on
second or higher-order moments of $m\left(Z_i,A_i,\beta_0  \right)$ are also conceivable,  but in all the applications we consider in the paper the original Assumption \ref{ass:MAIN}(iii) holds, and we find it attractive that this assumption can be verified without knowing anything about the data generating process of $Z_i$ and $A_i$.
More generally, Assumption~\ref{ass:MAIN}(iii) could
be replaced by any assumption that guarantees that
${\rm Var}\left[L(Z_i,Y_i, \beta_0 )\right]$, and ${\rm Var}\left[U(Z_i,Y_i, \beta_0 )\right]$ 
are finite in Theorem~\ref{th:ConsistencyKNOWN}.
}
Importantly, we do not put any restriction on the joint distribution of $Z_i$ and $A_i$. In particular, $Z_i$ can be discrete or continuous, $Z_i$ and $A_i$ can be arbitrarily related, and they do not have to be bounded or have bounded moments.

\begin{theorem}
    \label{th:ConsistencyKNOWN}
    Let Assumption \ref{ass:MAIN} hold,
    and let   $L, U : {\cal Z} \times {\cal Y} \times {\cal B} \rightarrow [b_{\min}, b_{\max}]$ satisfy equation \eqref{eq:boundcondition} for $\beta= \beta_0$
    and for all $z \in {\cal Z}$, $a \in {\cal A}$.
    Let  $\overline{m}$, $\widehat L$, $\widehat U$
    be as defined in \eqref{APE}
    and \eqref{DefineBoundEstimatesKNOWN}.  Then,
    \begin{align*}
        \displaystyle \widehat L +  O_p(n^{-1/2} )  \; \leq \;  \overline{m} \; \leq \; \widehat U +  O_p( n^{-1/2} )
        \qquad
        \text{as }
        n \to \infty.
    \end{align*}
    Furthermore, 
    assume that
    ${\rm Var}\left[L(Z_i,Y_i, \beta_0 )\right]>0$, and ${\rm Var}\left[U(Z_i,Y_i, \beta_0 )\right]>0$,
    and define    
    $\widehat{\sigma}^2_L
    := \frac 1 n \sum_{i=1}^n
    [ L(Z_i,Y_i, \beta_0 ) - \widehat L ]^2$
    and
    $\widehat{\sigma}^2_U
    := \frac 1 n \sum_{i=1}^n
    [ U(Z_i,Y_i, \beta_0 ) - \widehat U ]^2$.
    Then, for $\alpha \in [0,1]$, we have
    \begin{align*}
            \liminf_{n\rightarrow \infty }P\left(
        \widehat{L}- \frac{ c_{\alpha/2} \, \widehat{\sigma }_{L}} {\sqrt{n}} \leq \overline m \leq   \widehat{U}+
        \frac{ c_{\alpha/2}\,
        \widehat{\sigma }_{U}} {\sqrt{n}}\right) \geq 1-\alpha ,    
        \quad
        \text{where } \; \;
        \displaystyle c_{\alpha/2}=\Phi^{-1}\left(1- \frac \alpha 2 \right).
        \end{align*}
\end{theorem}

\section{Construction of the bounds}
\label{sect:boundconstruction}

We now introduce
our general construction of the bound
functions $L(z,y,\beta )$
and $U(z,y,\beta )$. 
To concentrate solely on bound construction, in this section we still consider the case with known $\beta_0$. A full theory with estimated $\beta_0$
is provided in Section~\ref{sect:betahat}. 
Implementation remains the same 
independent of whether 
we use
$\beta_0$ or its estimate.

In obtaining asymptotically valid bounds, the key requirement on the functions $L(z,y,\beta )$
and $U(z,y,\beta )$ is that they satisfy \eqref{eq:boundcondition} and that they are bounded. Of course, one wants the estimated bounds on $\overline m$ to be informative, in the sense that the interval in \eqref{eq:boundcondition} is as narrow as possible. At the same time, 
importantly, 
for given $z$ and $\beta$, $L(z,y,\beta )$
and $U(z,y,\beta )$ have to be chosen such that \eqref{eq:boundcondition} holds for all $a\in \mathcal A$. This can be reformulated as a standard optimization problem. Namely, for any given $z \in {\cal Z}$ and $\beta \in {\cal B}$ we can choose $L(z,y,\beta)=\ell(y)$ and $U(z,y,\beta)=u(y)$
as solutions to the following optimization problem with some appropriate objective function $Q(\ell(\cdot),u(\cdot),z,\beta)$,\footnote{
The solutions $L(z,y,\beta)=\ell(y)$ and $U(z,y,\beta)=u(y)$ may not be unique. But in a practical
implementation some concrete solution
will still be obtained by the specific linear solver used for implementation, 
and Theorem~\ref{th:ConsistencyKNOWN} will
still be valid, since it only depends on
the constraints being satisfied.
}
\begin{align}
\min_{\ell,u \, : \,  \mathcal{Y} \rightarrow \mathbb{R} }\,  &Q(\ell(\cdot),u(\cdot),z,\beta)
\notag \\
  \text{subject} & \text{ to} \label{eq:GeneralOptimization}  \\
&   \forall y\in \mathcal{Y}:\, \,b_{\min }\leq \ell(y)\leq 
u(y)\leq b_{\max }  \notag \\
\quad \text{and}\quad & \forall a\in \mathcal{A}:\, \, \sum_{y\in \mathcal{Y}}\ell (y) \, f(y \, | \, z, a;\beta )
 \leq m(z,a,\beta )\leq \sum_{y\in \mathcal{Y}}u(y) f(y \, | \, z, a;\beta ).    \notag
\end{align}
In the current setting where we assume that $\beta_0$ is known, \eqref{eq:GeneralOptimization} will be solved at $\beta=\beta_0$. When $\beta_0$ is estimated, \eqref{eq:GeneralOptimization} will be solved at some estimate $\beta=\widehat{\beta}$. When no common parameter is estimated (as in Example \ref{ex:rclogit}), the objective function and the constraints will be free of $\beta$.

The restrictions of the program \eqref{eq:GeneralOptimization} guarantee the conditions of Theorem \ref{th:ConsistencyKNOWN} and also impose that $\ell(y) \leq u(y)$. Consequently, any choice of the objective function $Q(\ell(\cdot),u(\cdot),z,\beta)$ yields valid bounds with $\widehat L \leq \widehat U$. 
It is important to stress that in order to construct the bounds $\widehat L$ and $\widehat U$
we only need to solve the 
program in \eqref{eq:GeneralOptimization}  once for every $i\in \{1,...,n\}$ at $z=Z_i$.   Contrary to the sharp identified set, construction of our bounds does not involve conditional choice probabilities, and therefore remains free of the curse of dimensionality.

Display \eqref{eq:GeneralOptimization} states our approach to obtaining bounds in its most general form, in the sense that the econometrician can choose any objective function $Q(\ell(\cdot),u(\cdot),z,\beta)$ that she sees fit. It is computationally attractive to consider objective functions which turn the optimization problem into a linear program, and we now discuss two
such
intuitive 
choices.\footnote{
In particular cases, it might be possible to
obtain analytic expressions for the bound functions.
But for the class of semi-parametric panel models and average effects introduced
in Section~\ref{sect:bounds},
 it is unlikely  that analytic expressions for the bounds can be obtained in general.
The distinction between a numerical method
and analytic expressions for the bound functions
 is analogous to the distinction between the functional differencing method in \cite{bonhomme2012functional}
and the analytical moment functions 
in \cite{honore2020dynamic} for the purpose of inference
on $\beta$. 
}

\subsection{Baseline linear program} \label{sec:baselinelp}
A linear program can be implemented by using the objective function
\begin{align}
  Q(\ell(\cdot),u(\cdot),z,\beta) = \int_{\mathcal{A}} \,   \sum_{y\in \mathcal{Y}}\left[ u(y)-\ell(y) \right] 
\,f(y \, | \, z,a;\beta )\,p(a|z)\,da ,  \label{eq:linprog1}
\end{align}
where $p(a|z)$ is some (potentially non-proper) ``prior distribution''. 
In the absence of any additional information on $a$, as is the case in this paper, one can simply use $p(a|z)=1$. Indeed, throughout all the applications based on this baseline linear program, we use the flat prior $p(a|z)=1$.
We note, however, that our bounds are valid for
any choice of ``prior''. When available, any extra information on the distribution of unobserved heterogeneity can be incorporated into the choice of $p(a|z)$, but we do not pursue this here.

\subsection{Uniform linear program}\label{sec:uniflp}
If we are unwilling to specify a prior $p(a|z)$, then we can choose 
the objective function
\begin{align}
  Q(\ell(\cdot),u(\cdot),z,\beta)
  =
  \max_{a\in \mathcal{A}} \left[ \sum_{y\in \mathcal{Y}}\left[ u(y)-\ell(y) \right] \,f(y \, | \, z,a;\beta ) \right] ,  
  \label{eq:linprog2}
\end{align}
where, instead of integrating over $a\in \mathcal{A}$ with a 
prior distribution, we choose the worst-case value of 
$a\in \mathcal{A}$ that maximizes the 
expected bounds 
$\sum_{y\in \mathcal{Y}}\left[ u(y)-\ell(y) \right] \,f(y \, | \, z,a;\beta )$. Hence, we call the ensuing approach the \textit{uniform linear program}.
To be precise, 
\eqref{eq:linprog2}
cannot be used directly to yield a linear program since it is not linear in 
$u(y)$ and $\ell(y) $; however, an equivalent representation of this problem as a linear program is obtained 
as follows:
\begin{align}
\min_{\left\{s \in \mathbb{R} , \; \ell,u \, : \,  \mathcal{Y} \rightarrow \mathbb{R} \right\}}\,  s
\notag \\
  \text{subject to } \quad & 
  \label{eq:UnifOptimization}  
  \\
 &   \forall y\in \mathcal{Y}:\, \,b_{\min }\leq \ell(y)\leq 
u(y)\leq b_{\max } 
\notag \\
  &  \forall a\in \mathcal{A}:\, \, 
 \sum_{y\in \mathcal{Y}}\left[ u(y)-\ell(y) \right] \,f(y \, | \, z,a;\beta ) \leq s
\notag \\  
\quad \text{and}\quad & \forall a\in \mathcal{A}:\, \, \sum_{y\in \mathcal{Y}}\ell (y) \, f(y \, | \, z, a;\beta )
 \leq m(z,a,\beta )\leq \sum_{y\in \mathcal{Y}}u(y) f(y \, | \, z, a;\beta ).    \notag
\end{align}
In this linear program, the variable set is extended by
$s \in \mathbb{R}$. 
When profiling out $s \in \mathbb{R}$ from this program
one finds that for given $\ell,u \, : \,  \mathcal{Y} \rightarrow \mathbb{R}$ the optimal $s$ is given by
\begin{align}
    s = \max_{a\in \mathcal{A}} \left[ \sum_{y\in \mathcal{Y}}\left[ u(y)-\ell(y) \right] \,f(y \, | \, z,a;\beta ) \right] ,
    \label{eq:s}
\end{align}
which is identical to the objective function in \eqref{eq:linprog2}. Thus,
solving the linear program
in \eqref{eq:UnifOptimization} gives the desired bound functions
$L(z,y,\beta)=\ell(y)$ and $U(z,y,\beta)=u(y)$ that correspond
to choosing the objective function 
\eqref{eq:linprog2}
in our general 
program \eqref{eq:GeneralOptimization}.

\subsection{Practical recommendations} \label{sec:practicalrec}
Both the baseline and the uniform linear programs are valid options and will yield valid outer bounds. 
The uniform program focuses on the worst-case value of $a\in \mathcal A$ and is therefore likely to yield (slightly) wider outer bounds compared to the baseline linear program.\footnote{In preliminary analysis available upon request, we have considered both linear programs for the four models analyzed in the simulation studies of Section \ref{sect:setcomparison} and Appendix \ref{app:dynamic}. Our results reveal for the static logit, dynamic logit and the random coefficient static logit models that while the  uniform linear program led to some widening of the bounds, the change was not substantial. The only significant change was observed for the random coefficient dynamic logit model.}
Nevertheless, our general recommendation is to use the uniform linear program, as it does not require specifying a prior $p(a|z)$ and therefore requires less input from the researcher.

In practice, implementation of the linear programs \eqref{eq:GeneralOptimization} and \eqref{eq:UnifOptimization} requires choosing a grid $\mathcal{A}_g \subset \mathcal{A}$ to approximate the constraints. For logit-based models, computational efficiency can be improved by rewriting the constraints in terms of sufficient statistics. These and other implementational details are discussed in Section~\ref{sec:compdetails} of the Supplementary Appendix.

\section{Further discussion of the bounds}\label{sec:illustration}

\subsection{An illustrative example}\label{sect:iexample}

The following example simply corresponds to the
nonparametric bounds in \citet{CFHN13}.
It is therefore not representative
 of  how we obtain the bounds in this 
paper in general, but we still find 
the example instructive, 
since it provides analytical expressions
for bounds   satisfying 
\eqref{eq:boundcondition}.

 We consider the static binary choice model of Example~\ref{ex:statlogit} for the case where $X_{it}\in\{0,1\}$ is the only covariate and
the error term $\varepsilon_{it}$ is stationary over time $t$. 
 The average effect  is given
by~\eqref{marg1} with $x_1=1$ and $x_2=0$, that is,
\begin{align}
    m(A_i,\beta_0) 
    &= \mathbb{E}\left[Y_{it} \, \big| \, X_{it}=1,A_i,\beta_0\right] 
    -  \mathbb{E}\left[Y_{it} \, \big| \,X_{it}=0,A_i,\beta_0\right],
    \label{AvTreatmentEffect}
\end{align}
where the time averaging is not needed due to stationarity.\footnote{In this case, $Z_i=X_i=(X_{i1},\ldots,X_{iT})$. Notice that, contrary to the general case, here $m(A_i,\beta_0)$ does not depend on $X_i$.
This is because 
the average effect is calculated with respect to  specific  values of $X_{it}$ and there are no other covariates.} 
For $d \in \{0,1\}$, let
\begin{align*}
   v(X_i,d) := \left\{
   \begin{array}{ll}
      0 & \text{if $X_{it}=1-d$ for all $t \in \{1,\ldots,T\}$,}
      \\
      1 & \text{if $X_{it}=d$ for some $t \in \{1,\ldots,T\}$.}
   \end{array}
  \right.
\end{align*}
For $v(X_i,d)=1$ we define $\overline{Y}(X_i,Y_i,d) := \sum_{t \in {\cal T}(X_i,d)} Y_{it} / \left| {\cal T}(X_i,d) \right|$ to be the average of 
$Y_{it}$ over those time periods ${\cal T}(X_i,d)=\{ t \, : \, X_{it}=d\}$ where $X_{it}$ equals $d$. For $v(X_i,d)=0$ we simply let $\overline{Y}(X_i,Y_i,d) :=0$.\footnote{
Essentially, $\overline{Y}(X_i,Y_i,d)$ can be defined as any real number, given that its contribution to the bounds will be equal to zero whenever $v(X_i,d)=0$.}
Valid outer bound functions are then given by
\begin{align}
    L(X_i,Y_i)
    &=
    v(X_i,1) \, \overline{Y}(X_i,Y_i,1) 
    - v(X_i,0) \, \overline{Y}(X_i,Y_i,0)  
    -  [1-v(X_i,0)] ,
    \label{AnalyticalBounds0} \\
    U(X_i,Y_i)
    &=
    v(X_i,1) \, \overline{Y}(X_i,Y_i,1) 
    - v(X_i,0) \, \overline{Y}(X_i,Y_i,0)  
    +  [1-v(X_i,1)] .
    \label{AnalyticalBounds}
\end{align}
The stationarity  assumption then guarantees that
\begin{align}
   \mathbb{E} \left[L(X_i,Y_i) \, \big|\, X_i, A_i \right] 
    \leq
    m(A_i,\beta_0) 
    \leq
    \mathbb{E} \left[U(X_i,Y_i) \, \big|\, X_i, A_i \right],
    \label{eq:weakbound}
\end{align}
which is exactly the condition
\eqref{eq:boundcondition} that our bound functions are supposed to satisfy.\footnote{
Due to stationarity we have
\begin{align*}
   & \mathbb{E} \left[L(X_i,Y_i) \, \big|\, X_i, \, v(X_i,0)=v(X_i,1)=1, \, A_i \right] 
   =
    m(A_i,\beta_0) 
    =  \mathbb{E} \left[U(X_i,Y_i) \, \big|\, X_i,\, v(X_i,0)=v(X_i,1)=1,\,A_i \right],
\end{align*}
while for $v(X_i,d)=0$,  the above
bounds $L(X_i,Y_i)$ and $U(X_i,Y_i)$ simply revert
to the appropriate worst-case bounds (zero or one) that are possible
for the unidentified  expectations.}

Again, we want to point out that this
 example is not  characteristic  of   our bounds more generally.
In particular, here $L(X_i,Y_i)$ and
$U(X_i,Y_i)$ do not depend on $\beta_0$, and
neither the single-index
structure  $X_{it}\beta_0 + A_i + \varepsilon_{it}$
nor the parametric assumption on the
error distribution
are utilized to show validity of the bounds --- the bounds here
are valid for  any model $Y_{it}=g(X_{it},A_i,\varepsilon_{it})$,
as long as the function $g(\cdot,\cdot,\cdot)$ is constant over $t$, and 
the conditional distribution of the shocks $\varepsilon_{it}$ is stationary over $t$.

From the corresponding 
  discussion
in \citet{CFHN13}
we also know that, as $T \rightarrow \infty$, the width of
 these bounds, $ \mathbb{E}[U(X_i,Y_i) - L(X_i,Y_i)]$, shrinks proportionally  to 
 the probability of $X_{it}$ being constant over $t$.
 Under appropriate distributional assumptions on $X_{it}$ (e.g.\ $X_{it}$ independent across $t$ and random), this implies that  the width of
 the bounds shrinks exponentially fast in $T$.
While we do not explore large-$T$ analysis 
here, we suspect that similar results hold more generally for the bounds in this paper.

\subsection{Comparison to the identified set}
\label{subsec:Comparision1}

The key difference between our outer bounds, $ \mathbb{E} U(X_i,Y_i)$ 
and $\mathbb{E} L(X_i,Y_i)$,
and the identified set for 
$\overline m$ is how they depend on the conditional choice probabilities, $P(Y_i|X_i)$.
While our outer bounds are linear functions of choice probabilities, the upper and lower
 boundaries  of the identified set are complicated nonlinear functions
of $P(Y_i|X_i)$. The goal of this subsection is to briefly explain this difference
and its consequences for inference
on the average effects.

For simplicity, we stick to the static binary choice example with a single binary
covariate discussed in the last subsection,
and we assume that $\varepsilon_{it}$
has standard logistic distribution.
Let $f(y|x,a;\beta_0)$ be the corresponding conditional distribution of $Y_i|X_i,A_i$.
As long as we have some variation
on the covariates across time, 
  $\beta_0$ is point-identified in this model (see e.g.\ \citealt{Chamberlain(85),Chamberlain(10)}). 

For $x \in \{0,1\}^T$, let
$p(x):=  \big[P(Y_i=y \, |\, X_i=x) \,: \, y \in \{0,1\}^T \big] $ be the $2^{T}$-vector of  choice probabilities conditional on $X_i=x$, and 
define
$\overline m(x) := \mathbb{E}\left[ m(A_i,\beta_0) \, \big| \, X_i=x \right]$.
Next, let $\Pi(x,p(x))$ be the set of conditional distributions $A_i|X_i$ that are compatible with the choice probabilities $p(x)$: that is, we have
$\pi(\cdot|x) \in \Pi(x,p(x))$ if and only if
$P(Y_i=y \, |\, X_i=x)=\int_{\mathbb{R}} f(y|x,a;\beta_0) \pi(a|x) da$.
Since $\beta_0$ and $p(x)$ are point-identified, the only ambiguity in the 
identification of $\overline m(x) $ is due to the unknown distribution of $A_i|X_i$. Then, defining
\begin{align*}
    L_{\rm id}(x,p(x)) &:= \inf_{\pi(\cdot|x) \in \Pi(x,p(x))}
    \int_{\mathbb{R}} m(a,\beta_0) \, \pi(a|x) da ,
  \\
   U_{\rm id}(x,p(x)) &:= \sup_{\pi(\cdot|x) \in \Pi(x,p(x))}
    \int_{\mathbb{R}} m(a,\beta_0) \, \pi(a|x) da ,
\end{align*}
the identified set for $\overline m = \mathbb{E}\left[ \overline m(X_i) \right]$ is given by
$\big[ \mathbb{E} L_{\rm id}(X_i,p(X_i)), \, \mathbb{E} U_{\rm id}(X_i,p(X_i))\big]$. All this is of course well-known.  What we want to highlight here is that the above construction
inevitably yields a complicated nonlinear dependence of the boundaries of the identified set on the observable choice probabilities $p(x)$ through $\Pi(x,p(x))$. In contrast, our bounds
\begin{align*}
   \mathbb{E} \, L(x,Y_i)
   &= \sum_{y \in \{0,1\}^T} L(x,y) \, P(Y_i=y|X_i=x) ,
   \\
   \mathbb{E} \, U(x,Y_i)
   &= \sum_{y \in \{0,1\}^T} U(x,y) \, P(Y_i=y|X_i=x) ,
\end{align*}
are by construction linear functions of the vector of conditional choice probabilities $p(x)$.

This distinction between nonlinearity (for the identified set)
vs linearity (for our outer bounds) in $p(x)$ has a fundamental effect on inference: the sample analogs of our bounds,
 $\frac 1 n \sum_{i=1}^n L(X_i,Y_i)$ and
$\frac 1 n \sum_{i=1}^n U(X_i,Y_i)$, avoid estimating $p(x)$ naturally. In contrast, we are not aware of any inference
procedure on the sharp identified set that would avoid consistent estimation of $p(x)$.\footnote{
  \cite{DDL21}  present two different inference
  procedures for average effects in static panel logit models, one that relies on consistent estimation of $p(x)$, and one that does not.
  In the latter case, they also obtain 
   certain outer bounds
  on the identified set that are different from our proposal. In Section \ref{sec:ddlvspw} in the Supplementary Appendix, we consider a brief comparison between their outer bounds and the ones proposed in this paper.
}
Especially when $p(x)$ is hard to estimate, the nonlinear dependence of the identified set on $p(x)$ can cause significant issues in inference. Hence, as already mentioned in the introduction, 
reliable inference on the identified set is problematic
unless the sample size $n$ is much larger than the number of possible
values for $(X_i,Y_i)$. Our outer bounds are by design immune to this.

\begin{figure}[p]
    \centering
    \begin{minipage}{0.5\textwidth}
        \centering
        \includegraphics[width=1\textwidth]{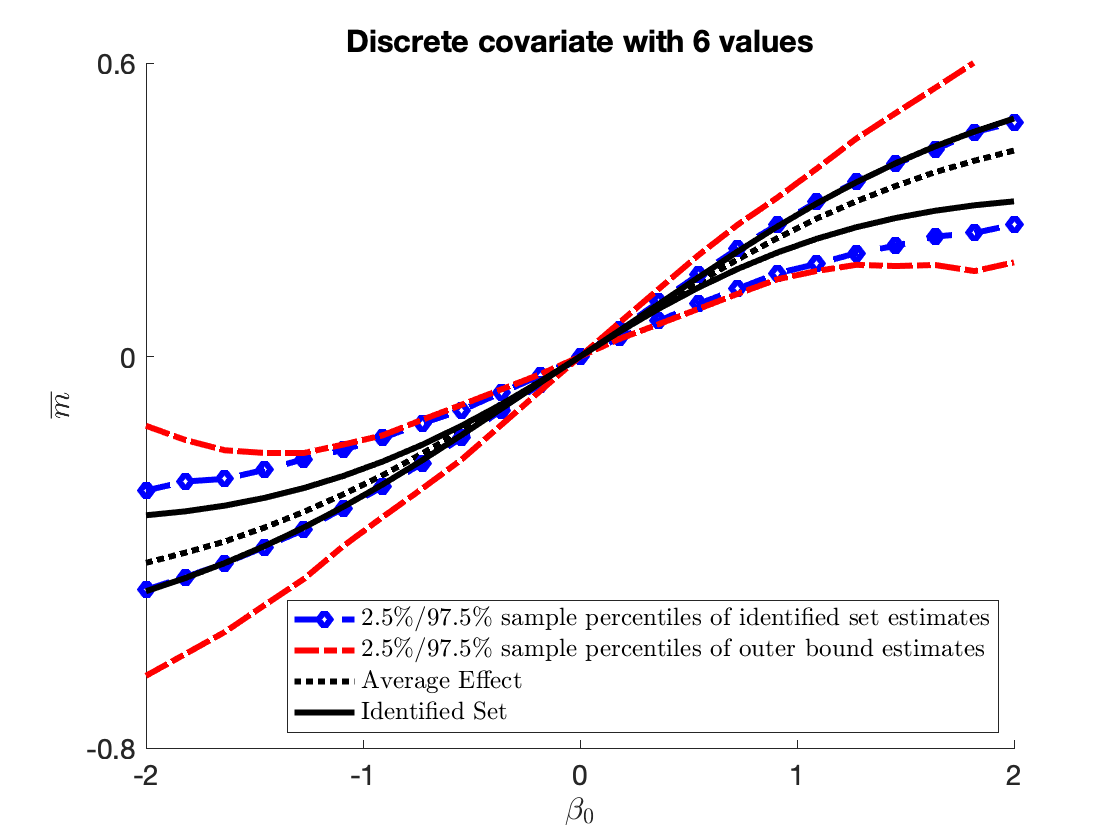}
    \end{minipage}\hfill
    \begin{minipage}{0.5\textwidth}
        \centering
        \includegraphics[width=1\textwidth]{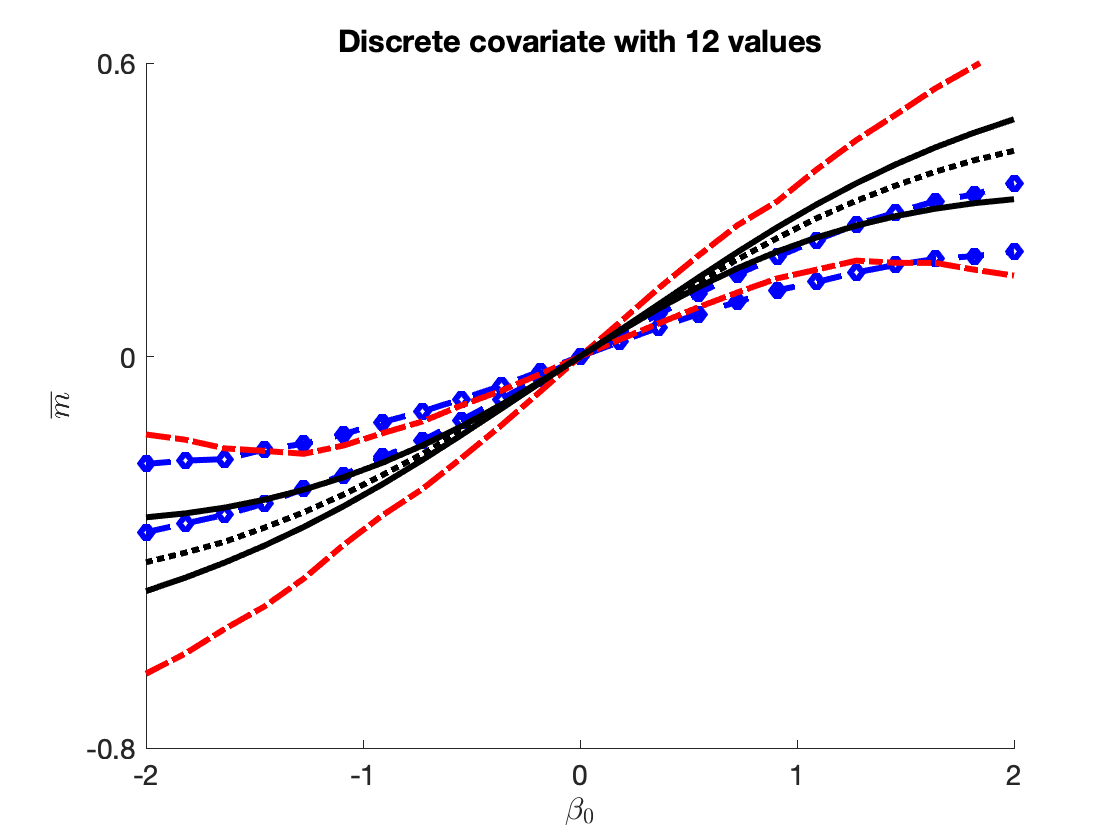}
    \end{minipage}
    \caption{Sample quantiles of estimates of the outer bounds and the identified set. The DGP is $Y_{it} = 1\{ X_{it}\beta + A_{i} \geq \varepsilon_{it}\}$ where 
    $A_{i} \sim N( T^{-1} \sum_{t=1}^T X_{it} - 1/2, 1 )$, $X_{it} = x_{it}/(|\mathcal X|-1)$, $x_{it}$ is discrete uniform with support $[0,|\mathcal X| -1]$, and $\varepsilon_{it}\sim \mathrm{Logit}(0,1)$. The average effect under consideration is $\mathbb{E}[Y_{it} \, \big| \, X_{it}=1,A_i,\beta_0] 
    -  \mathbb{E}[Y_{it} \, \big| \,X_{it}=0,A_i,\beta_0]$. For each $\beta_0 \in [-2,2]$, the quantiles are calculated across 1000 replications of panels with $T=2$ and $n=200$. The left panel contains the results for $|\mathcal X| = 6$ whereas the results for $|\mathcal X| =12$ are presented in the right panel.}
    \label{fig:comb1}
\end{figure}
\afterpage{\clearpage}

To illustrate the points made here, we consider a brief simulation exercise. Let
\begin{gather*}
    Y_{it} = 1\{ X_{it}\beta + A_{i} \geq \varepsilon_{it}\},
    \quad
    A_{i} \sim N\left( \frac{1}{T} \sum_{t=1}^T X_{it} - \frac{1}{2}, 1 \right),
    \quad
    X_{it} = x_{it}/(|\mathcal X|-1),
\end{gather*}
where $\varepsilon_{it}\sim \mathrm{Logit}(0,1)$ and $x_{it}$ is discrete uniform with support $[0,|\mathcal X| -1]$.\footnote{Throughout the paper, $\mathrm{Logit}(0,1)$ denotes a logistic distribution with mean 0 and variance 1.}
Then, $X_{it}$ can take on one of $|\mathcal X|$ equidistant values between 0 and 1.
We consider $|\mathcal X|\in\{6,12\}$. The analysis for either case is based on 1000 replications of panels with $T=2$, $n=200$. The average effect of interest is as in \eqref{AvTreatmentEffect}. For each replication, we obtain the estimated sharp identified set
and our outer bounds based on the
construction in Section~\ref{sect:boundconstruction}.\footnote{
The outer bounds presented here, based on
the construction in
Section~\ref{sect:boundconstruction},
provide much narrower bounds than the simple analytical expressions in \eqref{AnalyticalBounds0}-\eqref{AnalyticalBounds}. This is not surprising, given that our bounds utilize stronger model assumptions.}
Then we report the 2.5\% and 97.5\% sample quantiles of these quantities across all replications.\footnote{More precisely, the reported $2.5\%$ sample percentile for the estimated identified sets corresponds to the $2.5\%$ sample percentile of the estimated lower bounds (of identified sets) across all replications. Similarly, the reported $97.5\%$ sample percentile corresponds to the $97.5\%$ sample percentile of the estimated upper bounds (of identified sets) across all replications.
The sample percentiles for the outer bounds are obtained analogously. See also Section \ref{sec:archident} in the Supplementary Appendix for more information on the calculation of the identified set results.}
Doing so enables us to compare the limits of the estimated confidence intervals, without estimating the confidence bands directly. Results are presented in Figure \ref{fig:comb1}. When $|\mathcal X | = 6$, the lower and upper $2.5\%$ percentiles of the estimated bounds  of the identified set  provide valid coverage. However, when $|\mathcal X|$ increases to 12, the same percentiles  fail to   include the average effect itself almost all the time. This reflects an underlying bias in the estimation of the sharp identified set. The outer bounds are immune to this issue, and still provide valid coverage. This example illustrates that, although the outer bounds are not sharp, they can be more reliable in inference compared to estimators of the sharp identified set itself. The results suggest, as expected, that issues arise as the cardinality of the support of the covariate increases. Therefore, the case with continuous $X_{it}$ will be subject to more pronounced issues.

\subsection{Model-specific comparisons}
\label{sect:setcomparison}

\begin{figure}[p]
	\begin{subfigure}{.5\textwidth}
		\centering
		\includegraphics[width=0.9\linewidth]{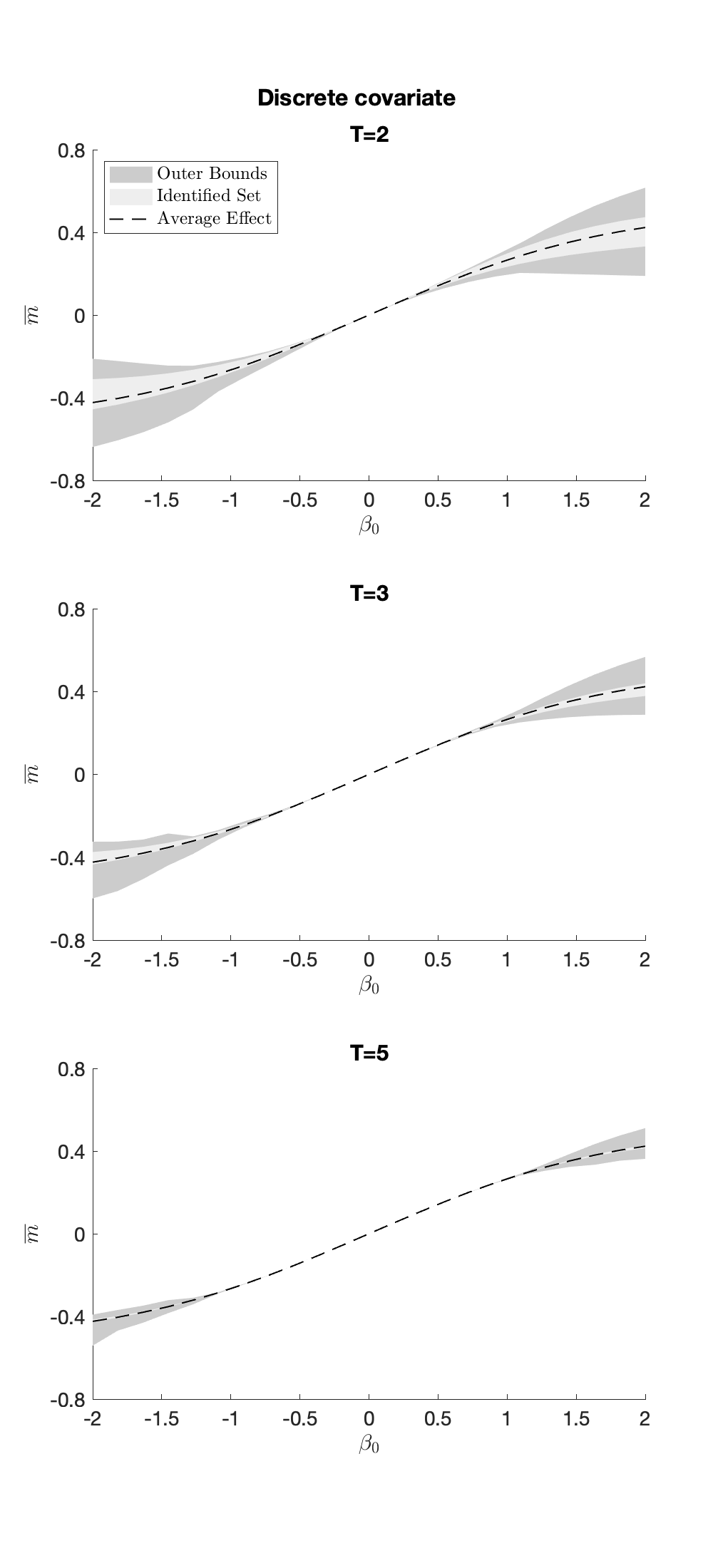}
	\end{subfigure}
	\begin{subfigure}{.5\textwidth}
		\centering
		\includegraphics[width=0.9\linewidth]{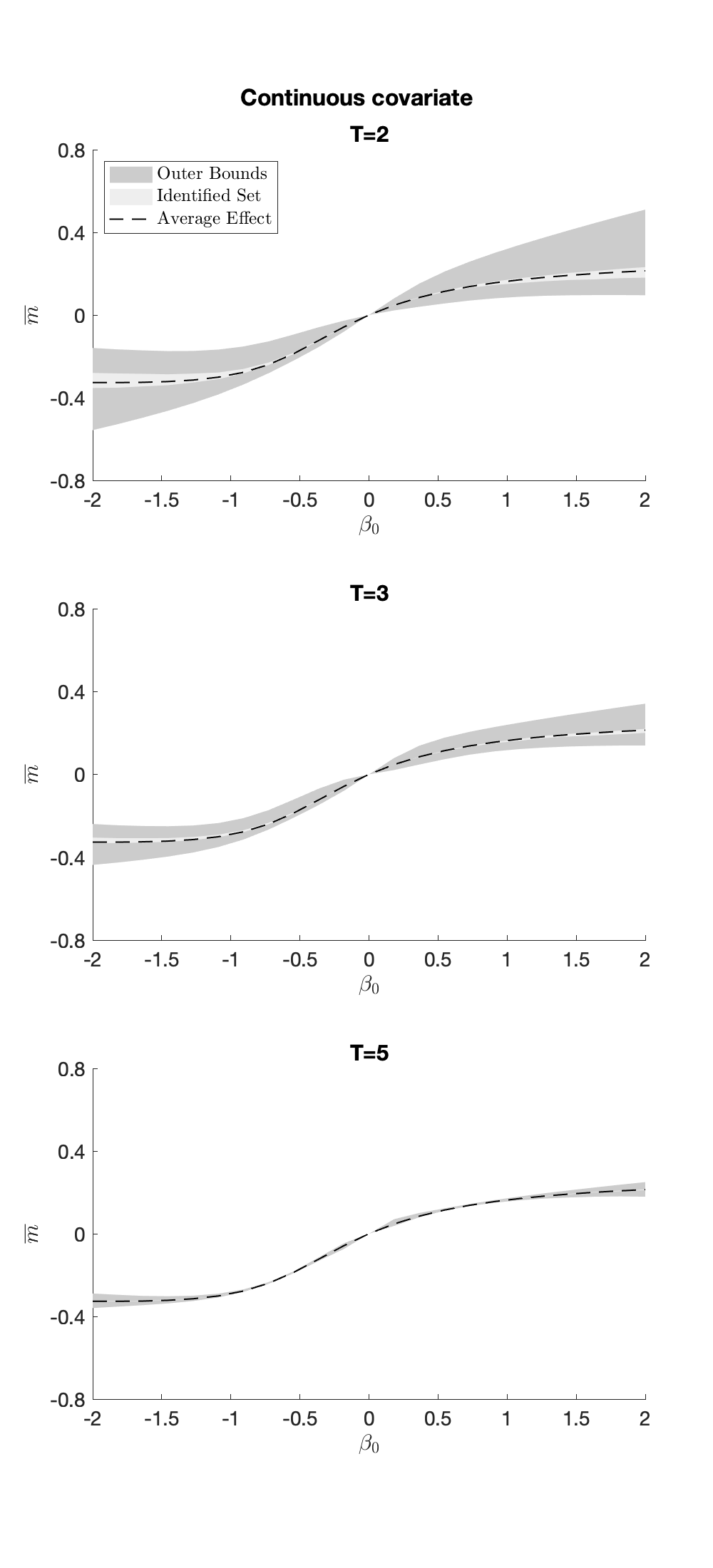}
	\end{subfigure}
	\caption{Comparison of the outer bounds  and the identified set for the static logit model $Y_{it} = 1\left\{ X_{it}\beta + A_i \geq \varepsilon_{it} \right\}$, where $\varepsilon_{it} \sim {\rm Logit} (0,1)$ and $A_i \sim N(0,1)$. Results for each $\beta_0\in[-2,2]$ are based on 1000 replications of panels with cross-section size $n=1000$. Reported outer bounds   are the cross-replication averages. Left panel: single discrete covariate $X_{it} = 1\left\{ A_i \geq \eta_{it} \right\}$ where $\eta_{it} \sim N(0,1)$. The average effect of interest is based on \eqref{marg1} with $(x_1,x_2)=(1,0)$. Right panel: single continuous covariate, $X_{it} \sim N(A_i,1)$. The average effect of interest is based on \eqref{marg2}.}
	\label{fig:idset_logit}
\end{figure}
\afterpage{\clearpage}

We now extend this comparison to the population level for the static logit and random coefficient logit models; the dynamic variants are considered in Appendix~\ref{app:dynamic}. We use the linear program in \eqref{eq:UnifOptimization}. Here we compare our outer bounds to the population sharp identified set, estimation of which is challenging whenever the support of the conditioning variables $Z_i$ is not small relative to the sample size.\footnote{Similar to \citet{CFHN13} we obtain the sharp identified set by solving an appropriate linear program. See Section \ref{sec:archident} in the Supplementary Appendix for a more detailed discussion.} Our simulation results below present the identified sets and our outer bounds. In Section \ref{sec:addsimdetails} in the Supplementary Appendix, we also provide results concerning the widths of the outer bounds and the identified sets.

For static logit we consider both the discrete and continuous covariate cases, with the data generating processes (DGPs) given by
\begin{align}
	Y_{it} = 1\left\{ X_{it}\beta + A_i \geq \varepsilon_{it} \right\}, \quad A_i \sim N(0,1),  \quad X_{it} = 1\left\{ A_i \geq \eta_{it} \right\}, \quad \eta_{it} \sim N(0,1), \label{eq:idlogitdisc}
\end{align}
and
\begin{align}
	Y_{it} = 1\left\{ X_{it} \beta + A_i \geq \varepsilon_{it} \right\}, \quad A_i \sim N(0,1),  \quad X_{it} \sim N(A_i,1), \label{eq:idlogitcont}
\end{align}
respectively. In both cases, $\varepsilon_{it} \sim {\rm Logit}(0,1)$. For the discrete covariate case we consider the average effect based on \eqref{marg1} with $(x_1,x_2)=(1,0)$. The analysis for the continuous covariate case focuses on the average effect based on \eqref{marg2}. To focus solely on the difference between the bounds and the identified set, we set $\beta=\beta_0$ (a simulation analysis for obtaining bounds when $\beta$ is estimated will be provided in Section \ref{sect:simulations}).

Results are presented in Figure \ref{fig:idset_logit}, where 
  the reported outer bounds are the averages of the estimated bounds across 1000 replications of panels with $n=1000$.\footnote{
   Since we are averaging over a large
   number of replications, the bounds reported
   in Figure \ref{fig:idset_logit} are essentially
   equal to the population outer bounds 
   $ \mathbb{E}[L(Z_i,Y_i,\beta_0)]$
   and $ \mathbb{E}[U(Z_i,Y_i,\beta_0)]$, which
   justifies the comparison to the identified set.  The same comment applies to the comparisons made in Figure \ref{fig:idset_rc}, as well as Figures \ref{fig:idset_dlogit} and \ref{fig:idset_rcd} in the Supplementary Appendix.
  }
 The identified set and   the outer bounds  are obtained for $\beta_0 \in [-2,2]$. The support of $A_i$ is approximated by a grid of 100 equidistant points between $-5$ and 5. Several observations are in order.  First, in all cases, 
 the outer   bounds mimic the behavior of the identified set. In particular, both the identified set and our bounds shrink to a point when $\beta_0=0$ but become wider as $|\beta_0|$ increases. At $T=5$ both the bounds and the identified set become almost a point for the majority of $\beta_0$ we consider. Also, both types of bounds yield the correct sign for the average effect. Second, the difference between the identified set and the outer  bounds 
 vanishes almost completely at $T=5$. This is an important result: as mentioned previously, obtaining the identified set in applications with moderate $T$ is practically infeasible due to the large number of conditional probabilities  $P(Y=y|Z=z)$ one has to estimate, even when $Z$ is discrete. Our results show that the method proposed here stands out as a viable and computationally feasible alternative in such cases. 

\begin{figure}[p]
	\centering
	\includegraphics[width=0.8\linewidth]{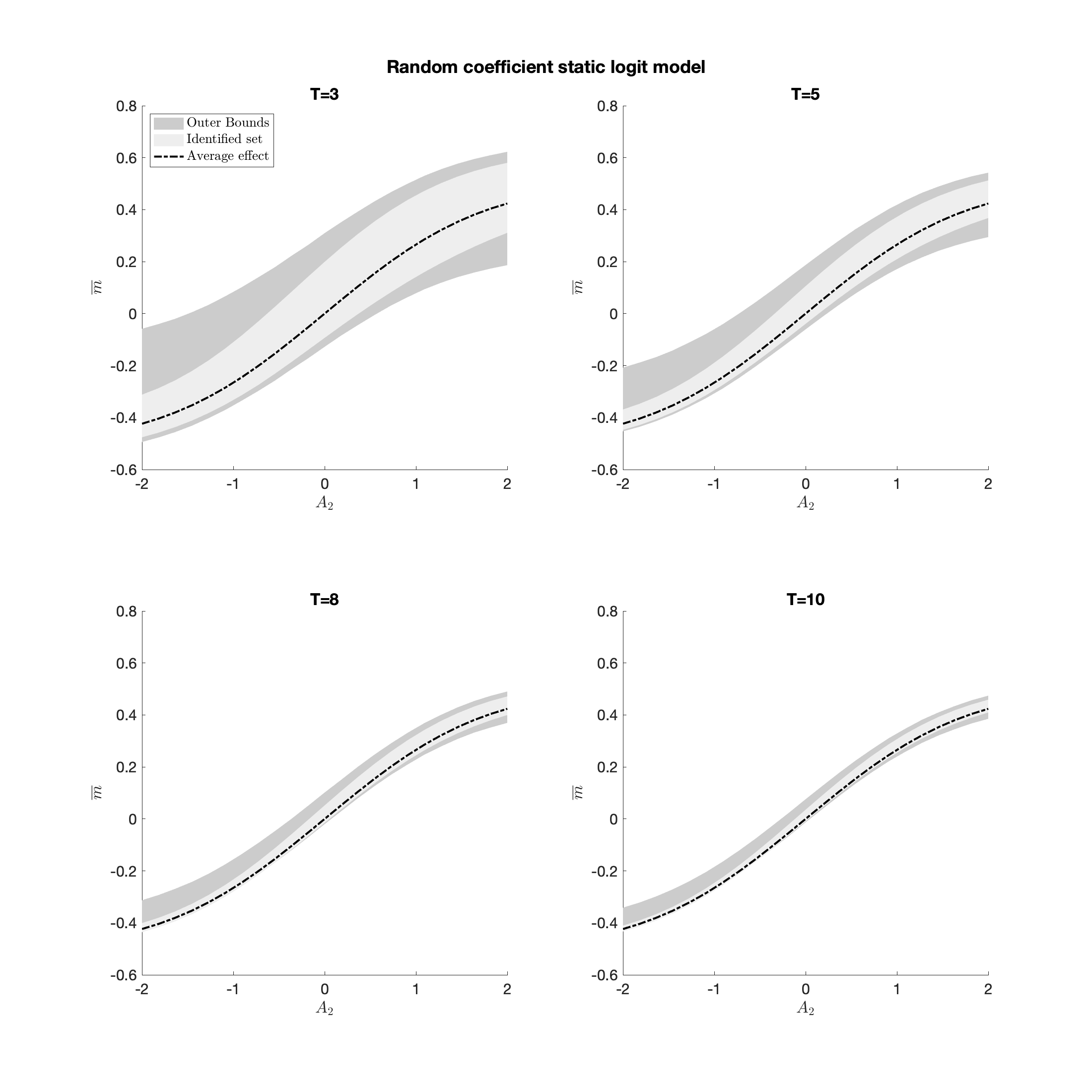}
	\caption{Comparison of the outer bounds  and the identified set for the random coefficient logit model $Y_{it} = 1\left\{ X_{it} A_{2,i} + A_{1,i} \geq \varepsilon_{it} \right\}$, where $\varepsilon_{it} \sim {\rm Logit}(0,1)$, $A_{1,i} \sim N(0,1/2)$, $A_{2,i} \sim N(A_2,1/2)$, 	$X_{it} = 1\left\{ A_{1,i} \geq \eta_{it} \right\}$ and $\eta_{it} \sim N(0,1)$. The average effect of interest is based on \eqref{margrc}. Results for each $A_2\in[-2,2]$ are based on 1000 replications of panels with cross-section size $n=1000$. Reported outer bounds  are the cross-replication averages.}
	\label{fig:idset_rc}
\end{figure}
\afterpage{\clearpage}

The random coefficient example is based on the DGP
\begin{eqnarray}
	& Y_{it} = 1\left\{ X_{it} A_{2,i} + A_{1,i} \geq \varepsilon_{it} \right\},  \quad A_{1,i} \sim N(0,1/2), \quad A_{2,i} \sim N(A_2,1/2) \label{sim:rc1}, \\
	& X_{it} = 1\left\{ A_{1,i} \geq \eta_{it} \right\}, \quad \eta_{it} \sim N(0,1), \label{sim:rc2}
\end{eqnarray}
where $\varepsilon_{it} \sim {\rm Logit } (0,1)$. Our interest is in identifying the average effect based on \eqref{margrc}. We note that Theorem \ref{th:ConsistencyKNOWN} fully applies here, as there are no structural parameters to be estimated. 

Results are based on 1000 replications, and are presented in Figure \ref{fig:idset_rc}. We consider $n=1000$ and $T\in \{ 3,5,8,10 \}$ with $A_2 \in [-2,2]$. To approximate the supports of $A_{1,i}$ and $A_{2,i}$ we use grids of 50 equidistant points between $-5$ and $5$, and between $-7$ and $7$, respectively. Not surprisingly, the presence of a random coefficient renders the average effect more difficult to identify. Indeed, for small $T$ even the sign of the average effect remains inconclusive for values of $A_2$ close to zero. More importantly, although the identified set becomes narrower as $T$ increases, it does not shrink to a point even when $T$ is 8 or 10. For reasons discussed before, obtaining the identified set for such large $T$ will in practice be infeasible. Simulation results confirm that our proposed method provides a reliable alternative. Indeed, the  outer  bounds   are quite close to the identified set at $T=8,10$.

\section{Accounting for estimated common parameters}\label{sect:betahat}

We now consider the case where the common parameter vector $\beta_0$
has to be estimated. Our construction of the bound functions
$L\left(z,y,\beta \right)$ and $U\left(z,y,\beta \right)$
 remains 
essentially unchanged, but they are now evaluated at a consistent estimator of $\beta_0$, rather than the true $\beta_0$.
The goal here is
to provide asymptotic results that account for the noise
in the estimation of $\beta_0$.

If the bound functions
$L\left(z,y,\beta \right)$ and $U\left(z,y,\beta \right)$ were 
differentiable in $\beta$, then 
accounting for the estimation of $\beta_0$ when
providing one-sided confidence intervals on 
the bounds
$\mathbb{E}[L(Z_i,Y_i,\beta_0)]$
and $\mathbb{E}[U(Z_i,Y_i,\beta_0)]$ 
would be a straightforward application 
of the delta method.
Unfortunately, because we obtain 
$L\left(z,y,\beta \right)$ and $U\left(z,y,\beta \right)$
as the solution to a linear  program, it is generally not possible to
verify any smoothness of those 
functions in $\beta$.\footnote{For example, in the static logit
model, we know that the upper and lower bound functions are generally not unique for any given $\beta$, so that the particular solution selected by a linear solver need not be continuous, let alone smooth, as a function of $\beta$. This shows that smoothness of
$f(y|z,a,\beta)$ and $m(z,a,\beta)$ in $\beta$ does not 
imply smoothness of the bounds in $\beta$.
}
The convergence rate and inference results in this section
therefore make no assumption whatsoever on the continuity or smoothness of the
bound functions.\footnote{One could, alternatively, construct $L\left(z,y,\beta \right)$ and $U\left(z,y,\beta \right)$ such that they still satisfy the
assumptions of Theorem~\ref{th:ConsistencyKNOWN}, but are also smooth in 
$\beta$ (e.g.\ in a particular model for a particular average effect
of interest, one may simply find explicit analytic expressions
for the bound functions). We leave the exploration of such possibilities to future work.}

\subsection{Consistency and convergence rate of the estimated bounds}

Before discussing inference on $\overline m$,
our first goal is to show that the population 
bounds
$\mathbb{E}[L(Z_i,Y_i,\beta_0)]$
and 
$\mathbb{E}[U(Z_i,Y_i,\beta_0)]$
can be estimated at $\sqrt{n}$ rate,
even if $\beta_0$ is estimated.
For that purpose,
we split the set of observations $\{1,\ldots,n\}$
into the disjoint subsets ${\cal I}_1 = \{1,\ldots,\lfloor n / 2 \rfloor \}$ and  ${\cal I}_2 = \{ \lfloor n / 2 \rfloor +1,\ldots,n\}$.
For any subset of observed units ${\cal I} \subset \{1,\ldots,n\}$ we denote by 
$Y_{({\cal I})}$ and $Z_{({\cal I})}$ the collection of all observations $Y_i$ and $Z_i$ with $i \in {\cal I}$.
Furthermore, we define the function $\bar s \,: \, \{ 1,\ldots, n\} \rightarrow \{1,2\}$ by
\begin{align*}
    \bar s(i) &:= \left\{ \begin{array}{ll}  
           2 & \text{if $i \in {\cal I}_1$},
           \\
           1 & \text{if $i \in {\cal I}_2$}.
      \end{array} \right.
\end{align*}
For each $s \in \{1,2\}$
we have an estimator $\widehat \beta_s = \widehat \beta_s(Y_{({\cal I}_s)},Z_{({\cal I}_s)})  $  that only depends on the observed
data $(Y_i,Z_i)$ for $i \in {\cal I}_s$.
In other words, $\widehat \beta_{1}$ and $\widehat \beta_{2}$ are estimators of $\beta$ obtained using the first and second half-sample, respectively.
Our estimates for the upper and lower bounds in \eqref{DefineBoundEstimatesKNOWN} then generalize to
\begin{align}
     \widehat L_S &:= \frac 1 n \sum_{i=1}^{n}  
     L\left(Z_i,Y_i, \widehat \beta_{\bar s(i)} \right) ,
     &
     \widehat U_S &:= \frac 1 n \sum_{i=1}^{n}  
     U\left(Z_i,Y_i, \widehat \beta_{\bar s(i)} \right) .
     \label{DefineBoundEstimates}
\end{align} 
Notice that the ``cross-fitting'' construction in \eqref{DefineBoundEstimates} ensures that for any $i$, $(Z_i,Y_i)$ and $\widehat \beta_{\bar s (i)}$ are always from two different half-samples, and therefore independent of each other. Consequently, conditional on the half-sample $\mathcal I_{\bar s(i)}$,
$L(Z_i,Y_i,\widehat \beta_{\bar s (i)})$ and $U(Z_i,Y_i,\widehat \beta_{\bar s (i)})$ are independently distributed over $i$. 
In contrast, if the bound estimators were based on $\widehat \beta$ obtained from the full-sample, $L(Z_i,Y_i,\widehat \beta)$ and $U(Z_i,Y_i,\widehat \beta)$ would be arbitrarily dependent over $i$, ruling out a standard  Law of Large Numbers. Along with reasonable assumptions on the behavior of $\widehat \beta_{\bar s(i)}$, as well as smoothness
conditions on the functions $f(y|z,a;\beta)$ and $m(z,a,\beta)$ in $\beta$, 
the \textit{conditional} independence is sufficient for proving the consistency of the bounds in \eqref{DefineBoundEstimates} for $\mathbb{E}[L(Z_i,Y_i,\beta_0)]$
and $\mathbb{E}[U(Z_i,Y_i,\beta_0)]$.

\begin{assumption}~
     \label{ass:MAIN2}
       
     \begin{enumerate}[(i)]

        \item For  $s \in \{1,2\}$
          the estimator $\widehat \beta_s = \widehat \beta_s(Y_{({\cal I}_s)},Z_{({\cal I}_s)})  $ 
          satisfies
          $    \widehat \beta_s = \beta_0 + O_p\left( n^{-1/2} \right) $, and $\widehat \beta_s \in {\cal B}$.
          
            \item   
            There exists $\epsilon>0$ such that for an $\epsilon$-ball $B_\epsilon(\beta_0)$ around $\beta_0$ we have
           \begin{align*}
                       \sup_{\beta \in {B}_\epsilon(\beta_0)}  \sum_{y\in \mathcal{Y}} \mathbb{E}  \left\|   \frac{\partial    f\left(y\,  |\, Z_i, A_i;  \beta \right) } {\partial \beta}  \right\|   &< \infty,
                       &
                        \sup_{\beta \in {B}_\epsilon(\beta_0)} 
                         \mathbb{E}  \left\|   \frac{\partial   m\left(Z_i,A_i,\beta  \right)  } {\partial \beta}  \right\|   &< \infty .
          \end{align*}
     \end{enumerate}

\end{assumption}

\begin{theorem}
     \label{th:Consistency}
     Let Assumptions~\ref{ass:MAIN} and~\ref{ass:MAIN2} hold,
     and let   $L, U : {\cal Z} \times {\cal Y} \times {\cal B} \rightarrow [b_{\min}, b_{\max}]$ be two non-random functions that satisfy 
     \eqref{eq:boundcondition}
     for all $z \in {\cal Z}$, $a \in {\cal A}$ and $\beta \in {\cal B}$.  Let, finally, $\overline{m}$ be as defined in \eqref{APE}, and  $\widehat L_S$ and   $\widehat U_S$ be as defined in \eqref{DefineBoundEstimates}.
    Then, as $n \rightarrow \infty$, we have
     \begin{align*}
          \widehat L_S +   O_{p}(n^{-1/2} )  \; \leq \;  \overline{m} \; \leq \; \widehat U_S +  O_{p} ( n^{-1/2} )  .
     \end{align*}
\end{theorem}

This theorem generalizes consistency of the outer bounds to the case of estimated $\beta_0$. The proof is straightforward and provided in the appendix.
By contrast,
obtaining inference results under estimated $\beta_0$ is more 
complicated due to the linear program yielding potentially non-smooth bound functions. This non-smoothness is not specific to our case: it arises generically in methods that construct estimating equations or bounds via linear programming. For example, \cite{bonhomme2012functional} constructs moment functions for panel data models by solving linear programs, and faces the same challenge that small changes in parameters can cause discrete jumps in the solution. Similarly, the bounds in \citet{CFHN13} are obtained via linear programming, and their ``perturbed bootstrap'' inference method is specifically designed to circumvent the non-smoothness problem. The bottom line is that one cannot simply deploy the delta method to account for randomness introduced
by the estimation of $\beta_0$, and so a different approach is needed.
In the remainder of this section, we introduce two inference methods.

\subsection{First inference method}
\label{sec:perturbed}

Our first inference method is inspired by the handling 
of common parameters in the ``perturbed bootstrap'' approach of \citet{CFHN13}. The idea is to simply take the union of our
``known $\beta_0$'' confidence intervals in Section~\ref{sect:bounds} over a confidence set of the unknown $\beta_0$. For that purpose, define
\begin{align*}
    \widehat L(\beta) &:= \frac 1 n \sum_{i=1}^{n}  L(Z_i,Y_i, \beta ) ,
    &
    \widehat U(\beta) &:= \frac 1 n \sum_{i=1}^{n}  U(Z_i,Y_i,  \beta ) ,
    \\
    \widehat{\sigma}^2_L(\beta)
    &:= \frac 1 n \sum_{i=1}^n
    \left[ L(Z_i,Y_i, \beta ) - \widehat L(\beta) \right]^2,
    &
    \widehat{\sigma}^2_U(\beta)
    &:= \frac 1 n \sum_{i=1}^n
    \left[ U(Z_i,Y_i, \beta ) - \widehat U(\beta) \right]^2.
\end{align*}

We then have the following theorem.

\begin{theorem}\label{theorem:perturbed}
    Let the assumptions of Theorem~\ref{th:ConsistencyKNOWN} hold, including ${\rm Var}\left[L(Z_i,Y_i, \beta_0 )\right]>0$ and ${\rm Var}\left[U(Z_i,Y_i, \beta_0 )\right]>0$, and let $\alpha \in [0,1]$.
    Let, for some $0<\nu<1$, 
    $\mathcal B_{1-\nu}$ 
    be such that $\lim_{n\to\infty} P (\beta_0 \in \mathcal B_{1-\nu}) \geq 1-\nu$. Then, letting $c_{\alpha/2}=\Phi^{-1}\left(1- \frac \alpha 2 \right)$, we have
    \begin{align*}
        \liminf_{n\rightarrow \infty }P\left\{
        \inf_{\beta \in {\cal B}_{1-\nu}} 
        \left[ \widehat{L}(\beta) - \frac{ c_{\alpha/2} \, \widehat{\sigma }_{L}(\beta)} {\sqrt{n}} \right] \leq \overline m  
        \leq   
        \sup_{\beta \in {\cal B}_{1-\nu}} 
        \left[ \widehat{U}(\beta)+
        \frac{ c_{\alpha/2}\,
        \widehat{\sigma }_{U}(\beta)} {\sqrt{n}}
        \right]
        \right\} \geq 1-\alpha - \nu.
\end{align*}
\end{theorem}

Theorem \ref{theorem:perturbed} provides a straightforward albeit potentially conservative way of obtaining confidence bands  that incorporate the uncertainty due to estimation of $\beta_0$. This uncertainty is captured by $\nu$ whereas $\alpha$ parameterizes the uncertainty due to estimation of the population outer bounds by sample averages. For a  desired level of confidence $1-c$, one can trade off between these two sources of uncertainty by choosing $\alpha$ and $\nu$ as desired. Another option is to find the narrowest confidence interval across all $(\alpha,\nu)$ such that $c=\alpha+\nu$. Notice that the infimum and supremum cannot be calculated exactly, so one has to do a grid search across a sufficiently large selection of $\beta\in \mathcal B _{1-\nu}$. 
Especially when $\beta$ contains several parameters, this method can be demanding. Nevertheless, the attraction of 
Theorem~\ref{theorem:perturbed} is that as long as
a valid confidence interval for $\beta_0$ can be constructed,
 inference on $\overline m$ requires only a straightforward application
of the methods described in Section~\ref{sect:bounds}.
Fortunately, there is a large literature on obtaining  valid
confidence intervals on the common parameters $\beta_0$
in the type of panel data models with fixed effects that
we consider here; see, for example, \cite{arellano2003discrete}
and \cite{ArellanoBonhomme(11)} for  reviews, as well as our
discussion in the introduction.

Interestingly, the confidence set
for $\beta_0$ in 
Theorem~\ref{theorem:perturbed}
can also accommodate cases where $\beta_0$
is not point-identified, as long as a valid
confidence set ${\cal B}_{1-\nu}$ can be constructed. This is particularly relevant for models such as the probit, where the common parameters are generally set-identified when $T$ is fixed \citep[see also][for examples of partial identification in logit models]{DGK24}. In Section~\ref{sec:setidentifiedbeta} of the Supplementary Appendix, we discuss two approaches for obtaining outer bounds when $\beta_0$ is set-identified: the first uses a random coefficients specification, while the second directly incorporates the identified set for $\beta_0$ into bound construction. Section~\ref{sec:appendixC} provides simulation evidence for the 
probit model.

\subsection{Second inference method}
\label{sec:alternativeinf}

As mentioned before, evaluating the infimum
and supremum over $\beta \in {\cal B}_{1-\nu}$
 in Theorem~\ref{theorem:perturbed} can be challenging.
As an alternative inference method, we therefore suggest modifying
the linear program that is used to calculate the upper and 
lower bounds for $\overline m$ such that the uncertainty about
$\beta_0$ is accounted for within the constraints of the 
linear program.

In Section~\ref{sect:bounds}, the crucial requirement on our bound functions $L\left( z, y, \beta \right)$
and $U\left( z, y, \beta \right)$ was that they satisfy
the inequalities in \eqref{eq:boundcondition}
for a fixed value $\beta$. To account for the 
fact that the true $\beta_0$ is unknown, we now 
slightly generalize this idea.
Given a {\it finite} set ${\cal B}_{\rm{sub}} \subset {\cal B}$
of possible values for $\beta$, we demand
that the bound functions 
$L\left( z, y, {\cal B}_{\rm{sub}} \right)$
and $U\left( z, y, {\cal B}_{\rm{sub}} \right)$ 
satisfy the inequalities in \eqref{eq:boundcondition}
for each value $\beta \in {\cal B}_{\rm{sub}}$, that is,
we demand
\begin{align}
       \label{eq:boundconditionSET}
       \forall \beta \in {\cal B}_{\rm{sub}}:\, 
\sum_{y\in \mathcal{Y}}  L\left( z, y, {\cal B}_{\rm{sub}} \right)   f\left(y\, |\, z, a; \beta \right)   \leq \,
  m\left(z,a,\beta \right) 
\,  \leq \,\sum_{y\in \mathcal{Y}}  U\left( z, y, {\cal B}_{\rm{sub}} \right)   f\left(y\, |\, z, a; \beta \right)   .
\end{align}
As in \eqref{eq:boundcondition}, we want 
the inequality in \eqref{eq:boundconditionSET}
to hold for all $z \in {\cal Z}$ and $a \in {\cal A}$.\footnote{For consistency of notation, our previous bounds 
$L\left( z, y, \beta \right)$
and $U\left( z, y, \beta \right)$ could have been written as 
$L\left( z, y, \{\beta\} \right)$
and $U\left( z, y, \{\beta\} \right)$ 
to agree with \eqref{eq:boundconditionSET}, but this is a minor mismatch of notation.
 }

Next, for each half-sample $s \in \{1,2\}$,
let $\widehat {\cal B}_{s}$ be a set of points
estimated only from observations $i \in {\cal I}_s$, such that
the convex hull of $\widehat {\cal B}_{s}$, ${\rm Conv}(\widehat {\cal B}_{s})$,
provides a $1-\nu/2$ confidence set for $\beta_0$. 
For example, for a one-dimensional 
parameter $\beta$, we  choose $\widehat {\cal B}_{s}=\{\widehat \beta_{{\rm low},s},\widehat \beta_{{\rm up},s}\}$
to consist of the lower and upper bounds of a confidence interval for $\beta_0$. Then, ${\rm Conv}(\widehat {\cal B}_{s})=
[\widehat \beta_{{\rm low},s},\widehat \beta_{{\rm up},s}]$
is just a standard confidence interval in that
case. More generally, we have to find a confidence set that can be generated as a 
convex hull of a finite number of points.\footnote{
If $\beta$ is higher-dimensional, then one simple choice
for $\widehat {\cal B}_{s}$ would be the Cartesian product
of one-dimensional confidence bounds for each component
of $\beta$, using a Bonferroni correction to maintain
the correct confidence level $1-\nu/2$. However, this yields a set with $2^{\dim(\beta)}$ vertices, which may be computationally costly.
A more efficient construction uses the cross-polytope, which requires only $2 \dim(\beta)$ vertices. Suppose that $\sqrt{n/2}(\widehat\beta_s - \beta_0) \xrightarrow{d} N(0, \Sigma_\beta)$ with consistent estimator $\widehat\Sigma_{\beta,s}$. Let $c_{1-\nu/2}^{(d)}$ be the $(1-\nu/2)$-quantile of $\|Z\|_1 = \sum_{j=1}^d |Z_j|$ where $Z \sim N(0, I_d)$. Then 
$\widehat{\cal B}_{s} = \left\{ \widehat\beta_s + \frac{c_{1-\nu/2}^{(d)}}{\sqrt{n/2}} \widehat\Sigma_{\beta,s}^{1/2} v \; : \; v \in \{\pm e_1, \ldots, \pm e_d\} \right\},$
where $e_j$ denotes the $j$-th standard basis vector in $\mathbb{R}^d$. 
The set $\widehat{\cal B}_{s}$ consists of $2d$ vertices whose convex hull provides an asymptotically valid $(1-\nu/2)$ confidence set for $\beta_0$.
The critical value $c_{1-\nu/2}^{(d)}$ can be computed numerically; for example, $c_{0.975}^{(1)} \approx 2.24$, $c_{0.975}^{(2)} \approx 3.53$, and $c_{0.975}^{(3)} \approx 4.70$.
}

Let also ${\rm diam}\left( {\cal B}_{\rm{sub}} \right)$ be the diameter of the set
${\cal B}_{\rm{sub}} $.
Finally, we define
\begin{align}
     \widehat L_C  &:= \frac 1 n \sum_{i=1}^{n}  
     L\left(Z_i,Y_i,  \widehat {\cal B}_{\bar s(i)} \right) ,
     &
     \widehat U_C  &:= \frac 1 n \sum_{i=1}^{n}  
     U\left(Z_i,Y_i,  \widehat {\cal B}_{\bar s(i)} \right) .
     \label{DefineBoundEstimatesBf}
\end{align} 

We require the following additional assumptions for this
inference method, which strengthen Assumption~\ref{ass:MAIN2}(ii) and also formalize the requirement that  ${\rm Conv}(\widehat{ \mathcal B} _s)$  is a confidence set.

\begin{assumption}~ \label{ass:Bf}

    \begin{enumerate}

        \item[(i)]\label{ass:Bf1}
        $\mathbb E \{ [
            {\rm diam}( {\widehat{\mathcal B}}_{s} )
        ]^2 \}
        =
        o(n^{-1/2})
        $ and $\lim_{n\rightarrow \infty }P\{ \beta_0  \in {\rm Conv}(\widehat{\mathcal B}_{ s})\}  \geq 1- \nu/2$ where $s\in \{1,2\}$.
        
        \item[(ii)] There exists $\epsilon>0$ such that for an $\epsilon$-ball $B_\epsilon(\beta_0)$ around $\beta_0$ we have
        \begin{align*}
            \sup_{\beta \in B_\epsilon(\beta_0)}  \sum_{y\in \mathcal{Y}} \mathbb{E}  \left\|   \frac{\partial^2    f\left(y\,  |\, Z_i, A_i;  \beta \right) } {\partial \beta^2}  \right\|   &< \infty,
            &
            \sup_{\beta \in B_\epsilon(\beta_0)} 
            \mathbb{E}  \left\|   \frac{\partial^2   m\left(Z_i,A_i,\beta  \right)  } {\partial \beta^2}  \right\|   &< \infty .
          \end{align*}
    \end{enumerate}
    
\end{assumption}

Then, the following lemma shows that conditional
on $ \beta_0  \in  {\rm Conv}(\widehat{\mathcal B}_{\bar s})$,
$ L(Z_i,Y_i,\widehat{\mathcal B}_{\bar{s}})$
and  $U(Z_i,Y_i,\widehat{\mathcal B}_{\bar{s}})$
provide
valid bounds on $\overline m$ in expectation.

\begin{lemma}\label{lem:inflemma}
Let Assumptions \ref{ass:MAIN} and \ref{ass:Bf} hold.
Let
    $L(\cdot,\cdot,\cdot)$ and $U(\cdot,\cdot,\cdot)$
    satisfy \eqref{eq:boundconditionSET}
    for all $z \in {\cal Z}$ and $a \in {\cal A}$,
    and, for $s \in \{1,2\}$ and $\bar s = 3-s$, let  $\widehat{\mathcal B}_{\bar s}$ be as defined after 
    display \eqref{eq:boundconditionSET}.
    Let ${\mathcal B}_{\rm sub} \subset
    {\cal B} \cap B_\epsilon(\beta_0)$ be such that
    $ \beta_0  \in {\rm Conv}({\mathcal B}_{\rm sub})$,
    where $\epsilon>0$ is as in Assumption~\ref{ass:Bf}(ii).
    Then, for sufficiently large $n$, we have
    \begin{align*}
    &  \mathbb E 
    \left[
    \left.
    L(Z_i,Y_i,\widehat{\mathcal B}_{\bar{s}})
    \, \right | \,  \widehat{\mathcal B}_{\bar s}
    ={\mathcal B}_{\rm sub}
    \right]
    + 
    O\left([{\rm {diam}}({\mathcal B}_{\rm sub})]^2\right)
    \leq 
    \overline m
    \\ & \qquad \qquad\qquad\qquad\qquad\qquad
    \leq 
    \mathbb E 
    \left[
    \left.
    U(Z_i,Y_i,\widehat{\mathcal B}_{\bar{s}})
    \, \right | \,  \widehat{\mathcal B}_{\bar s}
    ={\mathcal B}_{\rm sub}
    \right]
    + 
     O\left([{\rm {diam}}({\mathcal B}_{\rm sub})]^2\right).
\end{align*}
\end{lemma}

Once Lemma~\ref{lem:inflemma} is obtained, then all that is
left to do is to account for the sampling uncertainty when
replacing the expected value over
$ L(Z_i,Y_i,\widehat{\mathcal B}_{\bar{s}})$
and
$ U(Z_i,Y_i,\widehat{\mathcal B}_{\bar{s}})$
by the sample averages in \eqref{DefineBoundEstimatesBf},
analogously to Theorem~\ref{th:ConsistencyKNOWN}.

\begin{theorem}
    \label{th:ConsistencyBf}
    Let Assumptions \ref{ass:MAIN}    and \ref{ass:Bf} hold.
    For $s\in\{1,2\}$ and $\bar s = 3-s$, 
    let
    $L(\cdot,\cdot,\cdot)$ and $U(\cdot,\cdot,\cdot)$
    satisfy \eqref{eq:boundconditionSET}
    for all $z \in {\cal Z}$ and $a\in \mathcal A$, and be such that, for all $z \in {\cal Z}$ and $y \in \mathcal Y$, $b_{\rm{min}} \leq L(z,y,\widehat{\mathcal B}_{\bar s}) \leq b_{\rm{max}}$ and $b_{\rm{min}}\leq U(z,y,\widehat{\mathcal B}_{\bar s}) \leq b_{\rm{max}}$.
    Assume further that there exists a constant $c>0$ such that, with probability approaching one,
    ${\rm Var}\left[ \left. L(Z_i,Y_i, \widehat {\cal B}_{\bar s} ) \, \right| \, \widehat {\cal B}_{\bar s} \right] \geq c$
    and
    ${\rm Var}\left[ \left. U(Z_i,Y_i, \widehat {\cal B}_{\bar s} ) \, \right| \, \widehat {\cal B}_{\bar s} \right] \geq c$,
    for $i \in {\cal I}_s$.
    Let  $\overline{m}$, $\widehat L_C$, $\widehat U_C$
    be as defined in \eqref{APE}
    and \eqref{DefineBoundEstimatesBf},
    let $\widehat \sigma_{L,s}$ and   $\widehat \sigma_{U,s}$ be the
    sample standard deviations over $i \in {\cal I}_s$ of $L(Z_i,Y_i, \widehat {\cal B}_{\bar s} )$
    and $U(Z_i,Y_i, \widehat {\cal B}_{\bar s} )$, respectively.\footnote{
    Formally,
    letting $\widehat{L}_{C,s} = (2/n) \sum_{i\in \mathcal I_s}
    L(Z_i,Y_i, \widehat {\cal B}_{\bar s(i)} ) $
    and 
    $\widehat{U}_{C,s} = (2/n) \sum_{i\in \mathcal I_s}
    U(Z_i,Y_i, \widehat {\cal B}_{\bar s(i)} ) $
    we have
    $\widehat{\sigma}^2_{L,s}
    := (2/n)  \sum_{i\in \mathcal I_s}
    [ L(Z_i,Y_i, \widehat {\cal B}_{\bar s(i)} ) - \widehat L_{C,s} ]^2$
    and
    $\widehat{\sigma}^2_{U,s}
    := (2/n)  \sum_{i\in \mathcal I_s}
    [ U(Z_i,Y_i, \widehat {\cal B}_{\bar s(i)} ) - \widehat U_{C,s} ]^2$.
    }        
Let $\alpha \in [0,1]$ and $ \displaystyle c_{\alpha/4}=\Phi^{-1}\left(1- \frac \alpha 4 \right)$.
Then, as $n \rightarrow \infty$ we have:
\begin{align*}
    \liminf_{n\rightarrow \infty }P\left(
        \widehat L_C- \frac{ c_{\alpha/4} \,
        (\widehat{\sigma }_{L,1} + \widehat{\sigma }_{L,2})/2} {\sqrt{n/2}}
        \leq \overline m \leq  \widehat U_C+
        \frac{ c_{\alpha/4}\,
          (\widehat{\sigma }_{U,1} + \widehat{\sigma }_{U,2})/2} {\sqrt{n/2}}
         \right) \geq 1 - \alpha - \nu  .
\end{align*}
\end{theorem}

Theorem~\ref{th:ConsistencyKNOWN}
 demands  that  equation \eqref{eq:boundcondition}  holds, but
does not specify any explicit construction
of the bound functions. Analogously, Theorem~\ref{th:ConsistencyBf} requires that equation 
 \eqref{eq:boundconditionSET} hold, but again does not 
 specify any explicit construction of the bounds.
In order to actually construct the bounds we use the methods
described earlier, but we replace the constraint
 \eqref{eq:boundcondition}
by  \eqref{eq:boundconditionSET}. Specifically, 
the program in display \eqref{eq:GeneralOptimization}
then gets modified as follows:
For any given $z \in {\cal Z}$ and any finite set ${\cal B}_{\rm{sub}} \subset {\cal B}$
with $\overline \beta = \left| {\cal B}_{\rm{sub}} \right|^{-1} \sum_{\beta \in {\cal B}_{\rm{sub}}} \beta$
we can choose $L(z,y,{\cal B}_{\rm{sub}})=\ell(y)$ and $U(z,y,{\cal B}_{\rm{sub}})=u(y)$
as solutions to the following optimization problem:\footnote{Here, the choice of evaluating the objective at $\overline\beta$ is for convenience. Alternative choices, such as $\sum_{\beta \in {\cal B}_{\rm{sub}}} Q(\ell, u, z, \beta)$, are equally valid since the constraints (not the objective function) guarantee validity of the bounds.}
\begin{align}
& 
\min_{\ell,u \, : \,  \mathcal{Y} \rightarrow \mathbb{R} }\,   Q(\ell(\cdot),u(\cdot),z,\overline \beta)
\notag \\
 & \text{subject to }   \label{LinearProgramBf}   \\
&  \qquad  \forall y\in \mathcal{Y}:\, \,b_{\min }\leq \ell(y)\leq 
u(y)\leq b_{\max }  \notag \\
& \ \text{and} \; \;
\forall \beta \in {\cal B}_{\rm{sub}}:\, \, 
\forall a\in \mathcal{A}:\, \, \sum_{y\in \mathcal{Y}}\ell (y) \, f(y \, | \, z, a;\beta )
 \leq m(z,a,\beta )\leq \sum_{y\in \mathcal{Y}}u(y) f(y \, | \, z, a;\beta ).    \notag
\end{align}
By choosing the objective function
$Q(\ell(\cdot),u(\cdot),z,\overline \beta)$ as in
\eqref{eq:linprog1} or \eqref{eq:linprog2}, we again have to solve a linear 
program to obtain the bounds.

\section{Simulation evidence}\label{sect:simulations}

In this part we investigate the small sample behavior of the proposed bounds and confidence bands. We focus on the static logit and random coefficient logit models. The setting largely follows Section \ref{sect:setcomparison}. In particular, we use the DGPs in \eqref{eq:idlogitdisc} and \eqref{sim:rc1}-\eqref{sim:rc2} with a single discrete covariate, and focus on the same average effects. The main difference is that we now estimate $\beta_0$ in the static logit model, and also provide confidence bands. The results, presented in Figures \ref{fig:sim_logit_pb}-\ref{fig:sim_randc}, provide the population average effect, and the cross-replication averages of estimated bounds and 95\% confidence bands.

\begin{figure}[p]
	\includegraphics[width=0.9\linewidth]{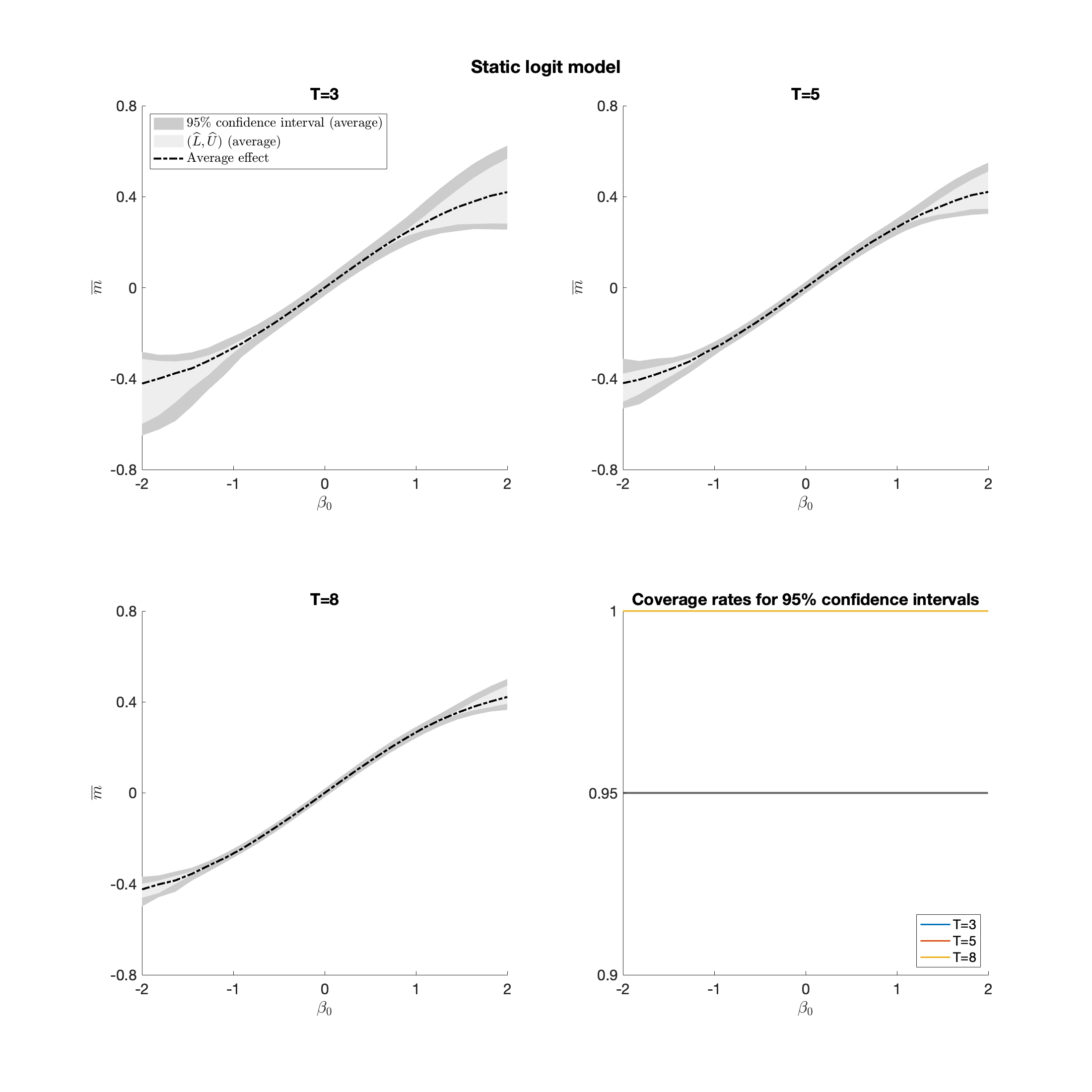}
	\caption{Simulation results for the static logit model with a single discrete covariate: $Y_{it} = 1\left\{ X_{it}\beta + A_i \geq \varepsilon_{it} \right\}$ where $\varepsilon_{it} \sim {\rm Logit} (0,1)$, $A_i \sim N(0,1)$, $X_{it} = 1\left\{ A_i \geq \eta_{it} \right\}$ and $\eta_{it} \sim N(0,1)$. Average effects are based on \eqref{marg1} with $(x_1,x_2)=(1,0)$. Results for each $\beta_0 \in [-2,2]$ are based on 1000 replications of panels with cross-section size $n=5000$. For each replication, $\widehat{L}$ and $\widehat{U}$ are obtained by the linear program of Section \ref{sec:uniflp}, using the conditional likelihood estimator $\widehat \beta$ of $\beta_0$. Confidence intervals are based on the inference method of Section \ref{sec:perturbed}, using $\nu=0.0001$. $\mathcal B_{1-\nu}$ is approximated by a grid of 5000 equidistant points. Reported confidence intervals and $(\widehat{L},\widehat{U})$ are cross-replication averages. The lower right panel presents the coverage rates.}
	\label{fig:sim_logit_pb}
\end{figure}
\afterpage{\clearpage}

We first consider the static logit model. $\beta_0$ is estimated using the conditional likelihood method. For inference we use the two inference methods proposed in Sections \ref{sec:perturbed} and 
\ref{sec:alternativeinf}. In either case, we consider 1000 replications of panels of size $n=5000$ and $T\in \{ 3,5,8 \}$. For $\mathcal{A}_g$, we use a grid of 100 equidistant points between $-5$ and $5$. 

The results using the inference method of Section 
\ref{sec:perturbed} are based on $\nu = 0.0001$, and $\mathcal B_{1-\nu}$ is approximated by a grid of 5000 equidistant points on $\mathcal B_{1-\nu}$. Outer bounds for this case are obtained by the uniform linear program of Section \ref{sec:uniflp}. Results are presented in Figure \ref{fig:sim_logit_pb}.
For moderate $T$, which is the main focus of this study, both the bounds and the confidence bands are quite tight. Interestingly, this is despite the fact that the bounds are based on a uniform linear program. In all cases, the confidence bands yield the correct sign for the average effect. The coverage rates of confidence bands are, not surprisingly, conservative. This is expected in partially identified settings and is acceptable given that the bands remain 
informative.

\begin{figure}[p]
	\includegraphics[width=0.9\linewidth]{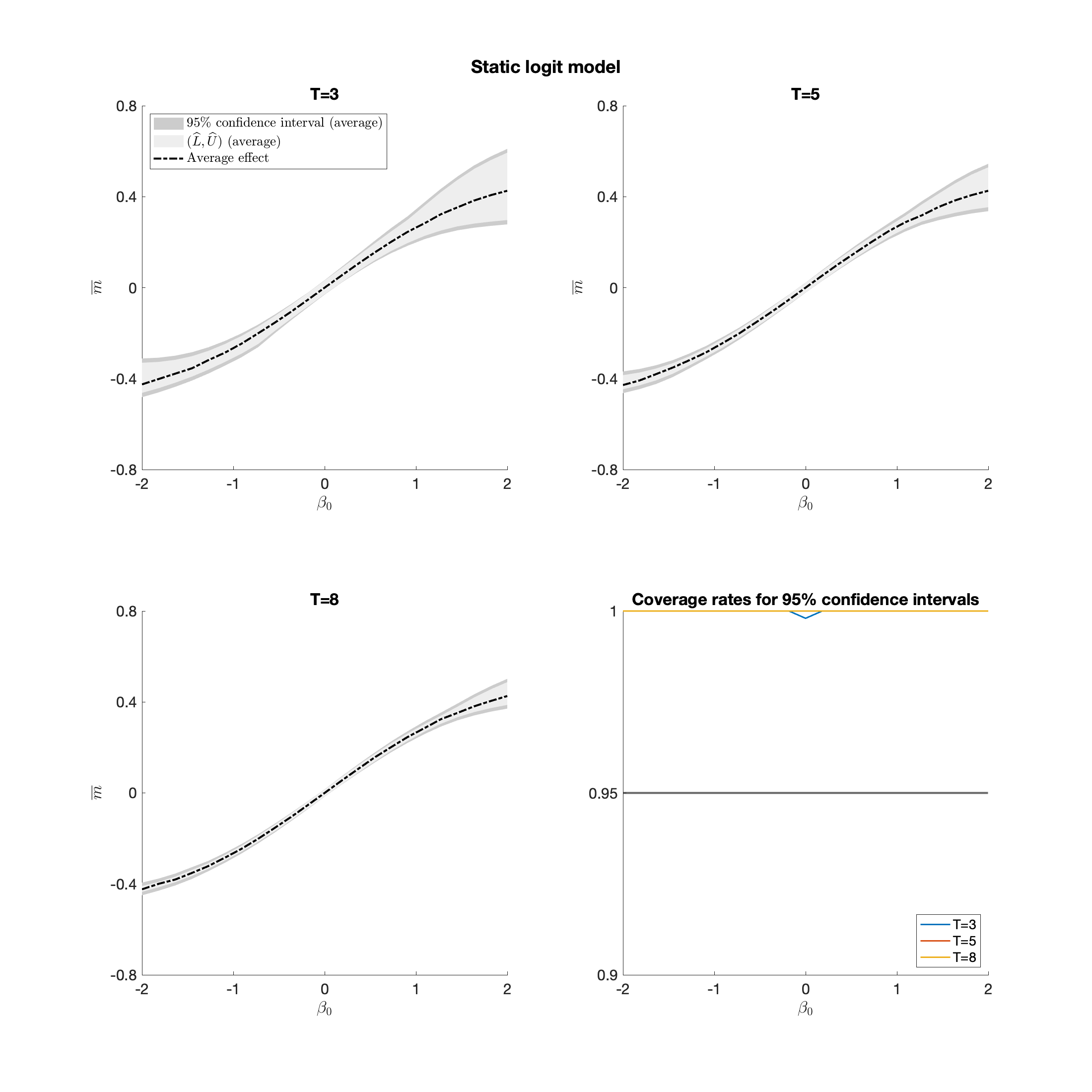}
	\caption{Simulation results for the static logit model with a single discrete covariate: $Y_{it} = 1\left\{ X_{it}\beta + A_i \geq \varepsilon_{it} \right\}$ where $\varepsilon_{it} \sim {\rm Logit} (0,1)$, $A_i \sim N(0,1)$, $X_{it} = 1\left\{ A_i \geq \eta_{it} \right\}$ and $\eta_{it} \sim N(0,1)$. Average effects are based on \eqref{marg1} with $(x_1,x_2)=(1,0)$. Results for each $\beta_0 \in [-2,2]$ are based on 1000 replications of panels with cross-section size $n=5000$. For each replication, outer bounds and confidence bands are obtained by the methods outlined in 
    Section~\ref{sec:alternativeinf}, using the conditional likelihood estimator $\widehat \beta$ of $\beta_0$. Confidence intervals are based on $\alpha = \frac{2}{3} \times 0.05$ and $\nu = \frac{1}{3} \times 0.05$. Reported confidence intervals and $(\widehat{L},\widehat{U})$ are cross-replication averages. The lower right panel presents the coverage rates.}
	\label{fig:sim_logit_nb}
\end{figure}	
\afterpage{\clearpage}

We next consider the  inference approach of Section \ref{sec:alternativeinf}, the results of which are presented in Figure~\ref{fig:sim_logit_nb}. Confidence 
bands are based on $\alpha = \frac{2}{3} \times 0.05$ and $\nu = \frac{1}
{3} \times 0.05$.\footnote{The choice of $\alpha = 2\nu $ is not crucial and was only 
imposed to compensate for the fact that the confidence interval ${\rm Conv}(\widehat 
{\mathcal B} _s)$ is subject to one Bonferroni split, whereas the interval for estimated 
outer bounds is subject to two.}
Both the confidence bands and the outer bounds are based on the linear program defined in \eqref{LinearProgramBf}.
Relative to the inference method of Section \ref{sec:perturbed}, there are two differences: first, the confidence bands are overall visibly closer to the estimated bounds, across all $T$. This is not surprising given that the inference method of Section \ref{sec:perturbed} is based on the infimum/supremum bands. Second, while the outer bounds improve with $T$, they are not as tight as the bounds produced by the linear program in \eqref{eq:GeneralOptimization}. This likely results from \eqref{LinearProgramBf} incorporating the uncertainty due to $\widehat{\beta}$ in outer bound estimation (as opposed to \eqref{eq:GeneralOptimization} which incorporates the same in the inference stage).

\begin{figure}[p]
	\includegraphics[width=0.9\linewidth]{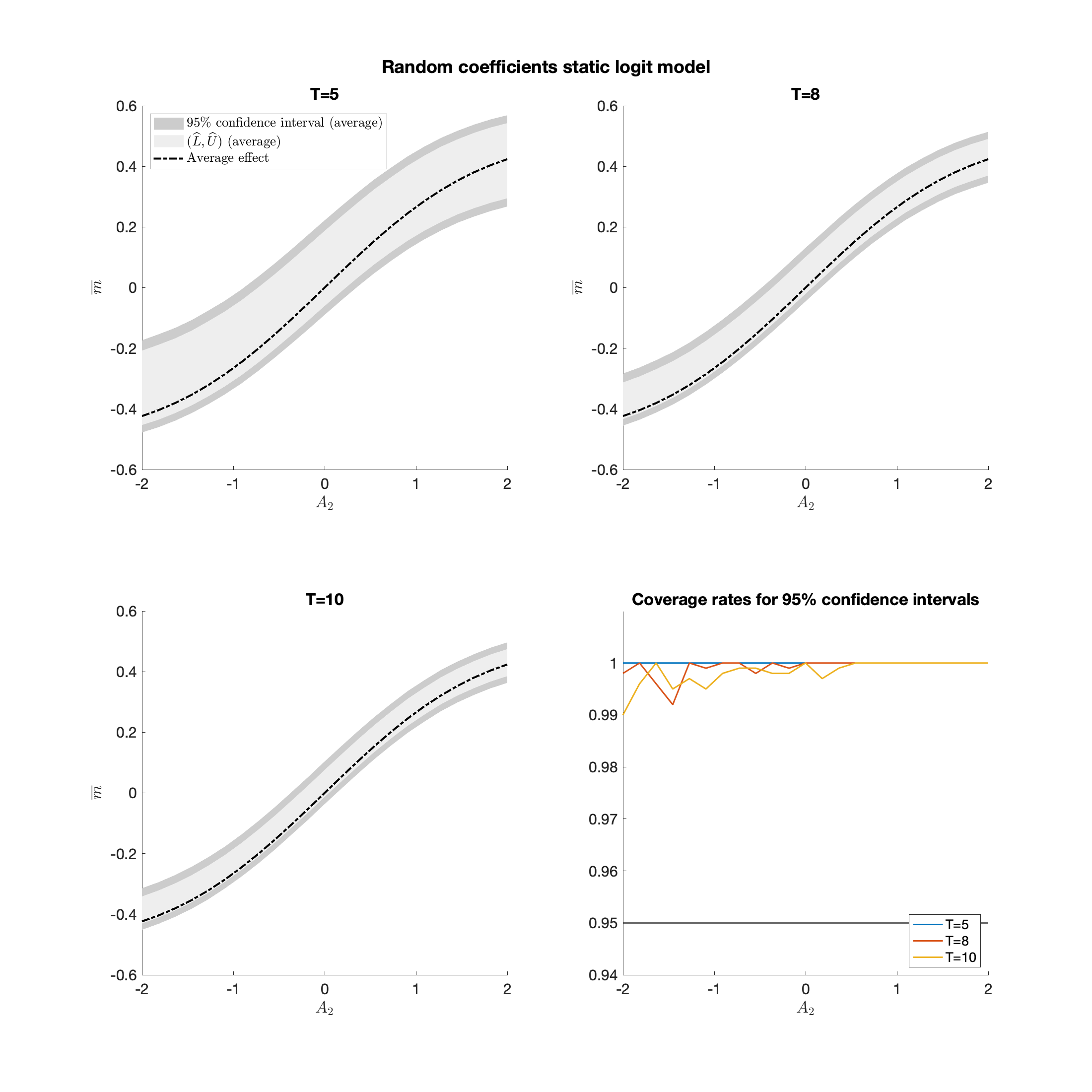}
	\caption{Simulation results for the random coefficient logit model with a single discrete covariate: $Y_{it} = 1\left\{ X_{it} A_{2,i} + A_{1,i} \geq \varepsilon_{it} \right\}$, where $\varepsilon_{it} \sim {\rm Logit}(0,1)$, $A_{1,i} \sim N(0,1/2)$, $A_{2,i} \sim N(A_2,1/2)$, 	$X_{it} = 1\left\{ A_{1,i} \geq \eta_{it} \right\}$ and $\eta_{it} \sim N(0,1)$. Average effects are based on \eqref{margrc}. Results for each $A_2 \in [-2,2]$ are based on 1000 replications of panels with cross-section size $n=1000$. For each replication, $\widehat{L}$ and $\widehat{U}$ are obtained by the linear program in \eqref{eq:UnifOptimization}. Confidence intervals are based on Theorem \ref{th:ConsistencyKNOWN}. Reported confidence intervals and $(\widehat{L},\widehat{U})$ are cross-replication averages. The lower right panel presents the coverage rates.}
	\label{fig:sim_randc}
\end{figure}	
\afterpage{\clearpage}

We move to the random coefficient static logit example. Figure \ref{fig:sim_randc} presents results based on 1000 replications of panels of size $n=1000$ and $T\in \{ 5,8,10  \}$. We construct $\mathcal{A}_g$ using 50 equidistant grid points between $-5$ and $5$ for $A_{1,i}$, and between $-7$ and $7$ for $A_{2,i}$, leading to 2,500 grid points in total. We note that the average effects and outer bounds are the same as in Section \ref{sect:setcomparison}, since no parameter estimation is involved in this setting. The new result is the confidence bands, which are based on Theorem \ref{th:ConsistencyKNOWN}. On average the confidence bands are reasonably close to the outer bounds.

A general observation across the three simulation exercises is that while the confidence bands are generally informative, the coverage rates are  conservative.  
Of course, since we are looking at coverage rates for the
true value of the average effects (as opposed to the identified set), it is well-known that coverage will generally be conservative in partially identified settings (e.g., see \citealt{ImbensManski04} and \citealt{Stoye21}).\footnote{{\citet{ImbensManski04} provide a generic method for obtaining narrower confidence bands on sets which can also be used here. In results not reported here we nevertheless observed that the decrease in the width of confidence bands was quite moderate, with little or no change in the actual coverage rates.}}

Finally, we note that the total width of the confidence bands can be decomposed into two components: (i) the gap between the outer bounds and the identified set, which measures the cost of using outer bounds rather than sharp identification, and (ii) the gap between the confidence bands and the outer bounds, which measures the cost of estimation uncertainty. This decomposition helps assess how much of the total uncertainty is due to our methodological choice versus statistical imprecision. We report these results in Section~\ref{sec:addsimdetails} of the Supplementary Appendix.

\section{Empirical analysis}\label{sect:empiricalanalysis}

 We consider an empirical analysis of female labor force participation, using the National Longitudinal Survey of Youth (NLSY) 1979 dataset.
 Our sample consists of data on women who were married throughout the sample and who were not in active forces or going to school.\footnote{An individual is classified as ``in the labor force'' if her status was recorded as \textit{working}, \textit{with job not at work} or \textit{unemployed}. Individuals are considered as not in the labor force if their recorded status was \textit{keeping house}, \textit{unable to work} or \textit{other}.}  Also, we only include individuals who were observed at all periods under consideration. 

First, we consider a random coefficient logit specification\footnote{Labor force participation is often modeled with state dependence \citep[e.g.,][]{hyslop1999state}. We focus on static models here for comparability with existing methods; see Appendix~\ref{app:dynamic} for simulation evidence on dynamic specifications.}:
\begin{equation}
	LFP_{it} = 1\left\{A_i + \beta_i \, kids3_{it} \geq \varepsilon_{it} \right\}, \label{emp:rc}
\end{equation}
where, for individual $i$ and at time $t$, $LFP_{it}$ is the labor force participation indicator whereas $kids3_{it}$ is a binary variable which equals one if the individual has at least one child below the age of three. This is almost identical to the example considered by \citet{CFHN13}, except that they assume a homogeneous coefficient $\beta$ for all individuals. Our objective is to obtain a confidence interval on the average effect 
\begin{equation*}
	\mathbb{E} \left[ 
	P(LFP_{it}=1|kids3_{it}=1, A_i,\beta_i) 
	- 
	P(LFP_{it}=1|kids3_{it}=0, A_i,\beta_i)
	\right].
\end{equation*}
Our sample period for this analysis covers all even years from 1986 to 1998, which yields data on 929 individuals over seven years. For comparison, we also report the average effects based on the fixed effects logit (FE logit) and probit (FE probit) models, as well as the linear fixed effects model. We note that all these alternatives impose homogeneity of $\beta_i$, and calculate the average effects using estimated $(A_i,\beta)$. Hence, they provide a point-estimate for the average effect. We also use (i) the bias-corrected logit (BC logit) and probit (BC probit) methods, which analytically correct $\widehat{\beta}$ for the incidental parameter bias, following \citet{cruz2017bias}, and (ii) the split-panel jackknife method (SPJ probit and SPJ logit) of \citet{DhaeneJochmans2015}, which directly corrects average marginal effects rather than $\widehat{\beta}$. We note that none of these alternative methods are designed for short-$T$ samples where average effects are not necessarily point-identified.
For all methods under consideration, we provide the 95\% confidence intervals. For the outer bounds this is obtained by using the normal approximation of Theorem \ref{th:ConsistencyKNOWN}.

\afterpage{
\begin{landscape}

\begin{table}[tbp] 
\setlength\tabcolsep{2pt}
\begin{tabular*}{\textwidth}{@{\extracolsep{4pt}}*{9}{c}} 
\multicolumn{9}{c}{$ LFP_{it} = 1\left\{A_i + \beta_i \, \, kids3_{it} \geq \varepsilon_{it} \right\} $, \quad $n=929$, \quad $T=7$} 
\\
\cline{1-9}
\\
& heterogeneous $\beta_i$ & \multicolumn{7}{c}{$\beta_i = \beta$ }
\\
\cline{2-2}  \cline{3-9}
\\
& $(\widehat{L},\widehat{U})$ & FE logit & BC logit & SPJ logit & FE probit & BC probit & SPJ probit & Linear model 
\\ 
$kids3$ 
& $\underset {[-.272 \, , \, -.044]} {-.245 \, ; \, -.067}$ 
& $\underset {[ -.142 \, , \,  -.102]} { -.122}$   
& $\underset {[ -.153 \, , \,  -.113]} {-.133}$ 
& $\underset {[-.179 \, , \, -.139]} {-.159}$ 
& $\underset {[ -.142 \, , \,  -.102]} {-.122}$ 
& $\underset {[-.151 \, , \, -.112]} {-.132}$ 
& $\underset {[-.178 \, , \, -.138]} {-.158}$ 
& $\underset {[-.144 \, , \, -.104]} {-.124}$  
\\ 
\\ 
\\ 
\multicolumn{9}{c}{$ LFP_{it} = 1\left\{A_i + \beta \, \, kids3_{it} + \gamma \, \, educ_{it} + \delta \, \, \ln(spouseinc_{it}) \geq \varepsilon_{it} \right\} $, \quad $n=993$, \quad $T=5$} 
\\ 
\cline{1-9}
\\ 
& $(\widehat{L},\widehat{U})$ & FE logit & BC logit & SPJ logit & FE probit & BC probit & SPJ probit & Linear model  
\\ 
\\ 
$kids3$ 
& $\underset {[-.169 \, , \, -.045]} {-.101 \, ; \, -.098}$     
& $\underset {[ -.118 \, , \, -.076]} {-.097}$  
& $\underset {[ -.141 \, , \, -.099]} { -.120 }$   
& $\underset {[ -.145 \, , \, -.103]} {-.124}$ 
& $\underset {[ -.116 \, , \,   -.076]} {-.096}$ 
& $\underset {[-.137 \, , \,  -.096]} { -.116}$ 
& $\underset {[-.143 \, , \,  -.102]} {-.123}$ 
& $\underset {[ -.117 \, , \,  -.072]} {-.095}$ 
\\ 
\\ 
$educ$ 
& $\underset {[-.036 \, , \, .078]} {.018 \, ; \, .020}$         
& $\underset {[-.024 \, , \, .060]} {.018}$   
& $\underset {[ -.020 \, , \, .063]} {.022}$      
& $\underset {[ -.024 \, , \, .059]} {.017}$ 
& $\underset {[ -.023 \, , \, .056]} { .016}$ 
& $\underset {[-.020 \, , \, .059]} {.020}$ 
& $\underset {[-.026 \, , \, .054]} {.014}$ 
& $\underset {[-.017 \, , \, .039]} { .011}$ 
\\ 
\\ 
$\ln(spouseinc)$ 
& $\underset {[-.112 \, , \, .077]} {-.101 \, ; \,-.040}$ 
& $\underset {[-.165 \, , \, .051]} {-.057}$  
& $\underset {[-.177 \, , \, .039]} {-.069}$  
& $\underset {[-.187 \, , \, .029]} {-.079}$ 
& $\underset {[-.137 \, , \, .042]} { -.047}$ 
& $\underset {[-.146 \, , \,   .033]} { -.057}$ 
& $\underset {[-.149 \, , \,  .030]} {-.059}$ 
& $\underset {[-.065 \, , \,  -.020]} { -.042}$ 
\end{tabular*}
\caption{Empirical analysis results. For the average effects of interest in each case, see the discussion in Section \ref{sect:empiricalanalysis}. $(\widehat{L},\widehat{U})$ are the outer bounds. FE logit and FE probit are the fixed effects panel logit and panel probit models. BC logit and BC probit are the bias-corrected versions, which analytically correct for the incidental parameter bias in estimating $\beta_0$. SPJ logit and SPJ probit, on the other hand, are the average effects obtained using the split-panel jackknife method. Linear model is the linear panel fixed effects model. Numbers in brackets are the 95\% confidence bands. All methods other than the outer bounds provide point estimates of the average effects. In addition, on the top panel these alternative methods impose homogeneity of $\beta_i$.}
\label{tbl:emp}
\end{table}

\end{landscape}
}

Results for this first illustration are reported in the top panel of Table \ref{tbl:emp}. All methods agree that having at least one child younger than three has a negative impact on labor force participation. This is also in line with the results obtained by \citet{CFHN13} who consider a shorter sample, covered by our dataset (see their Table III). The confidence intervals for the outer bounds are wider than the rest. This, however, is a reflection of the underlying identification problem rather than a deficiency of the method. The alternative methods all treat the average effect as point identified: in the nonlinear models this is through large-$T$ approximations for the individual effects, whereas in the linear model this is through its linearity restriction. Consequently, they produce confidence bands for point estimates, which are naturally narrower. 
The outer bounds, on the other hand, allow for partial identification of average effects. The resulting confidence intervals are for set-estimates, and are naturally wider.
The heterogeneity of $\beta_i$ permitted in this first illustration contributes further to the width.\footnote{Heterogeneity of $\beta_i$ is quite likely, as the effect of having a child younger than three will vary depending on various conditions. For example, families with higher income will have easier (and better) access to child care. Geographical proximity of grandparents (who can, at least from time to time, provide child care) is also likely to have an effect on $\beta_i$. Moreover, the effect of having children younger than three may differ depending on the actual number of children.} The narrower confidence intervals by the alternative methods therefore rely on assumptions that are uncertain in the given setting. Viewed this way, 
the outer bounds provide a reality check against intervals constructed under point identification, which can be misleading or artificially narrow.

In the second illustration, we consider the static logit 
specification
with a richer set of covariates:
\begin{equation}
	LFP_{it} = 1\left\{A_i + \beta \, kids3_{it} + \gamma \, educ_{it} + \delta  \,  \ln(spouseinc_{it}) \geq \varepsilon_{it} \right\}, \label{emp:stat}
\end{equation}
where $educ_{it}$ is the highest completed grade (as of May 1 of the survey year) and $spouseinc_{it}$ is the total income of the spouse from wages and salary in the past calendar year. The sample for this exercise covers all even years from 1990 to 1998. We do not include individuals whose spouse had zero income at any point during this period. Average effects for the covariates $kids3_{it}$ and $educ_{it}$ are based on \eqref{marg1}, where we use $(x_1,x_2)=(1,0)$ and $(x_1,x_2)=(educ_{it}+1,educ_{it})$, respectively. Average effects for log spouse income are calculated using $\eqref{marg2}$ with $x_{k,it}=\ln(spouseinc_{it})$. The outer bounds are obtained using the 
linear program of Section \ref{sec:uniflp} whereas the inference approach of Section 
\ref{sec:alternativeinf} 
is used to generate
the confidence bands.\footnote{We first obtain the confidence bands across a selection of $\alpha$ and $\nu$ such that $\alpha+\nu=0.05$, and then report the shortest confidence interval among these.}

Results are reported in the bottom panel of Table \ref{tbl:emp}. All methods agree that the average effect of $kids3$ is negative. For $educ$, all confidence bands are ambiguous about the sign of the effect. However, for all methods these bands are mostly on the positive side. In addition, estimated average effects and outer bounds all point to a positive effect of $educ$ on labor force participation.  We note that although almost all the point estimates for the average effects with respect to $kids3$ and $educ$ are outside the respective outer bounds, they are all covered by the confidence bands for the outer bounds.
Finally, for log of spouse income, confidence bands by all alternatives (other than the linear model) are inconclusive about the sign of the average effect, though they mostly lie on the negative side. Interestingly,  in this particular case the confidence bands for all methods other than the linear probability model lie partially outside the confidence intervals for the outer bounds. This is not necessarily surprising, given that none of the alternative methods considered here are designed to work 
in short samples.

\section{Conclusion}\label{sect:conclusion}

In this paper, we have introduced a new method for estimating bounds on average effects in discrete choice panel data models with fixed effects, including two approaches for obtaining asymptotically valid confidence intervals on the average effects. 
For realistic models and sample sizes, 
inference based on our outer bounds is easier and more robust than inference based on the sharp identified set. A key strength of our approach is its broad applicability: it is suitable for models with both discrete and continuous covariates, and it can be adapted for a variety of static and dynamic panel models.

We have focused here primarily on the case where the common model parameters $\beta_0$
are point-identified and can be estimated at the parametric rate. In the Supplementary Appendix, we show how our approach extends to cases where the structural parameters are set-identified, with simulation evidence for the static probit model.

Another potential extension is to models with continuous outcomes $Y_{it}$, where the sums over $\mathcal{Y}$ in our linear programs would be replaced by integrals. In principle, such integrals could be approximated by sums over a discretized support, similar to how we approximate the constraints over $\mathcal{A}$ by a grid $\mathcal{A}_g$. We leave the exploration of this to future 
work.

\clearpage

\setcounter{page}{1}
\renewcommand{\thepage}{S\arabic{page}}

\onehalfspacing
\begin{appendix}

\begin{center}
    \Large{Supplementary Appendix
    \\
    \vspace{20pt}
    Bounds on Average Effects in Discrete Choice Panel Data Models
    \\
    \vspace{20pt}
    Cavit Pakel
    \&
    Martin Weidner
    \vspace{20pt}
    \\
    Not Intended for Publication}
\end{center}

\section{Mathematical appendix}

\begin{proof}[Proof of Theorem \ref{th:ConsistencyKNOWN}]
     \underline{Part (i):}
     Define $U_{i}=U(Z_i, Y_i, \beta _{0})$ and $\overline{U} = \mathbb{E}  \left( U_{i} \right)$. 
     We have 
     $ {{\rm Var}(U_i)} < \infty $ since $U_i$ is, by design, uniformly bounded. Then, by Assumption \ref{ass:MAIN} and Chebyshev's inequality, for any $\varepsilon >0$ we have
    \begin{align*}
    P\left\{ | \widehat{U}-\overline{U} | \geq \varepsilon \right\}
    &= 
    P\left\{ \left[ \frac 1 n \sum_{i=1}^n \left( U_i-\overline{U}  \right)\right]^2 \geq \varepsilon^2 \right\}
    \\
    &\leq  \frac 1 {n^2 \, \varepsilon^2} \sum_{i=1}^n \sum_{j=1}^n
      \mathbb{E}\left[  \left( U_i-\overline{U} \right) 
      \left( U_j-\overline{U} \right) \right]
     =\frac { {\rm Var}(U_i)} {n \, \varepsilon^2} = O\left( \frac 1 n \right)  .
    \end{align*}
    We therefore have $\widehat{U}-\overline{U}=O_{p}(n^{-1/2})$.
    According to \eqref{eq:popbound} we have
    $\overline m \leq \overline U$, and therefore
    $\overline m \leq \widehat{U} + O_{p}(n^{-1/2}) $.
    By analogous arguments we obtain $\widehat{L} + O_{p}(n^{-1/2}) \leq \overline{m}$.  
    \medskip

    \underline{Part (ii):}    
     Define   $\sigma_{U}^{2}={\rm Var}[U(Z_i,Y_i, \beta_0 )]$.
     By the Weak Law of Large Numbers we have $\widehat{\sigma}_U^2 \to_p \sigma_U^2$.
    Remember that $\sigma_{U}^{2}>0$, by assumption. Then, by the Lindeberg-L\'{e}vy CLT it follows that
    \begin{equation*}
        \frac{1}{\sqrt{n}}\sum_{i=1}^{n}
        \left( U_i - \overline{U} \right)
        \overset{d}{\rightarrow }\mathcal{N}(0,\sigma_{U}^{2}),
    \end{equation*}
    and also using Slutsky's theorem we thus obtain  
   \begin{eqnarray}
         \lim_{n\rightarrow \infty }P\left( \overline{U}\leq \widehat{U}+ \frac{ c_{\alpha/2}\widehat{\sigma }_{U}} {\sqrt{n}}\right) =\Phi     (c_{\alpha/2}).
         \label{LimitUasy}
    \end{eqnarray}
    By analogous arguments, 
    \begin{eqnarray}
         \lim_{n\rightarrow \infty }P\left( \overline{L}\geq  \widehat{L}-\frac{c_{\alpha/2}\widehat{\sigma }_{L} } {\sqrt{n}}\right) =  \Phi     (c_{\alpha/2}).
         \label{LimitLasy}
    \end{eqnarray}
    Next, notice that
    \begin{align}
    P\left(
        \widehat{L}- \frac{ c_{\alpha/2} \, \widehat{\sigma }_{L}} {\sqrt{n}} \leq \overline m \leq   \widehat{U}+
        \frac{ c_{\alpha/2}\,
        \widehat{\sigma }_{U}} {\sqrt{n}}\right)
        &=
        P\left( \overline{m}\geq \widehat{L}-c_{\alpha/2}\frac{\widehat{\sigma }_{L}}{\sqrt{n}} \; \; \cap \; \; \overline{m}\leq \widehat{U}+c_{\alpha/2}\frac{\widehat{\sigma }_{U}}{\sqrt{n}}\right)
        \notag
    \\
        &\geq 
        P\left( \overline{L}\geq \widehat{L}-c_{\alpha/2}\frac{\widehat{\sigma }_{L}}{\sqrt{n}} \; \; \cap \; \; \overline{U}\leq \widehat{U}+c_{\alpha/2}\frac{\widehat{\sigma }_{U}}{\sqrt{n}}\right) 
        \notag
    \\
        &\geq 1-P\left( \overline{L}\leq \widehat{L}-c_{\alpha/2}\frac{\widehat{\sigma }_{L}}{\sqrt{n}} \; \; \cup \; \; \overline{U}\geq \widehat{U}+c_{\alpha/2}\frac{\widehat{\sigma }_{U}}{\sqrt{n}}\right) 
        \notag
	\\
	    &\geq 1-P\left( \overline{L}\leq \widehat{L}-c_{\alpha/2}\frac{\widehat{\sigma }_{L}}{\sqrt{n}}\right) -P\left( \overline{U}\geq \widehat{U}+c_{\alpha/2}\frac{\widehat{\sigma }_{U}}{\sqrt{n}}\right),
    \label{LimitR}
    \end{align}
    where in the first inequality we have used $\overline{L} \leq \overline{m} \leq \overline{U}$.
    Using \eqref{LimitUasy} and \eqref{LimitLasy} in 
    \eqref{LimitR}, and then taking limits, we finally obtain
    \begin{align*}
        \liminf_{n\rightarrow \infty }P\left(
        \widehat{L}- \frac{ c_{\alpha/2} \, \widehat{\sigma }_{L}} {\sqrt{n}} \leq \overline m \leq   \widehat{U}+
        \frac{ c_{\alpha/2}\,
        \widehat{\sigma }_{U}} {\sqrt{n}}\right) 
        &\geq    
        1-\left( 1-\Phi (c_{\alpha/2})\right) -\left( 1-\Phi (c_{\alpha/2})\right) 
    = 1 - \alpha,
    \end{align*}
    as stated.
\end{proof}

\begin{proof}[\bf Proof of Theorem~\ref{th:Consistency}]
    For $s \in \{1,2\}$ let $\bar s = 3-s$ and $n_s = | {\cal I}_s |$,
    which is either
    $\lfloor n / 2 \rfloor$ 
    or $\lceil n / 2 \rceil$.
Define also
    \begin{align*}
        \widehat L_s &= \frac 1 {n_s}  \sum_{i \in {\cal I}_s }  L(Z_i,Y_i, \widehat \beta_{\bar s} ) ,
        &
        \overline L(\beta) &=   \mathbb{E}\left[  \sum_{y\in \mathcal{Y}}  L\left( Z_i, y,  \beta \right)   f\left(y\, |\, Z_i, A_i; \beta \right)  \right].
    \end{align*}  
Note, importantly, that whenever $i \in {\cal I}_s$
    \begin{align*}
        \overline L(\beta) 
        &=
        \mathbb{E}\left[  \left. \sum_{y\in \mathcal{Y}}  L\left( Z_i, y,  \beta \right)   f\left(y\, |\, Z_i, A_i; \beta \right)
        \right| 
        Y_{({\cal I}_{\bar s})},Z_{({\cal I}_{\bar s})}
        \right]
    \end{align*}
due to cross-sectional independence.
Now, conditional on  $(Y_{({\cal I}_{\bar s})},Z_{({\cal I}_{\bar s})}) $ the terms $ L(Z_i,Y_i, \widehat \beta_{\bar s} )$ are independent and identically distributed across $i$
and have a variance bounded by $(b_{\max} - b_{\min})^2$, which implies that
\begin{align*}
{\rm Var}( \widehat L_s   \, | \,Y_{({\cal I}_{\bar s})},\, Z_{({\cal I}_{\bar s})} ) 
      \leq \frac{(b_{\max} - b_{\min})^2}  {n_s} = O(n^{-1}) .
\end{align*}
By an application of Markov's inequality  we therefore obtain 
\begin{align*}
    \widehat L_s  &=  \mathbb{E}\left[ \left. L(Z_i,Y_i, \widehat \beta_{\bar s} ) \, \right| \, Y_{({\cal I}_{\bar s})},\, Z_{({\cal I}_{\bar s})}  \right]  + O_p(n^{-1/2}),
\end{align*}
where here and in the following $i \in {\cal I}_s$. Evaluating the expectation over $Y_i$ gives
\begin{align}
  \widehat L_s
    &=  \mathbb{E}\left[  \sum_{y\in \mathcal{Y}}  L( Z_i, y, \widehat \beta_{\bar s} )   f\left(y\, |\, Z_i, A_i; \beta_0 \right) 
     \, \Bigg| \, Y_{({\cal I}_{\bar s})},\, Z_{({\cal I}_{\bar s})}  \right]   + O_p(n^{-1/2})
     \notag
  \\
    &=     \mathbb{E}\left[  \sum_{y\in \mathcal{Y}}  L( Z_i, y, \widehat \beta_{\bar s} )   f(y \, |\, Z_i, A_i; \widehat \beta_{\bar s} ) 
     \, \Bigg| \, Y_{({\cal I}_{\bar s})},\, Z_{({\cal I}_{\bar s})}  \right] 
     \notag
  \\ & \quad
     -   \mathbb{E}\left[ \left. \sum_{y\in \mathcal{Y}}  L( Z_i, y, \widehat \beta_{\bar s} )
     \int_0^1 \frac{\partial    f\left(y\, |\, Z_i, A_i;  \beta_0 + t \, ( \widehat \beta_{\bar s} - \beta_0 ) \right) } {\partial \beta'} \, dt
     \, \right| \, Y_{({\cal I}_{\bar s})},\, Z_{({\cal I}_{\bar s})}  \right]
        ( \widehat \beta_{\bar s}  - \beta_0 ) + O_p(n^{-1/2})
        \notag
  \\
  &=      \overline L( \widehat \beta_{\bar s}   )
     + O_p(n^{-1/2}) ,
     \label{eq:1more}
\end{align}
where in the second step we expanded $   f\left(y\,  |\, Z_i, A_i;   \beta \right)$ around $\beta_0$, using the exact first-order Taylor expansion with integral-form remainder,
and in the last step we used the definition of $\overline L(\beta) $ as well as $\widehat \beta_{\bar s}  - \beta_0=O_p(n^{-1/2})$ and
\begin{align*}
  & \left\|  \mathbb{E}\left[ \left.  \sum_{y\in \mathcal{Y}}  L( Z_i, y, \widehat \beta_{\bar s} )
     \int_0^1 \frac{\partial    f\left(y\, | \, Z_i, A_i;  \beta_0 + t \, ( \widehat \beta_{\bar s} - \beta_0 ) \right) } {\partial \beta} \, dt
     \, \right| \, Y_{({\cal I}_{\bar s})},\, Z_{({\cal I}_{\bar s})}  \right]
    \right\|
  \\
  &\qquad  \leq    \max(|b_{\min}|,|b_{\max}|)
     \sup_{\beta \in {B}_\epsilon(\beta_0)}  \sum_{y\in \mathcal{Y}} \mathbb{E}  \left\|   \frac{\partial    f\left(y\, |\, Z_i, A_i;  \beta \right) } {\partial \beta}  \right\| = O(1).
\end{align*}
Here we also used that, conditional on $Y_{({\cal I}_{\bar s})}$ and $Z_{({\cal I}_{\bar s})}$, the integration path $\beta_0 + t \, ( \widehat \beta_{\bar s} - \beta_0 )$, $t \in [0,1]$, is non-random, so that the conditional expectation and the integral over $t$ can be interchanged by Fubini's theorem, and that by the consistency of $\widehat \beta_{\bar s}$ this path lies in ${B}_\epsilon(\beta_0)$, for an $\epsilon>0$, with probability approaching one. Next, we define
\begin{align*}
  \overline m(\beta) &=   \mathbb{E}\left[    m\left( Z_i, A_i,  \beta \right)     \right] .
\end{align*}  
Then, applying the same integral-form expansion to $m\left( Z_i, A_i,  \beta \right)$ (again with a non-random integration path, conditional on $Y_{({\cal I}_{\bar s})}$ and $Z_{({\cal I}_{\bar s})}$) and using Assumption~\ref{ass:MAIN2}, we find that
\begin{align}
        \overline{m} =   \overline m\left(  \beta_0  \right) =   \overline m(  \widehat \beta_{\bar s} ) + O_p(n^{-1/2}).
        \label{eq:11more}
\end{align}
By condition \eqref{eq:boundcondition}, $\overline L( \widehat \beta_{\bar s}  ) \leq  \overline m(  \widehat \beta_{\bar s} )$,
and together with \eqref{eq:1more} and \eqref{eq:11more} this implies
$ \widehat L_s + O_p(n^{-1/2}) \leq    \overline{m} $, for each $s \in \{1,2\}$.
Since $\widehat L_S = ( n_1 \widehat L_1 + n_2 \widehat L_2 ) / n$ is a convex combination of
$\widehat L_1$ and $\widehat L_2$, the same bound carries over to  $\widehat L_S$, that is,
$ \widehat L_S + O_p(n^{-1/2}) \leq    \overline{m} $, as stated.
The derivation of $ \overline{m}  \leq  \widehat U_S + O_p(n^{-1/2}) $ is analogous. This completes the proof.
\end{proof}

\begin{proof}[Proof of Theorem~\ref{theorem:perturbed}]
Let $\widehat U = \widehat U(\beta_0)$, $\overline U = \mathbb E \left[ U(Z_i,Y_i, \beta_0 ) \right]$ and $\widehat{\sigma}_U^2 = \widehat{\sigma}_U^2(\beta_0)$, and let $\widehat{L},\overline{L}$, and $\widehat{\sigma}_L^2$ be defined analogously.
Remember that it was already obtained in \eqref{LimitUasy} and \eqref{LimitLasy} in the proof of Theorem \ref{th:ConsistencyKNOWN} that
\begin{align*}
    \lim_{n\rightarrow \infty }P\left( \overline{U}\leq \widehat{U} + \frac{ c_{\alpha/2}\widehat{\sigma }_{U}} {\sqrt{n}}\right)
    =
    \Phi(c_{\alpha/2}) 
    \qquad
    \text{and}
    \qquad
    \lim_{n\rightarrow \infty }P\left( \overline{L}\geq \widehat{L} - \frac{ c_{\alpha/2}\widehat{\sigma }_{L}} {\sqrt{n}}\right)
    =
    \Phi(c_{\alpha/2}) .
\end{align*}
To keep the notation simple, define (with some abuse of notation)
\begin{align*}
    L(1-\nu,\alpha) 
    &=
    \inf_{\beta \in \mathcal B_{1-\nu}}
            \left( \widehat{L}(\beta)-c_{\alpha/2}\frac{\widehat{\sigma }_{L}(\beta)}{\sqrt{n}}\right),
    \qquad
    L_0(\alpha) 
    =
    \left( \widehat{L}-c_{\alpha/2}\frac{\widehat{\sigma }_{L}}{\sqrt{n}}\right),
    \\
    U(1-\nu,\alpha) 
    &=
    \sup_{\beta \in \mathcal B_{1-\nu}}
        \left(
            \widehat{U}(\beta)+c_{\alpha/2}\frac{\widehat{\sigma }_{U}(\beta) }{\sqrt{n}}
            \right),
    \qquad
    U_0(\alpha) 
    =
    \left(
    \widehat{U}+c_{\alpha/2}\frac{\widehat{\sigma }_{U} }{\sqrt{n}}
            \right).
\end{align*}

Now, notice that 
\begin{align}
    P\left( 
            L_0(\alpha)
            \leq
            \overline{m}
            \leq 
            U_0(\alpha) \notag
        \right)
    \geq &
        P\left( 
            \overline{L}\geq L_0(\alpha) 
            \, \, \cap \, \,
            \overline{U}\leq U_0(\alpha)
        \right)     \notag
    \\
    \geq &
        1-P\left(
            \overline{L}\leq L_0(\alpha)
            \; \; \cup \; \;
            \overline{U}\geq U_0(\alpha)
            \right)     \notag
	\\
	    \geq & 1 - P\left( \overline{L}\leq L_0(\alpha) \right) 
            - P\left( \overline{U}\geq U_0(\alpha) \right) ,
        \label{mres1}
\end{align}
where in obtaining the first inequality we have used $\overline{L} \leq \overline{m} \leq \overline{U}$.
Moreover, analogous to the arguments used in the proof of Theorem 11 of \citet{CFHN13},
\begin{align}
    P\left( L_0(\alpha)
        \leq
        \overline{m}
        \leq 
        U_0(\alpha)
    \right)
    &=
    P\left( L_0(\alpha)
        \leq
        \overline{m}
        \leq 
        U_0(\alpha)
        \,\,\, \bigcap \,\,\,
        \beta_0 \in \mathcal B_{1-\nu}
    \right)
    \notag
    \\
    & \quad + 
    P\left( L_0(\alpha)
        \leq
        \overline{m}
        \leq 
        U_0(\alpha)
        \,\,\, \bigcap \,\,\,
        \beta_0 \notin \mathcal B_{1-\nu}
    \right)
    \notag
    \\
    & \leq
    P\left( L_0(\alpha)
        \leq
        \overline{m}
        \leq 
        U_0(\alpha)
        \,\,\, \bigcap \,\,\,
        \beta_0 \in \mathcal B_{1-\nu}
    \right)
    + 
    P\left( \beta_0 \notin \mathcal B_{1-\nu} \right)
    \notag
    \\
    & \leq
    P\left( L(1-\nu, \alpha)
        \leq
        \overline{m}
        \leq
        U(1-\nu, \alpha)
    \right)
    +
    \nu + o(1),
    \label{mres2}
\end{align}
where the final $o(1)$ term accounts for the fact that the hypothesis on $\mathcal B_{1-\nu}$ only guarantees $\limsup_{n\to\infty} P\left( \beta_0 \notin \mathcal B_{1-\nu} \right) \leq \nu$.
Combining \eqref{mres1} and \eqref{mres2}, and taking limits, it follows that
\begin{align*}
    \liminf_{n\to\infty} 
    P\left( L(1-\nu, \alpha)
        \leq
        \overline{m}
        \leq 
        U(1-\nu, \alpha)
    \right)
    \geq
    1-\alpha - \nu,
\end{align*}
as stated.
\end{proof}

\begin{proof}[\bf Proof of Lemma \ref{lem:inflemma}]
Given $\beta_0  \in {\rm Conv}(\mathcal B_{\rm sub})$,
by the definition of the convex hull we have
$
   \beta_0 = \sum_{\beta \in  \mathcal B_{\rm sub}} \lambda_\beta \,  \beta , 
$
for some convex weights $\lambda_\beta \geq 0$ such that $\sum_{\beta \in  {\mathcal B}_{\rm sub}} \lambda_\beta  = 1$. 
To keep the notation simple, define $\ell(y)=L(z,y, {\mathcal B}_{\rm sub})$ and $u(y)=U(z,y,{\mathcal B}_{\rm sub})$.
Then, under the assumption that $L(z,y,\widehat{\mathcal B}_{\bar s})$ and $U(z,y,\widehat{\mathcal B}_{\bar s})$ satisfy \eqref{eq:boundconditionSET}
for all
$ \beta \in \mathcal B_{\rm sub}$ and
$  a\in \mathcal{A}$,
conditional on the event $\widehat{\mathcal B}_{\bar s} = \mathcal{B}_{\rm sub}$
we have
$$\sum_{y\in \mathcal{Y}}\ell (y) \, f(y \, | \, z, a;\beta )
 \leq m(z,a,\beta )\leq \sum_{y\in \mathcal{Y}}u(y) f(y \, | \, z, a;\beta ) .
 $$
Multiplying this expression by $\lambda_\beta$ and summing
over $\beta \in \mathcal B_{\rm sub}$ then gives
\begin{align}
    \sum_{y\in \mathcal{Y}}\ell (y) \,
    \sum_{\beta \in \mathcal B_{\rm sub}} \lambda_\beta \, f(y \, | \, z, a;\beta )
    \leq 
    \sum_{\beta \in \mathcal B_{\rm sub}} \lambda_\beta \,  m(z,a,\beta )
    \leq
    \sum_{y\in \mathcal{Y}}u(y)
    \sum_{\beta \in \mathcal B_{\rm sub}} \lambda_\beta \,
    f(y \, | \, z, a;\beta ) .
    \label{eq:thm3bnd}
\end{align}
Using Assumption \ref{ass:Bf} we can employ
a second-order Taylor expansion of $f(y \, | \, z, a;\beta )$ around $\beta_0$, with integral-form remainder, to find for $\beta \in {\mathcal B}_{\rm sub}$,
\begin{align*}
    f(y \, | \, z, a;\beta )
    &=
    f(y \, | \, z, a;\beta_0 )
    +
    \frac{\partial f(y \, | \, z, a;\beta_0 )}{\partial \beta'}
    (\beta - \beta_0)
    \\
    &\qquad
    +
    \int_0^1 (1-t) \,
    (\beta - \beta_0)'
    \frac{\partial^2 f\left(y \, | \, z, a; \beta_0 + t \, (\beta - \beta_0) \right)}{\partial \beta \partial \beta'}
    (\beta - \beta_0)
    \, dt
    \\
    &=
    f(y \, | \, z, a;\beta_0 )
    +
    \frac{\partial f(y \, | \, z, a;\beta_0 )}{\partial \beta'}
    (\beta - \beta_0)
    +
    O\left({\rm d}_{\rm sub}^2\right),
\end{align*}
where ${\rm d}_{\rm sub} = {\rm diam}(\mathcal B_{\rm sub})$, and the pointwise $O\left({\rm d}_{\rm sub}^2\right)$ remainder terms here and below generally depend on $(y,z,a)$; since the integration path $\beta_0 + t \, (\beta - \beta_0)$ is deterministic and contained in $B_\epsilon(\beta_0)$, Assumption~\ref{ass:Bf}(ii) and Fubini's theorem guarantee that the expectations of these remainder terms over $(Z_i,A_i)$, taken below, are also bounded in absolute value by a constant times ${\rm d}_{\rm sub}^2$.
By definition $\sum_{\beta \in \mathcal B_{\rm sub} } \lambda_\beta (\beta - \beta_0)=0$. It then follows from the expansion in the previous display that 
\begin{align}
    \sum_{\beta \in  {\mathcal B}_{\rm sub}} \lambda_\beta \, f(y \, | \, z, a;\beta )
    &=   
    f(y \, | \, z, a;\beta_0 )
    +  
    O\left({\rm d}_{\rm sub}^2\right).
    \label{eq:thm3exp1}
\end{align}
Analogous arguments also yield
\begin{align}
    \sum_{\beta \in  \mathcal B_{\rm sub}} \lambda_\beta \,  m(z,a,\beta )
    &= 
    m(z,a,\beta_0 )
    +  
    O\left({\rm d}_{\rm sub}^2\right).
    \label{eq:thm3exp2}
\end{align}
Then, by combining \eqref{eq:thm3bnd}, \eqref{eq:thm3exp1} and \eqref{eq:thm3exp2} we obtain
\begin{align}
    \sum_{y\in \mathcal{Y}}\ell (y) \, f(y \, | \, z, a;\beta_0 )
    + 
    O\left({\rm d}_{\rm sub}^2\right)
    &\leq 
    m(z,a,\beta_0 )
    \leq \sum_{y\in \mathcal{Y}}u(y) f(y   \, | \, z, a;\beta_0 )
    +  
    O\left({\rm d}_{\rm sub}^2\right).
    \label{eq:lem1ineq}
\end{align}
Taking expectations of all sides of \eqref{eq:lem1ineq} finally yields,
\begin{align*}
    & \mathbb E 
    \left[
    \left.
    L(Z_i,Y_i,\widehat{\mathcal B}_{\bar{s}})
    \, \right | \, 
    \widehat{\mathcal B}_{\bar s} = \mathcal B_{\rm sub}
    \right]
    + 
    O\left({\rm d}_{\rm sub}^2\right)
    \leq 
    \overline m
    \leq 
    \mathbb E 
    \left[
    \left.
    U(Z_i,Y_i,\widehat{\mathcal B}_{\bar{s}})
    \, \right | \,
     \widehat{\mathcal B}_{\bar s} = \mathcal B_{\rm sub}
    \right]
    + 
    O\left({\rm d}_{\rm sub}^2\right),
\end{align*}
where the conditioning on $\widehat{\mathcal B}_{\bar s} = \mathcal{B}_{\rm sub}$ is required as the derivations leading up to \eqref{eq:lem1ineq} are based on this condition. Finally, $\mathbb E[m(Z_i,A_i,\beta_0)
    \, | \,
    \widehat{\mathcal B}_{\bar s} = \mathcal B_{\rm sub}] 
    = 
    \mathbb E[m(Z_i,A_i,\beta_0)]$
follows since the marginal distribution of $(Y_i,Z_i,A_i)$ is independent across $i$. This completes the proof.
\end{proof}

\begin{proof}[Proof of Theorem \ref{th:ConsistencyBf}]
Before moving to the proof, we make a series of definitions for notational brevity. First, define
\begin{align*}
     \overline L(  \widehat {\cal B}_{\bar s(i)}  ) 
     &=  
     \mathbb{E}\left[ L(Z_i,Y_i,   \widehat {\cal B}_{\bar s(i)}  ) 
      \, \Big| \,  \widehat {\cal B}_{\bar s(i)} 
       \right] ,
\end{align*}
where the expectation is with respect to the joint distribution of $(Y_i,Z_i)$ with $i\in \mathcal I _s$, conditional on $\widehat {\cal B}_{\bar s(i)}$. 
We next define the centered quantity
\begin{align*}
   \widetilde L(Z_i,Y_i,  \widehat {\cal B}_{\bar s(i)} ) 
   &= 
   L(Z_i,Y_i,  \widehat {\cal B}_{\bar s(i)} ) 
     - \overline L( \widehat {\cal B}_{\bar s(i)} ),
\end{align*}
and the half-sample averages
\begin{align*}
    \widetilde L_{C,s}
    =
    \frac{2}{n}
    \sum_{i\in \mathcal I_s}\widetilde L(Z_i,Y_i,  \widehat {\cal B}_{\bar s(i)} )
    \qquad
    \text{and}
    \qquad
    \widehat L_{C,s}  
    = 
    \frac 2 n \sum_{i \in {\cal I}_s} 
    L(Z_i,Y_i,  \widehat {\cal B}_{\bar s(i)} ).
\end{align*}
The corresponding quantities 
$\overline U( \widehat {\cal B}_{\bar s(i)} ) $, 
$\widetilde U(Z_i,Y_i,  \widehat {\cal B}_{\bar s(i)} ) $,
$\widetilde U_{C,s}$
and
$\widehat U_{C,s}  $
are defined analogously.
We further define
\begin{align*}
    \mathcal L_{C,s}
    =
    \widehat L_{C,s} - \frac{ c_{\alpha/4} \,
        \widehat{\sigma }_{L,s}} {\sqrt{n/2}}
    \qquad
    \text{and}
    \qquad
    \mathcal U_{C,s}
    =
    \widehat U_{C,s} + \frac{ c_{\alpha/4} \,
        \widehat{\sigma }_{U,s}} {\sqrt{n/2}}.
\end{align*}
Notice that
$ \widehat L_C =  (\widehat L_{C,1} + \widehat L_{C,2})/2 $ and $ \widehat U_C =  (\widehat U_{C,1} + \widehat U_{C,2})/2 $.
Finally, let 
\begin{align*}
    \mathcal L_C
    =
    \frac{\mathcal L_{C,1} + \mathcal L_{C,2}}{2}
    \qquad
    \text{and}
    \qquad
    \mathcal U_C
    =
    \frac{\mathcal U_{C,1} + \mathcal U_{C,2}}{2}.
\end{align*}
These quantities are equivalent to the lower and upper bounds in the probability statement of Theorem \ref{th:ConsistencyBf}. We therefore want to prove 
\begin{align*}
    P\left(
        \mathcal L_C
        \leq
        \overline m
        \leq  
        \mathcal U_C
    \right)
    \geq
    1-\alpha - \nu + o(1),
    \qquad
    \text{as } n\to\infty.    
\end{align*}
Lemma~\ref{lem:inflemma} states that
conditional on $\beta_0 \in  {\rm Conv}(\widehat {\cal B}_{\bar s(i)})$ and $\widehat {\cal B}_{\bar s(i)} \subset B_\epsilon(\beta_0)$, we have
\begin{align}
    \overline L( \widehat {\cal B}_{\bar s(i)} )
    +
    \delta_{L}(\widehat {\cal B}_{\bar s(i)})
    \leq 
    \overline m
    \leq
    \overline U( \widehat {\cal B}_{\bar s(i)} )
    +
    \delta_{U}(\widehat {\cal B}_{\bar s(i)}),
    \label{LemmaReformulated}
\end{align}
where we have introduced the notation
$\delta_{L}(\widehat {\cal B}_{\bar s(i)})$
and 
$\delta_{U}(\widehat {\cal B}_{\bar s(i)})$
for
the lower and upper bound
$O([{\rm diam}(\widehat {\cal B}_{\bar s})]^2)$
remainder terms in Lemma~\ref{lem:inflemma}.
In what follows, let ${\cal A}_{s}$ denote the event that $ \beta_0 \in  {\rm Conv}(\widehat {\cal B}_{s})$ and $\widehat {\cal B}_{s} \subset B_\epsilon(\beta_0)$, with the complement given by ${\cal A}_{s}^c$. Note that if $\beta_0 \in  {\rm Conv}(\widehat {\cal B}_{s})$ and ${\rm diam}(\widehat {\cal B}_{s}) \leq \epsilon$, then every element of $\widehat {\cal B}_{s}$ is within distance $\epsilon$ of $\beta_0$; together with Markov's inequality and Assumption~\ref{ass:Bf}(i), which give $P\{{\rm diam}(\widehat {\cal B}_{s}) > \epsilon\} \leq \mathbb E \{ [{\rm diam}(\widehat {\cal B}_{s})]^2 \}/\epsilon^2 = o(1)$, it follows that $P({\cal A}_{s}^c) \leq \nu/2 + o(1)$. Now, observe that
\begin{align}
     P & \left(
    \mathcal{L}_{C,s}\geq \overline m 
    \, \bigcup \, 
    \mathcal{U}_{C,s}\leq \overline m
    \right)
     \notag
    \\ 
    &= 
    P\left(
    \left\{
    \mathcal{L}_{C,s}\geq \overline m 
    \, \bigcup \, 
    \mathcal{U}_{C,s}\leq \overline m
    \right\}
    \, \bigcap \mathcal A_{\bar s}
    \right)
    \notag
    +
    P\left(
    \left\{
    \mathcal{L}_{C,s}\geq \overline m 
    \, \bigcup \, 
    \mathcal{U}_{C,s}\leq \overline m
    \right\}
    \, \bigcap \mathcal A_{\bar s}^c
    \right)
    \notag
    \\
    &\leq 
    P\left(
    \mathcal{L}_{C,s}\geq \overline m 
    \, \bigcap \mathcal A_{\bar s}
    \right)
    +
    P\left(
    \mathcal{U}_{C,s}\leq \overline m
    \, \bigcap \mathcal A_{\bar s}
    \right)
    +
    P\left(
        \mathcal A_{\bar s}^c
    \right)
    \notag
    \\
    &=
    P\left(
    \left.
    \mathcal{L}_{C,s}\geq \overline m 
    \, \right| \, \mathcal A_{\bar s}
    \right)
    P(\mathcal A_{\bar s})
    +
    P\left(
    \left.
    \mathcal{U}_{C,s}\leq \overline m
    \, \right| \, \mathcal A_{\bar s}
    \right)
    P(\mathcal A_{\bar s})
    +
    P\left(
        \mathcal A_{\bar s}^c
    \right)
     \notag
     \\
    &\leq
    P\left(
    \left.
    \mathcal{L}_{C,s}
        \geq 
        \overline L(  \widehat {\cal B}_{\bar s} )
        +
        \delta_{L,\bar s} 
        \, \right| \, \mathcal A_{\bar s}
    \right) P(\mathcal A_{\bar s})
    +
    P\left(
    \left.
    \mathcal{U}_{C,s}
        \leq 
        \overline U(  \widehat {\cal B}_{\bar s} )
        +
        \delta_{U,\bar s} 
          \, \right| \, \mathcal A_{\bar s}
    \right) P(\mathcal A_{\bar s})
    +
    \frac{\nu}{2} 
    + o(1)
       \notag
       \\
    &\leq
    P\left(
    \mathcal{L}_{C,s}
        \geq
        \overline L(  \widehat {\cal B}_{\bar s} )
        +
        \delta_{L,\bar s}
    \right)
    +
    P\left(
    \mathcal{U}_{C,s}
        \leq
        \overline U(  \widehat {\cal B}_{\bar s} )
        +
        \delta_{U,\bar s}
    \right)
    +
    \frac{\nu}{2}
    + o(1),
    \label{eq:be0}
\end{align}
where in the second to last step we have used $P({\cal A}_{\bar s}^c) \leq \nu/2 + o(1)$
and \eqref{LemmaReformulated},
and we defined
 $ \delta_{L,\bar s} =  \delta_{L}(  \widehat {\cal B}_{\bar s} ) $
 and 
 $ \delta_{U,\bar s} =  \delta_{U}(  \widehat {\cal B}_{\bar s} ) $.
Next, we obtain
\begin{align}
    P\left(
        \mathcal L_C
        \leq
        \overline m
        \leq  
        \mathcal U_C
    \right)
    \geq &
    P\left(
        \mathcal{L}_{C,1} \leq \overline m \leq \mathcal{U}_{C,1}
        \, \bigcap \,
        \mathcal{L}_{C,2} \leq \overline m \leq \mathcal{U}_{C,2}
    \right)
    \notag
    \\
    \geq &
    1
    -
    P\left(
    \mathcal{L}_{C,1}\geq \overline m 
    \, \bigcup \, 
    \mathcal{U}_{C,1}\leq \overline m
    \right)
    -
    P\left(
    \mathcal{L}_{C,2}\geq \overline m 
    \, \bigcup \, 
    \mathcal{U}_{C,2}\leq \overline m
    \right)
    \notag
    \\
    \geq &
    1
    -\nu
    -
    \sum_{s=1}^2
    \left[
    P\left(
    \mathcal{L}_{C,s}
        \geq 
        \overline L(  \widehat {\cal B}_{\bar s} )
        +
        \delta_{L,\bar s} 
    \right)
    +
    P\left(
    \mathcal{U}_{C,s}
        \leq 
        \overline U(  \widehat {\cal B}_{\bar s} )
        +
        \delta_{U,\bar s} 
    \right)
    \right] 
    + o(1)
    \notag 
    \\
    \geq &
    1
    -
    \nu
    -
    \sum_{s=1}^2
    P\left( 
    \sqrt{n/2}  \frac {\widetilde L_{C,s} } {\widehat{\sigma}_{L,s}} 
    \geq 
    c_{\alpha/4}+
    \sqrt{n/2}\frac{ \delta_{L,\bar s} }
    {\widehat{\sigma}_{L, s}} 
    \right) 
    \notag
    \\
    &-
    \sum_{s=1}^2
    P\left( 
    \sqrt{n/2}  \frac {\widetilde U_{C,s} } {\widehat{\sigma}_{U,s}} 
    \leq 
    - c_{\alpha/4}
    +
    \sqrt{n/2}\frac{ \delta_{U,\bar s} }
    {\widehat{\sigma}_{U, s}} 
    \right)  
     + o(1),
    \label{eq:be1}
\end{align}
where the third inequality follows from \eqref{eq:be0} and the last inequality applies the various definitions we made earlier, on the event $\{\widehat{\sigma}_{L,s}>0 \text{ and } \widehat{\sigma}_{U,s}>0\}$; the complement of this event has probability $o(1)$ (it is contained in $({\cal E}_{n,s} \cap {\cal V}_{n,s})^c$, as shown below) and is absorbed into the $o(1)$ term.

It remains to show that the probabilities $P(\cdot)$ that explicitly appear in the last inequality of \eqref{eq:be1} are all bounded from above by $\alpha/4 + o(1)$.
To show this, we first note that conditional on ${\cal G}_{\bar s}=\{(Y_j,Z_j):j\in \mathcal I_{\bar s}\}$, 
$\widetilde L (Z_i,Y_i,\widehat{\mathcal B}_{\bar s})$ is centered and  iid over $i \in \mathcal I_s$. Next, define $M_{r,\bar{s}} = \mathbb E \left[\left. |\widetilde L (Z_i,Y_i,\widehat{\mathcal B}_{\bar s})|^r \, \right| \, \mathcal{G}_{\bar s}\right]$. Since $L(z,y,\widehat{\mathcal B}_{\bar s})$ is bounded by $b_{\rm min}$ and $b_{\rm max}$, $M_{r,\bar s}$ exists for any $r>0$.
Note also that $M_{2,\bar s} = \sigma^2_{L,\bar s} := {\rm Var}\left[ \left. L(Z_i,Y_i,  \widehat {\cal B}_{\bar s} ) \, \right| \, {\cal G}_{\bar s} \right]$. Since the observations with $i \in {\cal I}_s$ are independent of ${\cal G}_{\bar s}$, this conditional variance depends on ${\cal G}_{\bar s}$ only through $\widehat {\cal B}_{\bar s}$, and the variance condition of the theorem therefore guarantees that the ${\cal G}_{\bar s}$-measurable event
${\cal V}_{n,s} := \{ \sigma^2_{L,\bar s} \geq c \}$
satisfies $P({\cal V}_{n,s}^c) = o(1)$.
It follows by Theorem 1.1 of \citet{BG96} that there exists a universal constant $k>0$ such that, for all realizations of ${\cal G}_{\bar s}$ with $M_{2,\bar s}>0$ (in particular on ${\cal V}_{n,s}$),
\begin{align*}
    \sup_{t\in \mathbb{R}}
    \left|
        P\left(
        \left.
            \sqrt{n/2}
            \frac{\widetilde L_{C,s}}{\widehat{\sigma}_{L,s}}
            <t
        \right|
        \mathcal G_{\bar s}
        \right)
        -
        \Phi(t)
    \right|
    \leq
    \frac{1}{\sqrt{n/2}}\frac{k \, M_{3,\bar s}}{(M_{2,\bar s})^{3/2}}.
\end{align*}
Since $|\widetilde L (Z_i,Y_i,\widehat{\mathcal B}_{\bar s})| \leq b_{\max} - b_{\min}$, we have $M_{3,\bar s} \leq (b_{\max} - b_{\min}) \, M_{2,\bar s}$, so that on the event ${\cal V}_{n,s}$ the right-hand side of this bound is at most $k \, (b_{\max} - b_{\min}) \, c^{-1/2} \, (n/2)^{-1/2}$. Hence, on ${\cal V}_{n,s}$ it holds that
\begin{align}
    P\left(
        \left.
            \sqrt{n/2}
            \frac{\widetilde L_{C,s}}{\widehat{\sigma}_{L,s}}
            <t
        \, \right| \,
        \mathcal G_{\bar s}
    \right)
    =
    \Phi(t) + O\left( n^{-1/2} \right),
    \label{eq:BE1}
\end{align}
where the rate $O\left( n^{-1/2} \right)$ holds uniformly over $t\in \mathbb R$ and over the realizations of ${\cal G}_{\bar s}$ in ${\cal V}_{n,s}$.
The threshold
$c_{\alpha/4}
    +
    \sqrt{n/2} \delta_{L,\bar s}
    /
    \widehat{\sigma}_{L,s} $
that we are interested in is itself random conditional on ${\cal G}_{\bar s}$, because $\widehat{\sigma}_{L,s}$ is computed from the observations $i \in {\cal I}_s$. We therefore first bound it by a threshold that is measurable with respect to  ${\cal G}_{\bar s}$. Define the event
\begin{align*}
    {\cal E}_{n,s}
    :=
    \left\{ \widehat{\sigma}_{L,s} \geq \sigma_{L,\bar s} / 2 \right\} .
\end{align*}
Conditional on ${\cal G}_{\bar s}$, the variables $L(Z_i,Y_i,  \widehat {\cal B}_{\bar s} )$, $i \in {\cal I}_s$, are independent and identically distributed and take values in $[b_{\min},b_{\max}]$, so that
$\mathbb E \left[ \left. ( \widehat{\sigma}^2_{L,s} - \sigma^2_{L,\bar s} )^2 \, \right| \, {\cal G}_{\bar s} \right] = O(n^{-1})$
uniformly over ${\cal G}_{\bar s}$. On ${\cal V}_{n,s}$ we have $\sigma^2_{L,\bar s} \geq c$, and Chebyshev's inequality then yields
$P( {\cal E}^c_{n,s} \cap {\cal V}_{n,s} ) = O( n^{-1} )$; together with $P( {\cal V}^c_{n,s} ) = o(1)$ this gives
$P\left( ( {\cal E}_{n,s} \cap {\cal V}_{n,s} )^c \right) = o(1)$.
Moreover, on the event ${\cal E}_{n,s} \cap {\cal V}_{n,s}$ we have
$\left| \delta_{L,\bar s} / \widehat{\sigma}_{L,s} \right|
\leq  2 \left| \delta_{L,\bar s} \right| / \sigma_{L,\bar s}
\leq  2 \left| \delta_{L,\bar s} \right| / \sqrt{c}$,
and therefore
\begin{align*}
    \left\{
            \sqrt{n/2}
            \frac{\widetilde L_{C,s}}{\widehat{\sigma}_{L,s}}
            \geq
            c_{\alpha/4}+
            \sqrt{n/2}\frac{ \delta_{L,\bar s}}
            {\widehat{ \sigma}_{L,s}}
    \right\}
    \bigcap \,
    {\cal E}_{n,s}
    \cap \,
    {\cal V}_{n,s}
    \;  \subseteq  \;
    \left\{
            \sqrt{n/2}
            \frac{\widetilde L_{C,s}}{\widehat{\sigma}_{L,s}}
            \geq
            c_{\alpha/4}
            -
            \frac{2}{\sqrt{c}} \sqrt{\frac n 2} \left| \delta_{L,\bar s} \right|
    \right\} .
\end{align*}
The threshold on the right-hand side is measurable with respect to ${\cal G}_{\bar s}$, so that
\eqref{eq:BE1} applies to it and yields
\begin{align}
    & P\left(
        \left.
            \left\{
            \sqrt{n/2}
            \frac{\widetilde L_{C,s}}{\widehat{\sigma}_{L,s}}
            \geq
            c_{\alpha/4}+
            \sqrt{n/2}\frac{ \delta_{L,\bar s}}
            {\widehat{ \sigma}_{L,s}}
            \right\}
            \bigcap \,
            {\cal E}_{n,s}
            \cap \,
            {\cal V}_{n,s}
        \, \right| \,
        \mathcal G_{\bar s}
    \right)
    \notag
    \\
  & \qquad \qquad \leq
    1
    -
    \Phi\left(
        c_{\alpha/4}
        -
        \frac{2}{\sqrt{c}}
        \sqrt{\frac{n}{2}}
        \left| \delta_{L,\bar s} \right|
    \right)
    +
    O\left( n^{-1/2} \right)
    \notag
    \\
  & \qquad \qquad =
    1
    -
    \Phi\left( c_{\alpha/4} \right)
    +
    O\left(
        \sqrt{\frac{n}{2}}
        \left| \delta_{L,\bar s} \right|
        +
        \frac{1}{\sqrt{n}}
    \right)
    =
    \frac{\alpha}{4}
    +
    O\left(
        \sqrt{\frac{n}{2}}
        \left| \delta_{L,\bar s} \right|
        +
        \frac{1}{\sqrt{n}}
    \right),
    \label{eq:BE2}
\end{align}
where the second to last equality expands $\Phi$ around $ c_{\alpha/4}$,
and the final equality follows from the definition of $c_{\alpha/4}$. Note that \eqref{eq:BE2} also holds for realizations of ${\cal G}_{\bar s}$ outside of ${\cal V}_{n,s}$, because there the left-hand side is zero (${\cal V}_{n,s}$ being ${\cal G}_{\bar s}$-measurable).
Since $\delta_{L,\bar s}$ is measurable with respect to ${\cal G}_{\bar s}$,
taking expectations over  ${\cal G}_{\bar s}$ in \eqref{eq:BE2}, and adding
$P\left( ( {\cal E}_{n,s} \cap {\cal V}_{n,s} )^c \right) = o(1)$, finally yields
\begin{align}
    P\left(
        \sqrt{n/2}
        \frac{\widetilde L_{C,s}}{\widehat{\sigma}_{L,s}}
        \geq
        c_{\alpha/4}+
        \sqrt{n/2}\frac{ \delta_{L,\bar s}}
        {\widehat{ \sigma}_{L,s}}
    \right)
    &\leq
    \frac{\alpha}{4}
    +
    O\left(
        \sqrt{\frac{n}{2}}
        \, \mathbb E \left|  \delta_{L,\bar s} \right|
        +
        \frac{1}{\sqrt{n}}
    \right)
    +
    o(1)
    =
    \frac{\alpha}{4}
    +
    o\left( 1 \right)
    \label{eq:be2},
\end{align}
where the final $o(1)$ rate follows from
$\mathbb E \left|  \delta_{L,\bar s} \right|
= O\left( \mathbb E \left\{ \left[ {\rm diam}( \widehat{\cal B}_{\bar s} ) \right]^2 \right\} \right)
= o\left( n^{-1/2} \right)$,
which holds by Assumption \ref{ass:Bf}(i). By analogous arguments one obtains
\begin{align}
    P\left( 
    \sqrt{n/2}  \frac {\widetilde U_{C,s} } {\widehat{\sigma}_{U,s}} 
    \leq 
    - c_{\alpha/4}
    +
    \sqrt{n/2}\frac{ \delta_{U,\bar s} }
    {\widehat{\sigma}_{U, s}}
    \right)
    \leq
    \frac{\alpha}{4} + o(1).
    \label{eq:be3}
\end{align}
Combining \eqref{eq:be1}, \eqref{eq:be2} and \eqref{eq:be3} yields the stated result.
\end{proof}

\newpage

\section{Additional results for dynamic models}\label{app:dynamic}

Continuing the discussion of Section~\ref{sect:setcomparison}, we now also explore the dynamic logit model with a continuous covariate. The DGP is
	\begin{eqnarray*}
		& Y_{it} = 1\left\{ Y_{i,t-1} \gamma + X_{it} \beta + A_i \geq \varepsilon_{it}  \right\} \quad \text{for } t=1,\ldots, T, \\
		& Y_{i0} = 1\left\{ X_{i0} \beta + A_i \geq \varepsilon_{i0}  \right\}, \quad X_{it} \sim N (A_i,1), \quad A_i \sim N(0,1), \quad \varepsilon_{it} \sim {\rm Logit}(0,1),
	\end{eqnarray*}
	and we consider the average effect
	\begin{equation}
		\mathbb{E} 
		\left[ 
			\frac{1}{T} \sum_{t=1}^T \mathbb{E} \left[ m(X_{it},A_i, \gamma,\beta) | X_{it} \right]
		\right], \label{eq:aped}
	\end{equation}
	where
	\begin{align*}    
		m(X_{it},A_i,\gamma,\beta) &= P(Y_{it}=1 | Y_{i,t-1}=1 \, ,X_{it}, \, A_i \, , \gamma , \beta ) 
  \\ & \qquad 
		- P(Y_{it}=1 | Y_{i,t-1}=0 \, ,X_{it} ,A_i, \, \gamma ,\beta ).
	\end{align*}	
	
   \citet{AC21} have shown that in the given setting, the average effect in 
	\eqref{eq:aped} is point-identified (see their Proposition 3). 
    The comparison in this part is then that between the outer bounds and the \textit{point-identified} average effect. We investigate the behavior of the outer bounds in this case in panels of size $n=1000$ and $T\in \{ 4,6,8 \}$ with $\beta=1$ and $\gamma \in \left[ -2,2 \right]$. The support of $A_i$ is approximated by   a grid of 50 equidistant points between $-5$ and $5$. The results, presented in Figure \ref{fig:idset_dlogit}, are based on 1000 replications and confirm that the outer bounds  nearly point-identify the average effect, except when $\gamma$ is large; however this issue tends to disappear as $T$ increases. This is not surprising since under a large $\gamma$, the term $\gamma Y_{i,t-1}$ will act similarly to a fixed effect. 

\begin{figure}[htb]
	\centering
	\includegraphics[width=1\linewidth]{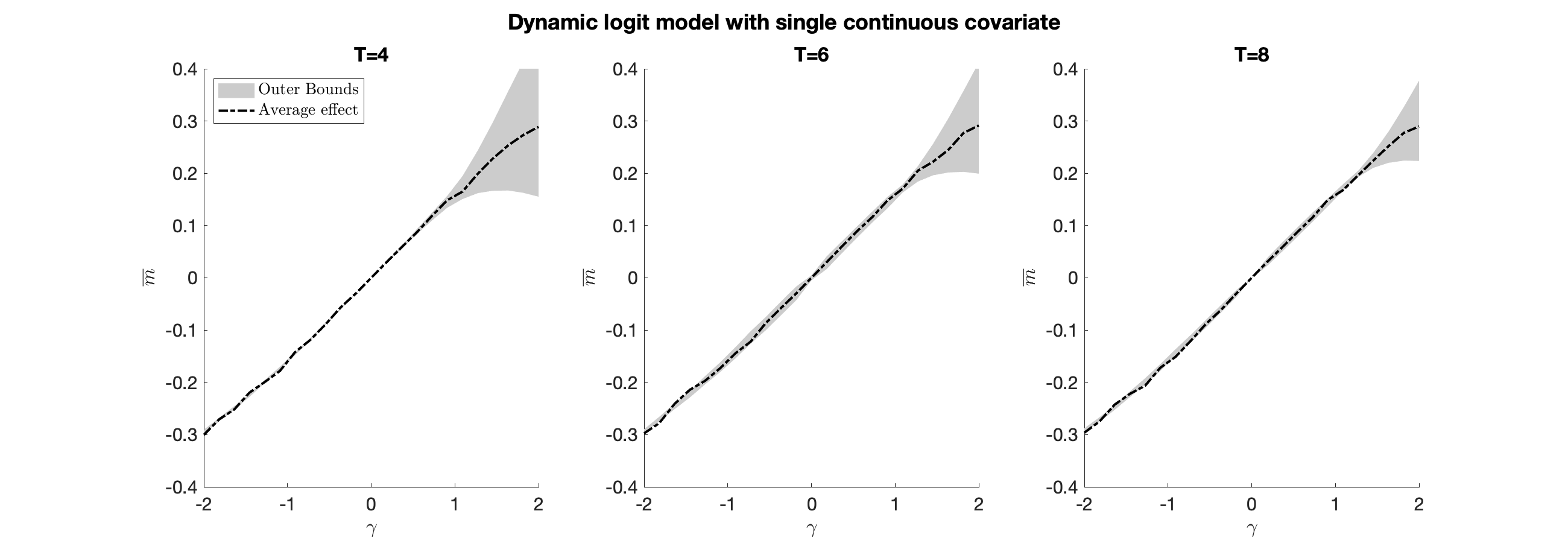}
	\caption{{Comparison of the outer bounds  and the point-identified average effect for the dynamic logit model $Y_{it} = 1\left\{ Y_{i,t-1}\gamma + X_{it} \beta + A_i \geq \varepsilon_{it} \right\}$, where $\varepsilon_{it} \sim {\rm Logit} (0,1)$, $A_i \sim N(0,1)$, $X_{it}\sim N(A_i,1)$, and $Y_{i0} = 1\{X_{i0}\beta + A_i \geq \varepsilon_{i0} \}$. Results for each $\gamma\in[-2,2]$ are based on 1000 replications of panels with cross-section size $n=1000$. Reported  outer bounds   are the cross-replication averages. The average effect of interest is based on \eqref{eq:aped}.}}
	\label{fig:idset_dlogit}
\end{figure}

Finally, we consider the random coefficient dynamic logit model given by
\begin{eqnarray*}
	& Y_{it} = 1 \left\{ Y_{i,t-1} A_{2,i} + A_{1,i} \geq \varepsilon_{it} \right\} \quad \text{for } t=1,\ldots,T, \\
	& Y_{i0} = 1 \left\{ A_{1,i} \geq \varepsilon_{i0} \right\}, \quad A_{1,i} \sim N(0,1/2), \quad A_{2,i} \sim N(A_2,1/2), \quad \varepsilon_{it} \sim {\rm Logit}(0,1).
\end{eqnarray*}
For this exercise, we focus on the average effect
\begin{align}
	\mathbb{E} \left[ P(Y_{it}=1 | Y_{i,t-1}=1, A_{1,i},A_{2,i}) - P(Y_{it}=1 | Y_{i,t-1}=0 , A_{1,i},A_{2,i}) \right].
    \label{eq:aerdc}
\end{align}
\begin{figure}[ht]
	\centering
	\includegraphics[width=0.8\linewidth]{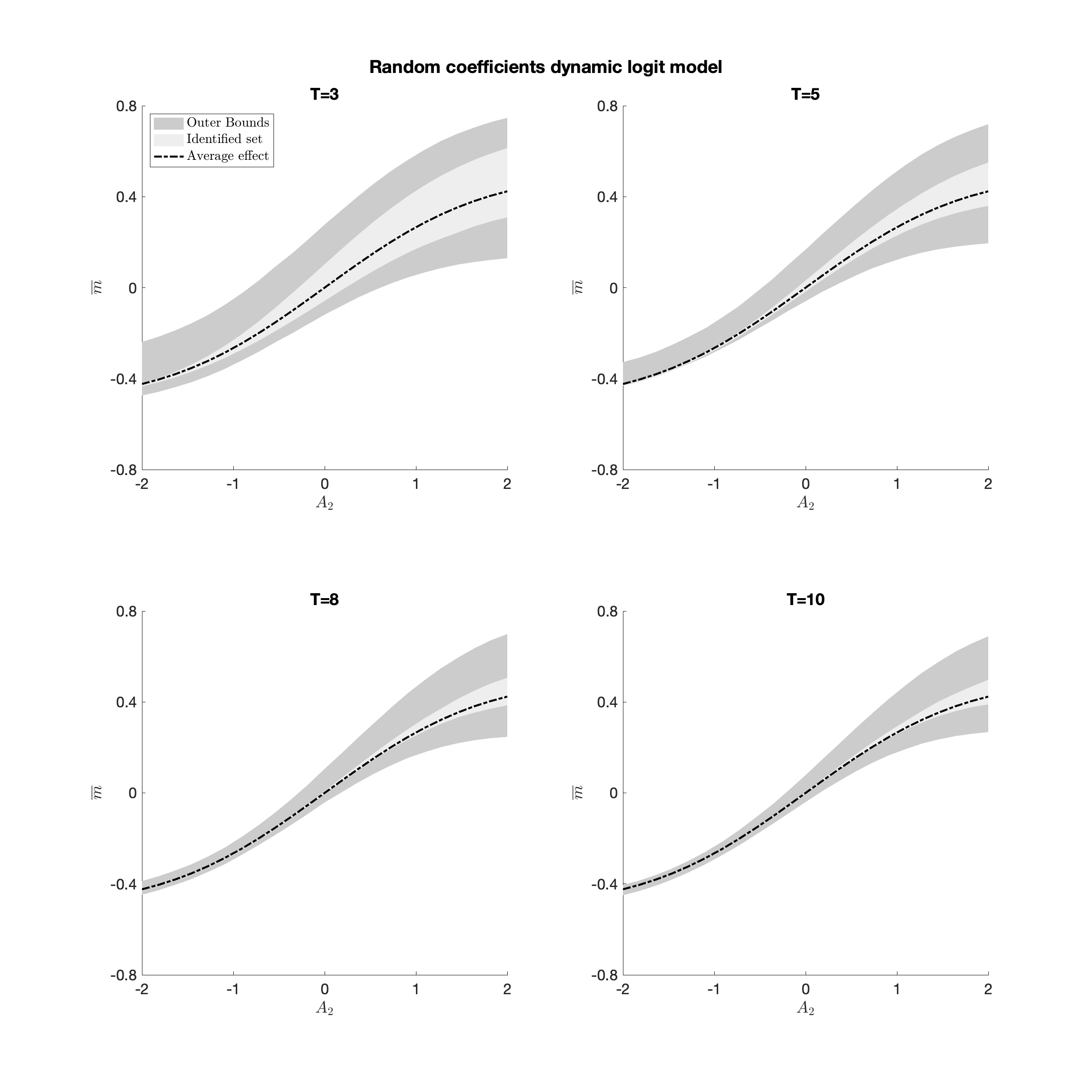}
	\caption{Comparison of the outer bounds  and the identified set for the random coefficient dynamic logit model $Y_{it} = 1\left\{ Y_{i,t-1} A_{2,i} + A_{1,i} \geq \varepsilon_{it} \right\}$, where $\varepsilon_{it} \sim {\rm Logit}(0,1)$, $A_{1,i} \sim N(0,1/2)$, $A_{2,i} \sim N(A_2,1/2)$, and $Y_{i0} = 1\left\{ A_{1,i} \geq \varepsilon_{i0} \right\}$. The average effect of interest is based on \eqref{eq:aerdc}. Results for each $A_2\in[-2,2]$ are based on 1000 replications of panels with cross-section size $n=1000$. Reported  outer bounds   are the cross-replication averages.}
	\label{fig:idset_rcd}
\end{figure}
We consider 1000 replications where $n=1000$ and $T\in \left\{ 3,5,8,10 \right\}$, and vary $A_2$ between $-2$ and $2$. As in the static logit variant of this model, the supports of $A_{1,i}$ and $A_{2,i}$ are approximated by   grids of 50 equidistant points between $-5$ and $5$, and between $-7$ and $7$, respectively. Figure~\ref{fig:idset_rcd} reveals that the identified set can be quite wide. This is in line with the earlier observations for the random coefficient static logit model. However, the identified set becomes wider as $A_2$ increases. This is similar to the asymmetry observed in the dynamic logit case. When $A_2$ is large, $Y_i$ is more likely to be a vector of 1s. Hence, again, the effect of $A_{2,i}Y_{i,t-1}$ is hard to distinguish from that of $A_{1,i}$. In results not reported here, we observed that the outer bounds tend to be conservative in this particular case when the uniform linear program \eqref{eq:UnifOptimization} is used. We therefore used the baseline linear program which utilizes \eqref{eq:linprog1} in obtaining the bounds reported in Figure \ref{fig:idset_rcd}. The resulting outer bounds perform well in tracking the identified set as $T$ increases. 

\newpage

\clearpage
\section{Models with set-identified structural parameters}\label{sec:setidentifiedbeta}

Our standard approach assumes that the structural parameters are point-identified. This assumption covers a large variety of models that are used in empirical analysis (e.g. binary, count data, ordered choice and multinomial choice) based on logit or Poisson specifications. Nevertheless, other models of interest with set-identified structural parameters are also used in the empirical literature. The classical example is the probit model---but even some logit-based models have set-identified common parameters; see, e.g., \citet{DGK24}. In this part, we propose two approaches for obtaining outer bounds when the structural parameters are set-identified. To fix ideas, let the data generating process be given by
\begin{align}
    Y_{it} = 1\{ X_{it}\beta + A_i \geq \varepsilon_{it} \},
    \label{eq:dgpsetid}
\end{align}
where $\beta$ is potentially set-identified. 
To facilitate the discussion, henceforth we will assume that $X_{it}$ is scalar although all our arguments will be valid for vector-valued covariates, as well.

\subsection{First approach: random coefficients specification}\label{sec:seta1}
A straightforward approach to obtaining outer bounds when $\beta$ is set-identified is to use the random coefficient model. To see why, notice that
\begin{align}
    Y_{it} = 1\{ X_{it}\beta_i + A_i \geq \varepsilon_{it} \},
    \label{eq:rcsetid}
\end{align}
contains \eqref{eq:dgpsetid} as a special case and is, therefore, a valid specification even when the DGP is given by \eqref{eq:dgpsetid}. Consequently, we can use \eqref{eq:rcsetid} to obtain valid outer bounds. The advantage of this approach is that it side-steps the issue of set-identified $\beta$: this is because the random coefficient approach treats $\beta_i$ as a fixed effect, and is therefore completely agnostic to set-identification of $\beta$. The implementation of this idea is straightforward: one just has to obtain the outer bounds on the average effect of interest by treating \eqref{eq:rcsetid} as the DGP. Estimation of outer bounds under a random coefficient specification has already been considered earlier; indeed, since there are no structural parameters to be estimated, the theory developed in Theorem \ref{th:ConsistencyKNOWN} directly applies.

This approach has the advantage of being generically applicable to many different models. However, depending on the specific data generating process, the random coefficient variant can result in conservative outer bounds.  Nevertheless, especially in complicated models with arbitrarily set-identified parameters or when the identified set for $\beta$ is already suspected to be quite wide, this approach provides a viable alternative.

\subsection{Second approach: using the identified set for $\beta$}\label{sec:seta2}
A second and more principled approach is to incorporate the identified set for $\beta$ into the estimation of outer bounds. This can be done easily by using our baseline approach for obtaining bounds. Specifically, let $\mathcal {B} _{\mathrm{id}}$ be the identified set for $\beta_0$. For any $\beta \in \mathcal B _{\mathrm{id}}$, one can solve the linear program
\begin{align}
\min_{\ell,u \, : \,  \mathcal{Y} \rightarrow \mathbb{R} }\,  &Q(\ell(\cdot),u(\cdot),z,\beta)
\notag 
\\
  \text{subject} & \text{ to} \notag  
\\
&   \forall y\in \mathcal{Y}:\, \,b_{\min }\leq \ell(y)\leq 
u(y)\leq b_{\max }  \notag 
\\
\quad \text{and}\quad & \forall a\in \mathcal{A}:\, \, \sum_{y\in \mathcal{Y}}\ell (y) \, f(y \, | \, z, a;\beta )
 \leq m(z,a,\beta )\leq \sum_{y\in \mathcal{Y}}u(y) f(y \, | \, z, a;\beta ),    \notag
\end{align}
and obtain the solutions $\widehat L(\beta)$ and $\widehat U(\beta)$ in the same way as in Section \ref{sect:boundconstruction}. The outer bounds that incorporate the identified set for $\beta_0$ can then be obtained by searching for the minimum lower bound and maximum upper bound across all $\beta \in \mathcal B _{\mathrm{id}}$; that is
\begin{align}
    \widehat L (\mathcal{B}_{\mathrm{id}}) = \min_{\beta \in \mathcal B _{\mathrm{id}}} \widehat L (\beta)
    \qquad
    \text{and}
    \qquad
    \widehat U (\mathcal{B}_{\mathrm{id}}) = \max_{\beta \in \mathcal B _{\mathrm{id}}} \widehat U (\beta)
    \label{eq:boundswithid}
\end{align}
will yield the final lower and upper outer bound estimators, respectively. 

In implementation one first has to estimate the identified set $\mathcal B _{\mathrm{id}}$ and also use a suitable grid to approximate this estimated set.  To conduct inference, suppose that $\widehat{\mathcal B}_{1-\nu}$ is a confidence set for $\mathcal B_{\mathrm{id}}$ with asymptotic coverage at least $1-\nu$. One can then apply Theorem \ref{theorem:perturbed} by replacing $\mathcal B_{1-\nu}$ with $\widehat{\mathcal B}_{1-\nu}$: compute the bounds $\widehat L(\beta)$ and $\widehat U(\beta)$ together with their standard errors for each $\beta$ on a grid approximating $\widehat{\mathcal B}_{1-\nu}$, and then take the infimum and supremum as in Theorem \ref{theorem:perturbed}. The resulting confidence interval will have asymptotic coverage of at least $1-\alpha-\nu$. Estimating $\mathcal B _{\mathrm{id}}$ and constructing confidence sets for partially identified parameters is outside the scope of this paper; see, for example, \citet{CFHN13} and \citet{ChernozhukovHongTamer(07)} for general methods.

\newpage

\section{Outer bounds for the probit model}\label{sec:appendixC}
The probit model is routinely used in empirical and theoretical studies but, unlike the logit model, its model parameters are set-identified when $T$ is fixed. This makes it more challenging to obtain the identified set and the outer bounds on average effects. In this part, we use the probit model to evaluate the two approaches introduced in Section \ref{sec:setidentifiedbeta}. In particular, in Section \ref{sec:approb1} we employ the approaches of Sections \ref{sec:seta1} and \ref{sec:seta2} to obtain outer bounds for the static probit model. This is followed by Section \ref{sec:approb2} where we compare the coverage of the outer bounds against that of the identified set, similar to the analysis of Section \ref{subsec:Comparision1} for the logit model.

\subsection{Obtaining outer bounds on the static probit model}\label{sec:approb1}
We consider the static probit model
\begin{align}
   Y_{it} = 1\{ X_{it}\beta + A_i \geq \varepsilon_{it} \}
   \label{eq:probitdgp}
\end{align}
where
$\varepsilon_{it}\sim N(0,1)$, $A_i \sim N(0,1)$, $X_{it}=1\{A_i \geq \eta_{it} \}$ and $\eta_{it}\sim N(0,1)$. Our interest is in obtaining bounds on the average effect
\begin{align}
 \mathbb E \{ 
    \mathbb E[ Y_{it} | X_{it}=1, A_i,\beta_0]
    -
    \mathbb E[ Y_{it} | X_{it}=0, A_i,\beta_0]
 \}.
 \label{eq:aeprob1}
\end{align}

We first use the random coefficient approach of Section \ref{sec:seta1} to obtain valid outer bounds on \eqref{eq:aeprob1}. We do this by obtaining outer bounds for the model
\begin{align}
    Y_{it} = 1\{ X_{it}\beta_i + A_i \geq \varepsilon_{it} \},
    \qquad
    \varepsilon_{it}\sim N(0,1),
    \label{eq:rcprobit}
\end{align}
which contains the probit model of \eqref{eq:probitdgp}. For each value of $\beta_0\in[-2,2]$, the simulation results are based on 1000 replications where $n=1000$ and $T\in\{2,3,5\}$. Outer bounds are obtained using the random coefficient specification given in \eqref{eq:rcprobit}. In obtaining these bounds, the support of $A_{i}$ is approximated by a grid of 50 equidistant points between $-5$ and $5$; similarly a grid of 50 equidistant points between $-7$ and $7$ is used to approximate the support of $\beta_i$. The results of this analysis are presented in Figure \ref{fig:rcprobit}.
\begin{figure}[p]
	\includegraphics[width=1\linewidth]{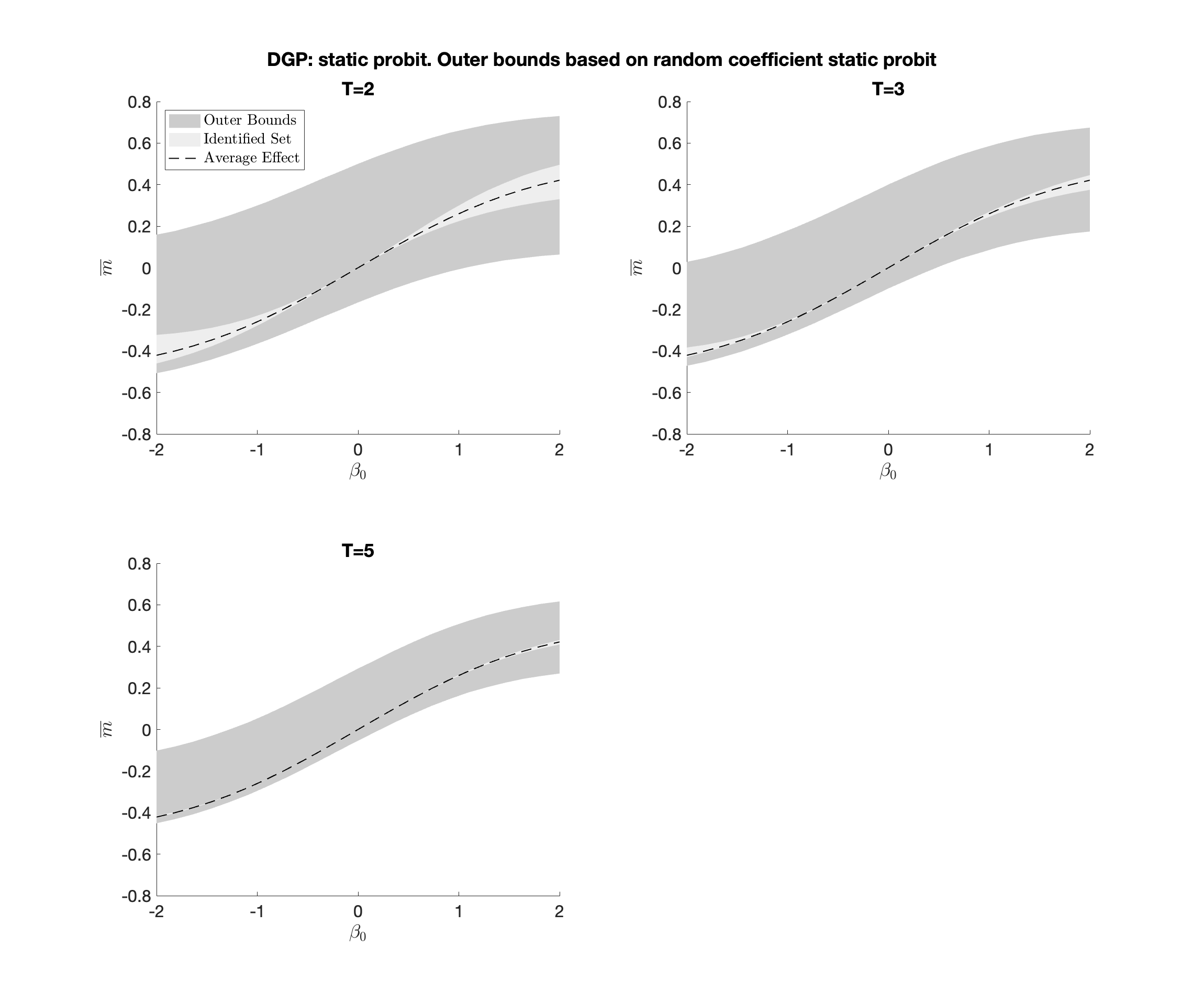}
	\caption{Simulation results for the static probit model in equation \eqref{eq:probitdgp}. The average effect of interest is  
    $\mathbb E \{ 
        \mathbb E[ Y_{it} | X_{it}=1, A_i]
        -
        \mathbb E[ Y_{it} | X_{it}=0, A_i]
     \}.
     $
     Outer bounds are obtained using the random coefficient static probit model defined in equation \eqref{eq:rcprobit}. Based on 1000 replications of panels with cross-section size $n=1000$.  Outer bounds are obtained by the linear program in \eqref{eq:UnifOptimization}. Reported outer bounds are cross-replication averages.}
	\label{fig:rcprobit}
\end{figure}	

As expected, the outer bounds are valid but also wider than the identified set, as the model in \eqref{eq:rcprobit} is much more flexible than \eqref{eq:probitdgp}. Moreover, the lack of point-identification in the static probit specification is known to be mild (see, e.g., Figures 2 and 3 in \citealt{CFHN13}). Consequently, using the random coefficient approach to side-step set-identification of $\beta$ turns out to be a valid but conservative option in this particular case.
\begin{figure}[p]
\begin{centering}
	\includegraphics[width=1\linewidth]{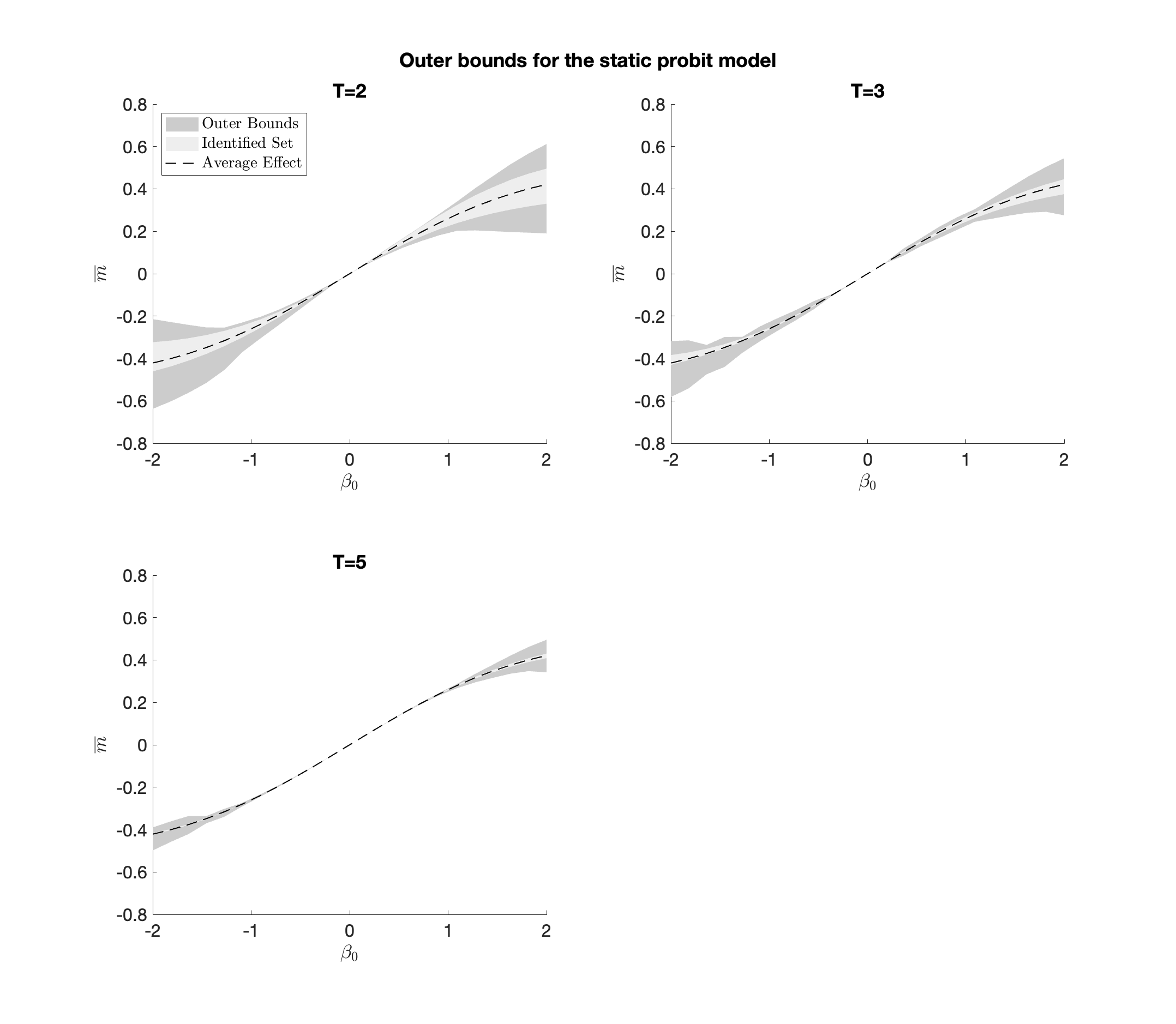}
	\caption{Simulation results for the static probit model in equation \eqref{eq:probitdgp}. The average effect of interest is  
    $\mathbb E \{ 
        \mathbb E[ Y_{it} | X_{it}=1, A_i]
        -
        \mathbb E[ Y_{it} | X_{it}=0, A_i]
     \}.
     $
     Outer bounds are obtained using the method of Section \ref{sec:seta2} which incorporates the fact that $\beta$ is not point-identified. Based on 1000 replications of panels with cross-section size $n=1000$.  Outer bounds are obtained by the linear program in \eqref{eq:UnifOptimization}. Reported outer bounds are cross-replication averages.}
     \label{fig:probit}
\end{centering}
\end{figure}	

Next, we consider the approach of Section \ref{sec:seta2} and obtain outer bounds that incorporate the identified set for $\beta$. To isolate away the estimation error arising from estimation of the identified set, in this part we use the sharp identified set for $\beta$.\footnote{The difference lies in whether one uses the `true' choice probabilities $P(Y_i|X_i)$ implied by the DGP or the estimated choice probabilities implied by data. We use the former to obtain the identified set for $\beta$. See Section \ref{sec:t10p1} for the algorithm used in obtaining the identified set.} Our simulation results, presented in Figure \ref{fig:probit}, are based on 1000 replications for panels with $n=1000$ and $T\in \{2,3,5\}$. The support for $A_i$ is approximated by a grid of 100 equidistant points between $-5$ and $5$. The results reveal that the outer bounds perform quite well. In particular, the width of outer bounds decreases with $T$ and the outer bounds generally collapse to the average effect whenever the average effect is point-identified. All these results echo our findings for the static logit model. Hence, the approach of Section \ref{sec:seta2} stands as a viable and computationally feasible alternative to obtaining the sharp identified set on the average effect when $\beta$ is set-identified.

\subsection{Comparison to the identified set}\label{sec:approb2}

In this part, we repeat the analysis of Section \ref{subsec:Comparision1} for the static probit model with a discrete (but non-binary) covariate. We use the same DGP as before, except that the error terms are now normally distributed:
\begin{gather*}
    Y_{it} = 1\{ X_{it}\beta + A_{i} \geq \varepsilon_{it}\},
    \quad
    A_{i} \sim N\left( \frac{1}{T} \sum_{t=1}^T X_{it} - \frac{1}{2}, 1 \right),
    \quad
    X_{it} = x_{it}/(|\mathcal X|-1),
\end{gather*}
where $\varepsilon_{it}\sim N(0,1)$ and $x_{it}$ is discrete uniform with support $[0,|\mathcal X| -1]$. The aim of this part is to investigate the sensitivity of inference for the sharp identified set (for the average effect) on the estimation of choice probabilities $P(Y_i|X_i)$ when $\beta$ is set-identified. We have already seen in the static logit case of Section \ref{subsec:Comparision1} that the nonlinear dependence of the identified set on choice probabilities leads to misleading inference results as the dimensionality of the choice probability estimation problem increases. We now extend this analysis to the static probit case which is subject to the additional element of set-identified $\beta$.

As before, we consider two cases $|\mathcal X|\in \{6,12\}$, and focus on the average effect
$\mathbb E \{ 
    \mathbb E[ Y_{it} | X_{it}=1, A_i]
    -
    \mathbb E[ Y_{it} | X_{it}=0, A_i]
    \}.
$
The simulation results are based on 1000 replications of panels with dimensions $T=2$ and $n=200$.
The outer bounds are obtained using the method of Section \ref{sec:seta2}. Details on the estimation of the identified set can be found in Section \ref{sec:idmanual1} of the Supplementary Appendix. The reported upper and lower $2.5\%$ percentiles are the corresponding sample percentiles across replications.

\begin{figure}[htb]
    \centering
    \begin{minipage}{0.5\textwidth}
        \centering
        \includegraphics[width=1\textwidth]{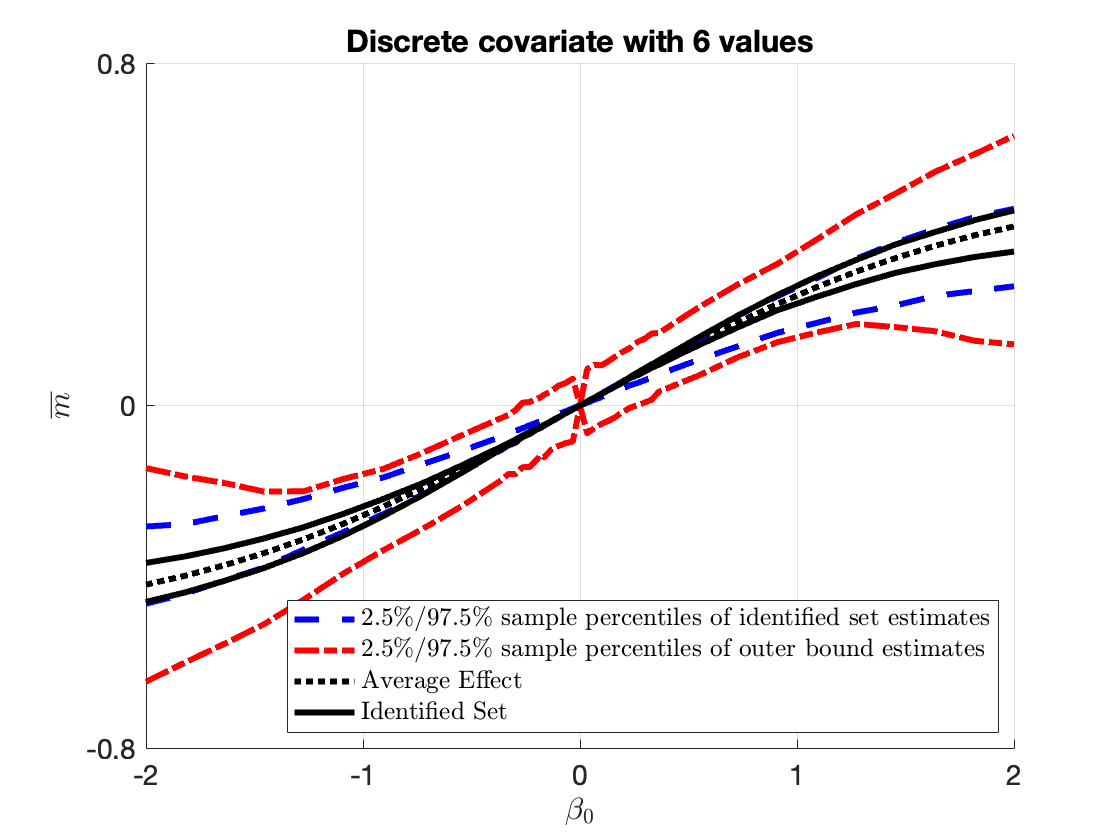}
    \end{minipage}\hfill
    \begin{minipage}{0.5\textwidth}
        \centering
        \includegraphics[width=1\textwidth]{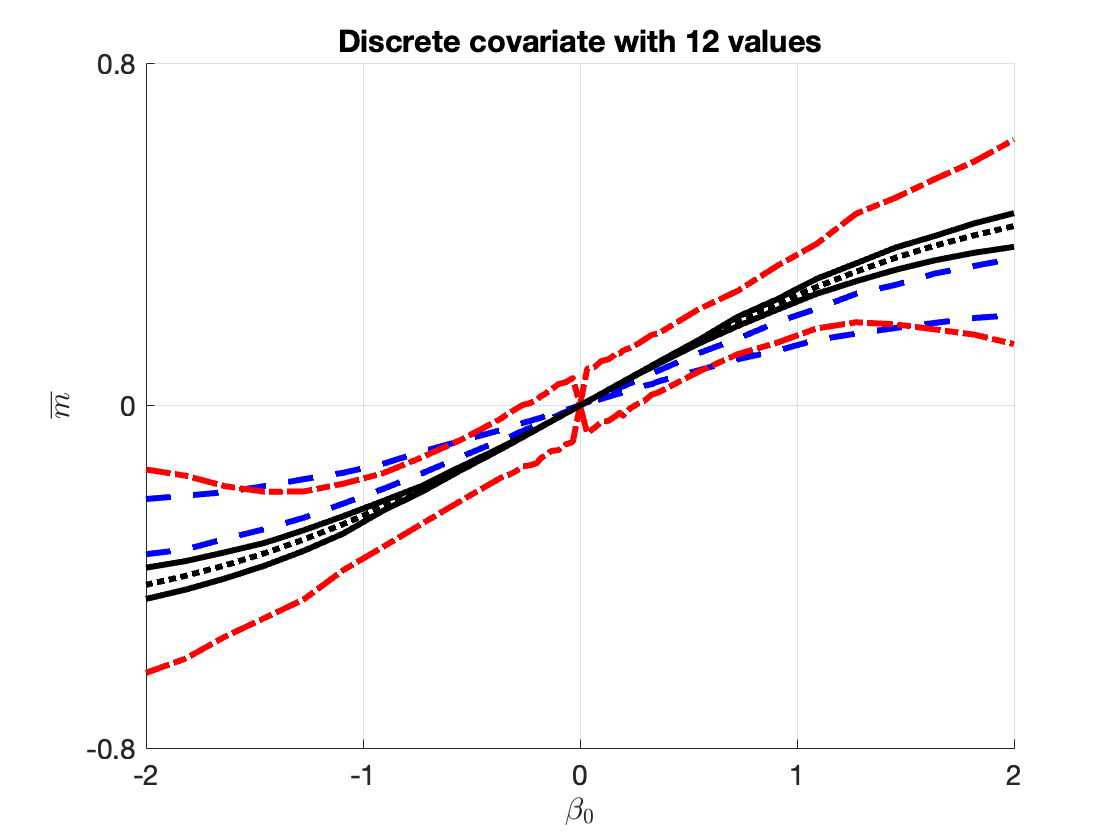}
    \end{minipage}
    \caption{Sample quantiles of estimates of the outer bounds and the identified set. The DGP is $Y_{it} = 1\{ X_{it}\beta + A_{i} \geq \varepsilon_{it}\}$ where 
    $A_{i} \sim N( T^{-1} \sum_{t=1}^T X_{it} - 1/2, 1 )$, $X_{it} = x_{it}/(|\mathcal X|-1)$, $x_{it}$ is discrete uniform with support $[0,|\mathcal X| -1]$, and $\varepsilon_{it}\sim N(0,1)$. The average effect under consideration is $\mathbb{E}[Y_{it} \, \big| \, X_{it}=1,A_i,\beta_0] 
    -  \mathbb{E}[Y_{it} \, \big| \,X_{it}=0,A_i,\beta_0]$. For each $\beta_0 \in [-2,2]$, the quantiles are calculated across 1000 replications of panels with $T=2$ and $n=200$. The left panel contains the results for $|\mathcal X| = 6$ whereas the results for $|\mathcal X| =12$ are presented in the right panel.}
    \label{fig:idanalysis_probit}
\end{figure}

The results for this analysis, presented in Figure \ref{fig:idanalysis_probit}, are generally qualitatively similar to the results obtained for the logit model in Figure \ref{fig:comb1}. In particular, as in the logit case, inference based on the sharp identified set is affected by the dimensionality of choice probabilities. We observe this even when $|\mathcal X|$ increases from 6 to 12; in the latter case, the $2.5\%/97.5\%$ percentiles of the sample analogue of the identified set completely miss the average effect. The same is not observed for the outer bound percentiles. This is all analogous to the logit analysis --- the only difference is that the sample quantiles in the probit case do not collapse towards the population average effect as $\beta_0$ approaches zero, except for when $\beta_0=0$. This is likely caused by $\beta$ being set-identified in the probit setting.

\newpage

\clearpage
\section{Comparison to the \citet{DDL21} outer bounds}\label{sec:ddlvspw}

In a recent paper, \citet{DDL21}---referred to as DDL henceforth---propose a novel approach for conducting inference on average effects. Their main approach focuses on the sharp identified set. However, they also propose a second approach which essentially yields ``outer bounds'' that are different than ours. 

In this part, we provide a comparison of our outer bounds to theirs, in the case of a static logit model. In particular, we consider 
\begin{align*}
    Y_{it} = 1\left\{ X_{it}\beta + A_i \geq \varepsilon_{it} \right\},    
\end{align*}
where 
$\varepsilon_{it} \sim {\rm Logit} (0,1)$, $A_i \sim N(0,1)$, $X_{it} = 1\left\{ A_i \geq \eta_{it} \right\}$ and $\eta_{it} \sim N(0,1)$. The average effect of interest is
\begin{align*}
    \mathbb E \{ 
    \mathbb E[ Y_{it} | X_{it}=1, A_i, \beta_0]
    -
    \mathbb E[ Y_{it} | X_{it}=0, A_i, \beta_0]
    \}.    
\end{align*}
We replicated 1000 panels with cross-section size $n=5000$ and various $T$. To obtain our outer bounds (which we call the PW bounds), we use the linear program of Section \ref{sec:uniflp}; the unknown $\beta$ is estimated by the  conditional likelihood method.
For the confidence intervals we use the method of Section \ref{sec:perturbed} with $\nu=0.0001$, and we also approximate $\mathcal B_{1-\nu}$ by a grid of 5000 equidistant points. Outer bounds by the DDL method (which we call the DDL bounds) and also the confidence bands are obtained using the \texttt{MarginalFElogit} package in R (\citealt{marginalfelogit}), which is the official R package for the DDL method. All reported outer bound and confidence bands are cross-replication averages.

The top row of Figure \ref{fig:ddlpw1} compares the outer bounds by the two methods. Clearly, both methods perform equally well. The DDL bounds become substantially wide as $|\beta_0|$ increases, especially on the negative side---however, this is most likely caused by some numerical sensitivity in the estimation package. Therefore, in our comparison we will focus on a limited range of $\beta_0$. The bottom row of Figure \ref{fig:ddlpw1} presents the differences between the widths of the PW and DDL outer bounds. A positive value corresponds to cases where the DDL bounds are narrower than the PW bounds. Negative values, on the other hand, occur when the PW bounds are narrower than the DDL bounds. These figures clearly show that both methods perform almost identically well. 

Moving next to the comparison of the 95\% confidence bands, presented in the top row of Figure \ref{fig:ddlpw2}, we observe that the PW confidence bands are slightly wider compared to the DDL bands. However, the difference is generally modest as documented in the bottom row of Figure \ref{fig:ddlpw2}, which shows the differences between the widths of the PW and DDL bands. In particular, the difference decreases with $T$ and is almost always below 0.05. We finally note that, as mentioned before, the performance of the DDL bands deteriorates quickly as $|\beta_0|$ increases towards 1, and especially for negative $\beta_0$;  however, this is most likely due to a numerical sensitivity issue in the estimation package.

All in all, the results of this analysis confirm that the PW and DDL outer bounds perform comparably well, which underlines the reliability of our approach.

\begin{figure}[H]
  \centering
  \begin{subfigure}[t]{0.9\textwidth}
    \centering
    \includegraphics[width=\textwidth]{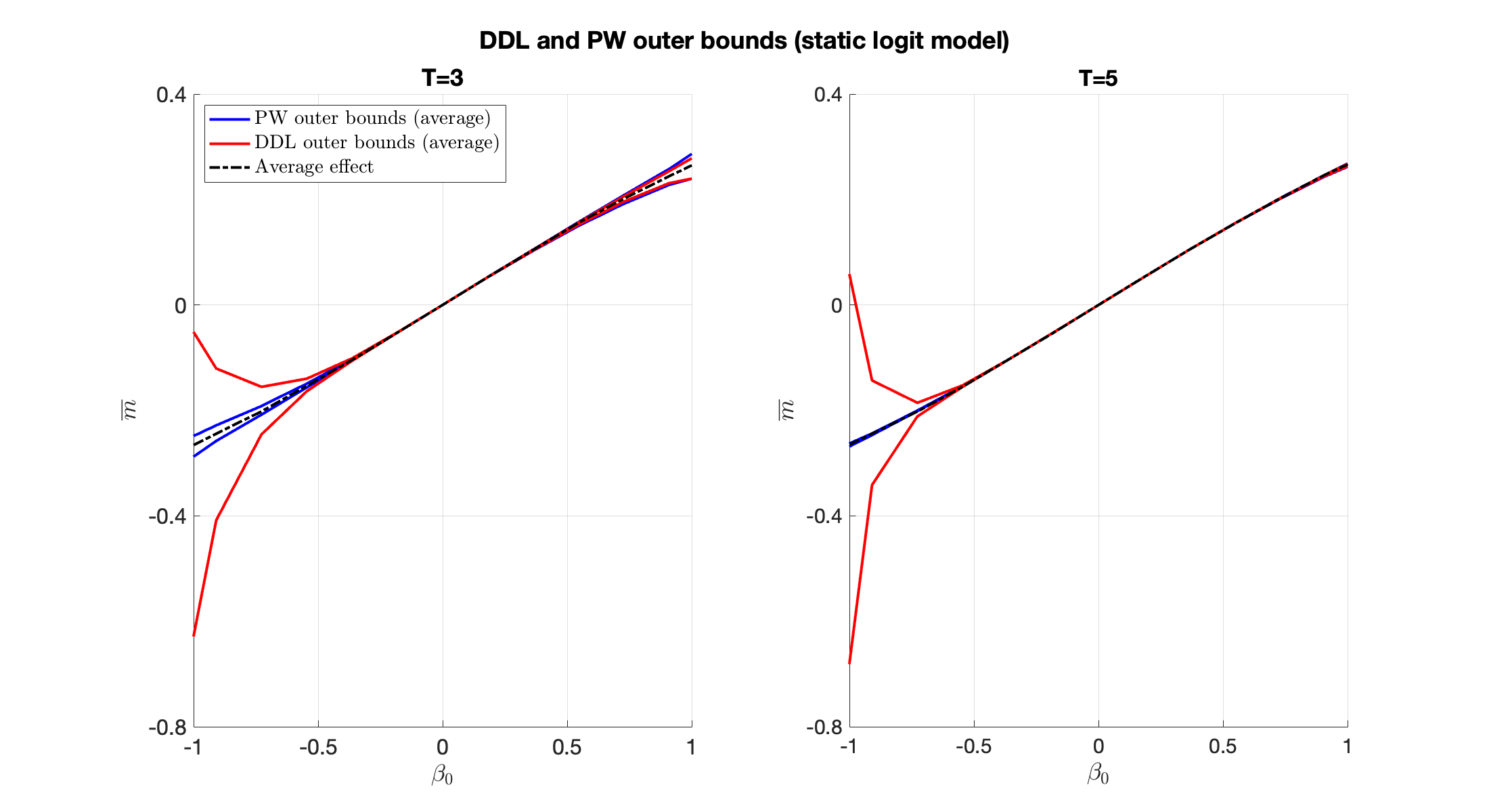}
  \end{subfigure}

  \vspace{1em}

  \begin{subfigure}[t]{0.9\textwidth}
    \centering
    \includegraphics[width=\textwidth]{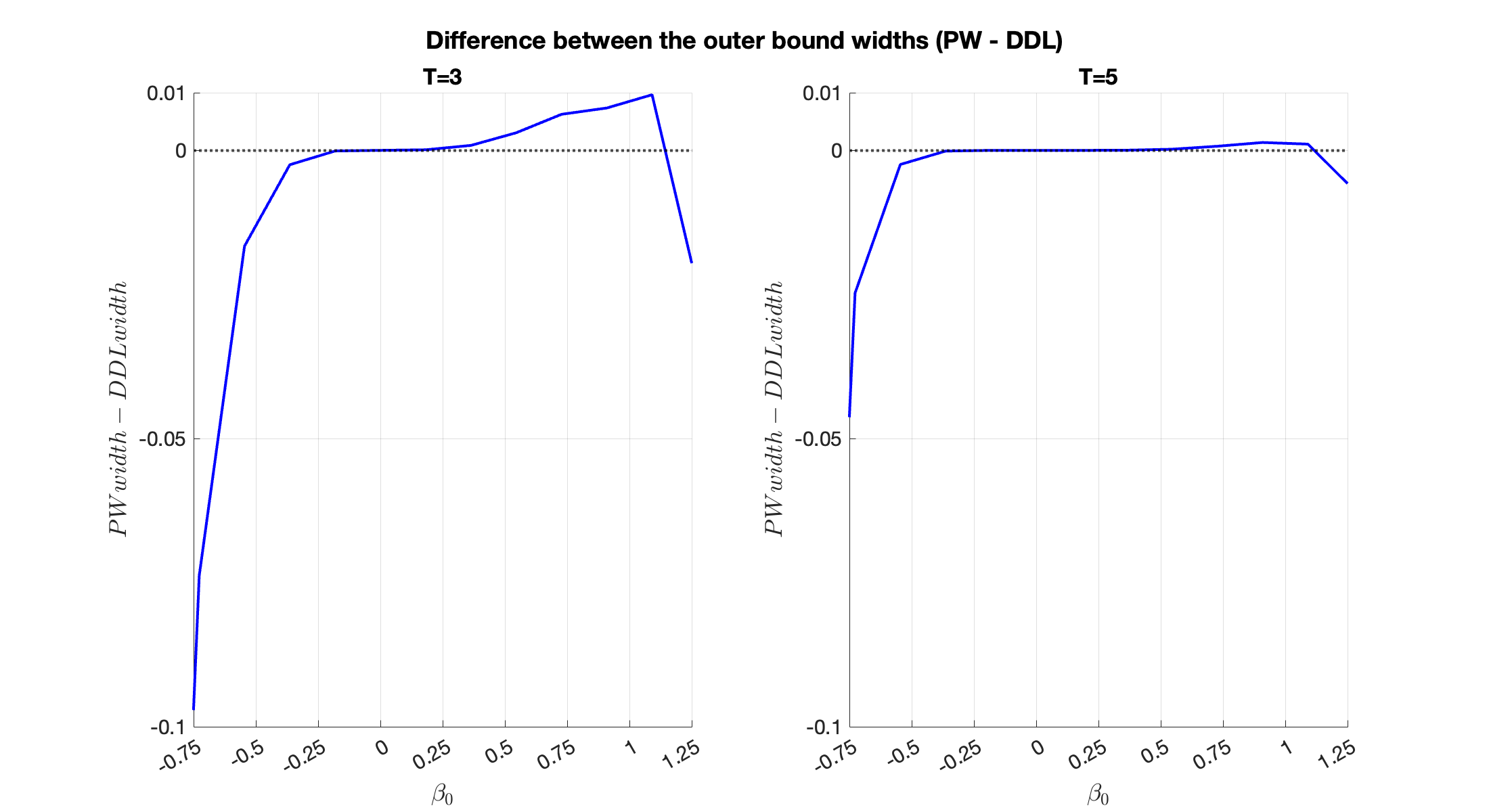}
  \end{subfigure}

  \caption{Comparison of the DDL and PW outer bounds. The DGP is $Y_{it} = 1\left\{ X_{it}\beta + A_i \geq \varepsilon_{it} \right\}$ where $\varepsilon_{it} \sim {\rm Logit} (0,1)$, $A_i \sim N(0,1)$, $X_{it} = 1\left\{ A_i \geq \eta_{it} \right\}$ and $\eta_{it} \sim N(0,1)$. The average effect under investigation is $\mathbb E \{ 
    \mathbb E[ Y_{it} | X_{it}=1, A_i, \beta_0]-\mathbb E[ Y_{it} | X_{it}=0, A_i, \beta_0]\}.$ Results for the DDL method are obtained using the \texttt{MarginalFElogit} package in R. For information on the PW bounds see Section \ref{sec:ddlvspw}. All presented bound results are cross-replication averages.}
  \label{fig:ddlpw1}
\end{figure}

\begin{figure}[H]
  \centering
  \begin{subfigure}[t]{0.9\textwidth}
    \centering
    \includegraphics[width=\textwidth]{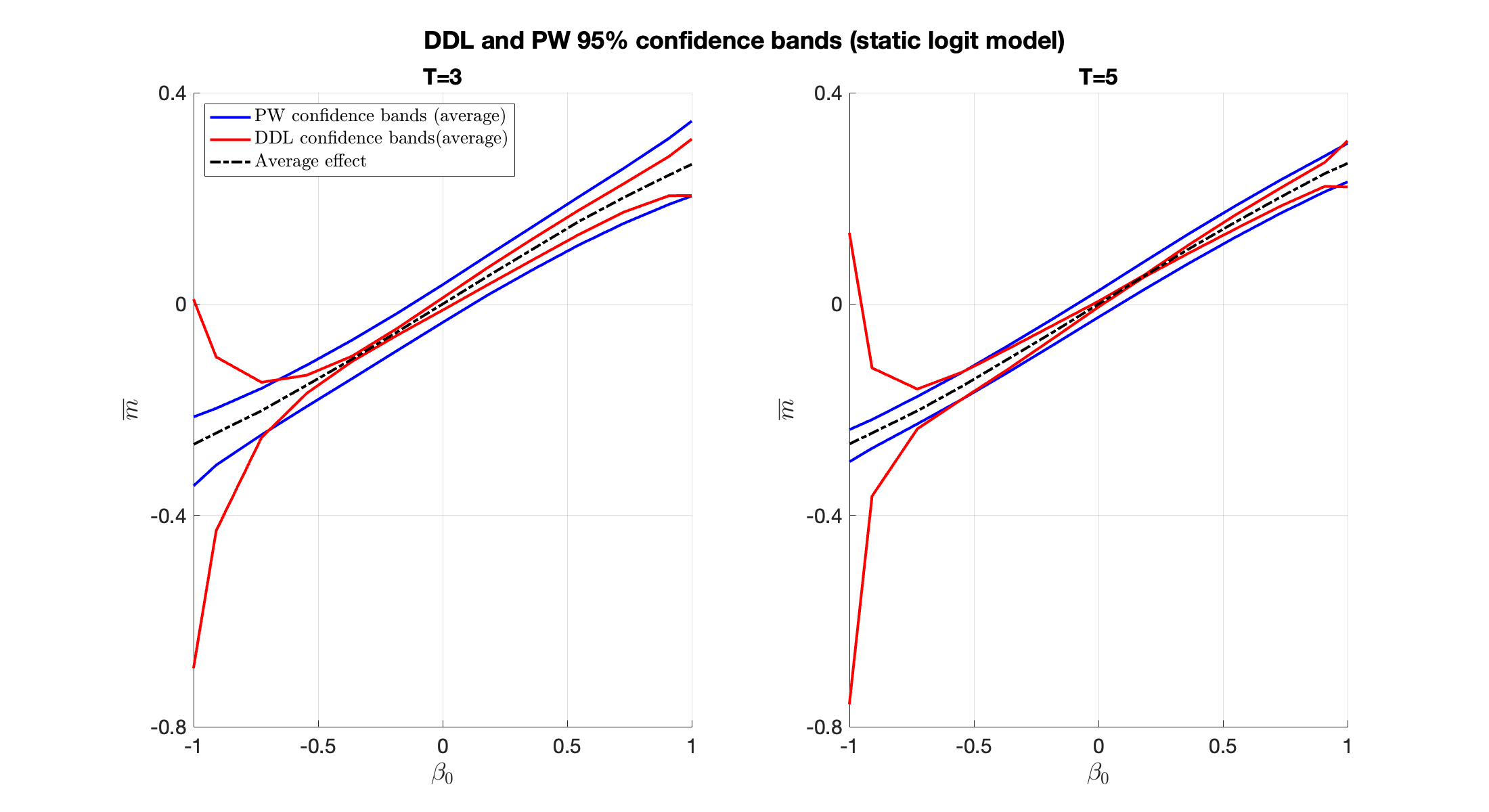}
  \end{subfigure}

  \vspace{1em}

  \begin{subfigure}[t]{0.9\textwidth}
    \centering
    \includegraphics[width=\textwidth]{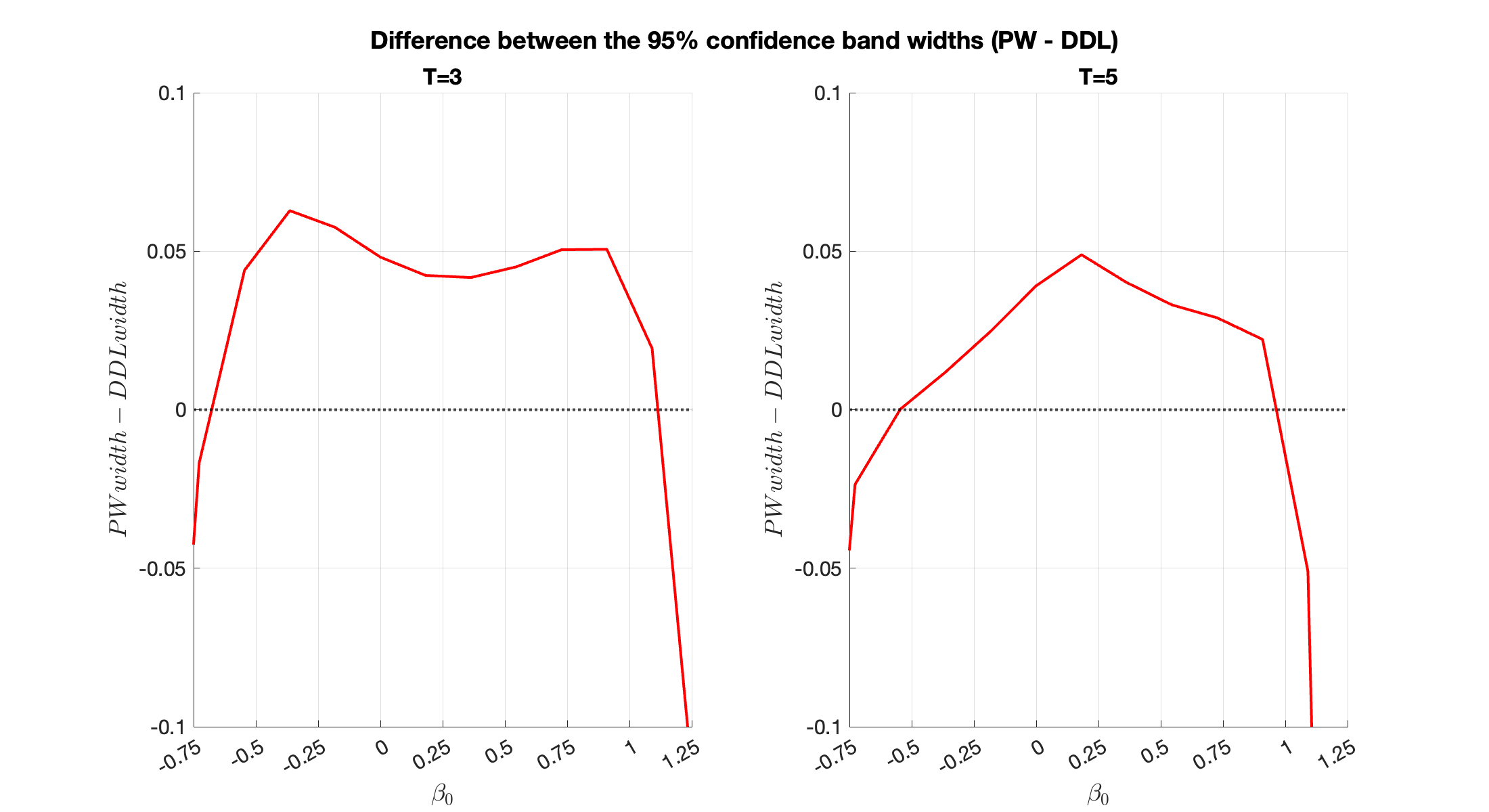}
  \end{subfigure}

  \caption{Comparison of the DDL and PW confidence bands (for the DDL and PW outer bounds). The DGP is $Y_{it} = 1\left\{ X_{it}\beta + A_i \geq \varepsilon_{it} \right\}$ where $\varepsilon_{it} \sim {\rm Logit} (0,1)$, $A_i \sim N(0,1)$, $X_{it} = 1\left\{ A_i \geq \eta_{it} \right\}$ and $\eta_{it} \sim N(0,1)$. The average effect under investigation is $\mathbb E \{ 
    \mathbb E[ Y_{it} | X_{it}=1, A_i, \beta_0]-\mathbb E[ Y_{it} | X_{it}=0, A_i, \beta_0]\}.$ Results for the DDL method are obtained using the \texttt{MarginalFElogit} package in R. For information on the PW bounds see Section \ref{sec:ddlvspw}. All presented bound results are cross-replication averages.}
  \label{fig:ddlpw2}
\end{figure}

\newpage

\section{Computational and practical details}
\label{sec:compdetails}

\subsection{Implementational details}\label{sec:appimp}

In practice, we usually cannot solve the linear programs \eqref{eq:GeneralOptimization} and \eqref{eq:UnifOptimization} exactly. This is because the functions \eqref{eq:linprog1} and \eqref{eq:s} require evaluation over $\mathcal{A}$ which  typically  has infinite cardinality. Instead, we approximate these objects by choosing a subset of grid points $\mathcal{A}_g \subset \mathcal{A}$ and imposing 
the constraints only at $a \in \mathcal{A}_g$. This yields
\begin{align*}
& Q(\ell(\cdot),u(\cdot), z, \beta) = \sum_{a \in \mathcal{A}_g} \,  \sum_{y\in \mathcal{Y}}\left[ u(y)-\ell(y) \right] \,f(y \, | \, z,a;\beta )\,p(a|z)\, ,
\end{align*}
and
\begin{align*}
    s
    =
    \max_{a\in \mathcal{A}_g} \left[ \sum_{y\in \mathcal{Y}}\left[ u(y)-\ell(y) \right] \,f(y \, | \, z,a;\beta ) \right] .
\end{align*}

The size of the grid $\mathcal{A}_g$ directly controls the number of restrictions in \eqref{eq:GeneralOptimization} and  \eqref{eq:UnifOptimization}, and so, especially in complicated applications, computational concerns may put a limit on how fine the grid $\mathcal{A}_g$ can be. However, even then, our approach provides an easy way to obtain solutions that work on a much finer grid. To illustrate how, let $\mathcal{A}_g,\mathcal{A}_G \subset \mathcal{A}$ be two grids where the cardinality of $\mathcal{A}_G$ is 
(much) larger than that of $\mathcal{A}_{g}$. Let $(L(z,y,\beta),U(z,y,\beta))$ be the solution to, say, \eqref{eq:GeneralOptimization} on $\mathcal{A}_g$. It is computationally almost cost-free to check whether this solution satisfies the restriction \eqref{eq:boundcondition} on $\mathcal{A}_G$. If violations occur, one can adjust the original solution $(L(z,y,\beta),U(z,y,\beta))$ to fit the restriction \eqref{eq:boundcondition} on $\mathcal{A}_G$, thereby obtaining a valid solution to the constraints on the much finer  grid $\mathcal{A}_G$. A simple way to do this is to replace the original solution by $(\dot{L}(z,y,\beta),\dot{U}(z,y,\beta))$ where
\begin{align*}
    & \dot{L}(z,y,\beta) := L(z,y,\beta) + \min \left\{ 0, \textstyle{\min_{a \in \mathcal{A}_G} }
    \left[m(z,a,\beta) - \textstyle{\sum_{y\in \mathcal{Y}}} L(z,y,\beta) f(y|z,a;\beta)\right] \right\}, \\
    & \dot{U}(z,y,\beta) := U(z,y,\beta) + \max\left\{0,\textstyle{\max_{a \in \mathcal{A}_G} } \left[m(z,a,\beta) - \textstyle{\sum_{y\in \mathcal{Y}}} U(z,y,\beta) f(y|z,a;\beta)\right]\right\}.
\end{align*}
Here, we simply add the maximum deviation across all grid points to the original solution, automatically yielding a solution that satisfies \eqref{eq:boundcondition} on $\mathcal{A}_G$. The grid $\mathcal{A}_G$ can be made very fine, making the difference between $\mathcal{A}_G$ and $\mathcal A$ negligible. Of course, depending on the case at hand, one can devise options that yield less ``conservative'' solutions compared to $(\dot{L}(z,y,\beta),\dot{U}(z,y,\beta))$. We suggest comparing results from different selections of  $(\mathcal{A}_g,\mathcal{A}_G)$ in order to find some $\mathcal{A}_g$ which is computationally feasible and yet yields reliable bounds.

Computational properties will also depend on the cardinality of $\mathcal{Y}$, which determines the number of variables in the  programs \eqref{eq:GeneralOptimization} and \eqref{eq:UnifOptimization}. Although the cardinality of $\mathcal{Y}$ is finite in a fixed-$T$ setting, the size of the support can still be large enough to cause computational difficulties. In logit-based applications it is relatively straightforward to mitigate this problem, by re-writing, e.g., the restriction \eqref{eq:boundcondition} in terms of the conditional density of the sufficient statistic. We next illustrate this for Examples \ref{ex:statlogit} and \ref{ex:rclogit}.

\setcounter{example}{0}

\begin{example}[continued]
Fix $z$ and $\beta$. Let $y=(y_1,...,y_T)'$ and define  
\begin{equation*}
    \overline{P}(k|z,a,\beta) := \sum_{ \{y: \sum_t y_t = k \} } P(y|z,a,\beta),
\end{equation*} 
the conditional density of the sufficient statistic $\sum_{t=1}^T y_t$. Then, there is some $\overline{u}(k)$ such that
\begin{equation*}
	m(z,a,\beta) \leq \sum_{y\in\mathcal Y} u(y)P(y|z,a,\beta) = \sum_{k=0}^T \overline{u}(k)\overline{P}(k|z,a,\beta), \quad \forall a \in \mathbb{R}.
\end{equation*}
As such, one can solve the problem for $\overline{u}(k)$, $k=0,...,T$, and then use, for instance, $u(y)=\overline{u}(k)$ for all $y$ with $\sum_{t=1}^T y_t = k$. This effectively decreases the number of variables from $2^T$ to $T+1$. An analogous argument applies to $\ell(y)$.
\end{example}

\begin{example}[continued]
	Suppose that $z_t$ is binary. In this case there are two unobserved effects, and so the argument will be based on the sufficient statistics $\sum_{t=1}^T y_t$ and $\sum_{t=1}^T y_t z_t$. Fix $z=(z_1,...,z_T)'$. Define $\overline{P} (k_1, k_2 | z ,a_1 ,a_2)$, the conditional density of $k_1=\sum_{t=1}^T y_t$ and $k_2=\sum_{t=1}^T y_t z_t$. Then, similar to Example \ref{ex:statlogit}, there is some $\overline{u}(k_1,k_2)$ such that
	\begin{eqnarray*}
		m(z,a_1,a_2) 
		&\leq & \sum_{y\in\mathcal Y} u(y)P(y|z,a_1,a_2) 
		\\
		&=& \sum_{k_1=0}^T \sum_{k_2=0}^T \overline{u}(k_1,k_2)\overline{P} (k_1, k_2 |z,a_1,a_2 ), \quad \forall a_1 \in \mathbb{R} \text{ and } a_2 \in \mathbb{R},
	\end{eqnarray*}
	implying that it is sufficient to solve the linear program for $\overline{u}(k_1,k_2)$. At first sight, it appears that the number of variables in this problem is $(T+1)^2$. However, notice that one cannot have  $k_2 > k_1$. Moreover, for a given $z$ some combinations of $(k_1,k_2)$ will have zero probability. Consequently, the actual number of variables will usually be less than $(T+1)^2$. We note this method will not work with continuous covariates, since in that case $z$ has infinite support and so does $k_2$.
\end{example}

An analogous idea for non-logit applications (where the outlined approach does not necessarily exist) is to reduce the dimension of the problem by partitioning the support $\mathcal Y$ using some meaningful criterion. One can, for example, partition $\mathcal Y$ such that $\arg \max_{a \in \mathcal{A}_g} P(y|z,a,\beta)$ is the same for all $y\in \mathcal Y$ in the same subset. Generation of the partition can also be based on extra information specific to the application or data at hand.

Another way to decrease the computational complexities is to solve the linear program separately for the upper and lower bounds. Note that \eqref{eq:GeneralOptimization}  and \eqref{eq:UnifOptimization} put the restriction $\ell(\cdot) \leq u(\cdot)$ to avoid any crossover between the upper and lower bounds. Solving the bound problem separately would 
drop this convenient additional condition. However, for moderately large $T$ the probability of such a crossover between the bounds is expected to be quite low, 
and
a potential solution of the crossover problem is, for example, given in
\cite{Stoye21}.

\subsection{Computation times}\label{sec:comptimes}

In Table \ref{tbl:comptimes} we present the average estimation times required for obtaining the outer bounds in each of the four models considered in Section \ref{sect:setcomparison} and Appendix \ref{app:dynamic}. The average estimation times are calculated by estimating the outer bounds once for each value of the main parameter ($\beta$, $A_2$ or $\gamma$) on a grid of 23 equidistant values between $-2$ and $2$, and then taking the average of the elapsed time across all grid points. This grid exactly corresponds to the grid of values we used in running the simulations. 

We note that the support of the outcome variable for the random coefficient static logit model depends on the structure of the sufficient statistic used in this exercise (see the information provided in Example 2).

We use Matlab for coding and estimation. All results reported in Table \ref{tbl:comptimes} were obtained on a MacBook Pro (Apple M1 Max chip (10 Cores) and 64 GB RAM capacity), without employing parallel computing tools.
\begin{table}[H]
  \centering
  \small
  \renewcommand{\arraystretch}{0.85}
  \begin{tabular}{@{} c c c r c c  @{}}
    \toprule
    \multicolumn{1}{c}{Model} & \multicolumn{1}{c}{Conditioning variables} & \multicolumn{1}{c}{T} & \multicolumn{1}{c}{\shortstack{Estimation time\\(seconds)}} & \multicolumn{1}{c}{$|\mathcal A _g|$} & \multicolumn{1}{c}{$|\mathcal Y|$} \\
    \midrule
    \multirow{6}{*}{\shortstack{Static logit\\(Figure \ref{fig:idset_logit})}} 
    & $X_{it}$ (cont.) & 2  & 3.41  & 100  & 4 \\
    & $X_{it}$ (cont.) & 3  & 3.88  & 100  & 8 \\
    & $X_{it}$ (cont.) & 5  & 7.08  & 100  & 32 \\
    & $X_{it}$ (disc.) & 2  & 3.27  & 100  & 4 \\
    & $X_{it}$ (disc.) & 3  & 3.81  & 100  & 8 \\
    & $X_{it}$ (disc.) & 5  & 6.05  & 100  & 32 \\
    \midrule
    \addlinespace
    \multirow{4}{*}{\shortstack{RC static logit\\(Figure \ref{fig:idset_rc})}} 
    & 
    \multirow{4}{*}{$X_{it}$ (disc.)}
    & 3  & 26.22 & 2500 & 4 to 6 \\
    & & 5  & 49.36 & 2500 & 6 to 12 \\
    & & 8  & 118.59& 2500 & 9 to 25 \\
    & & 10 & 419.23& 2500 & 11 to 36 \\
    \midrule
    \addlinespace
    \multirow{3}{*}{\shortstack{Dynamic logit\\(Figure \ref{fig:idset_dlogit})}} 
    &
    \multirow{3}{*}{$Y_{i,t-1}$ \& $X_{it}$ (cont.)} 
    & 4  & 5.11  & 50   & 16 \\
    & &  6  & 13.35 & 50   & 64 \\
    & &  8  & 138.31  & 50   & 256 \\
    \midrule
    \addlinespace

    \multirow{4}{*}{\shortstack{RC dynamic logit\\(Figure \ref{fig:idset_rcd})}} 
    &
    \multirow{4}{*}{$Y_{i,t-1}$} 
    & 3  & $<1$  & 2500 & 8 \\
    & & 5  & $<1$  & 2500 & 32 \\
    & & 8  & $<1$  & 2500 & 256 \\
    & & 10 & $<1$  & 2500 & 1024 \\
    \bottomrule
  \end{tabular}

  \caption{Average time (in seconds) to obtain the outer bounds for a single model. The models correspond to those considered in Section \ref{sect:setcomparison} and Appendix \ref{app:dynamic}. Average times are based on one replication for each value of the main parameter ($\beta$, $A_2$ or $\gamma$) on a grid of 23 equidistant values between $-2$ and $2$. $\mathcal A _g$ is the grid for $\mathcal A$ used in outer bound estimation. $\mathcal Y$ is the support of the outcome variable. For more information on the models, see the information provided in Section \ref{sect:setcomparison} and Appendix \ref{app:dynamic}.}
  \label{tbl:comptimes}
\end{table}

\newpage

\section{Additional details}\label{sec:addsimdetails}
In this part we present some additional details on the results obtained in Sections \ref{sect:setcomparison} and \ref{sect:simulations}, and Appendix \ref{app:dynamic}. First we consider a comparison of the widths of outer bounds and the identified sets for all the examples considered in Section \ref{sect:setcomparison} and Appendix \ref{app:dynamic}. In particular (i) Figure \ref{fig:slogit_width} represents this information for the static logit model example of Figure \ref{fig:idset_logit}; (ii) Figure \ref{fig:rcslogit_width} provides this information for the random coefficient static logit model example of Figure \ref{fig:idset_rc}; (iii) Figure \ref{fig:dlogit_width} represents this comparison for the dynamic logit model analysis of Figure \ref{fig:idset_dlogit}; and finally (iv) Figure \ref{fig:rcdlogit_width} provides the same for the random coefficient dynamic logit analysis of Figure \ref{fig:idset_rcd}. We note that the results for the dynamic logit model do not include the widths of the identified sets as in this example the average effect is point-identified.

Second, for the static logit and random coefficient static logit models considered in the simulation analysis of Section \ref{sect:simulations}, we provide additional figures that show the width difference between (i) the outer bounds and the identified set, and (ii) the confidence bands and the outer bounds. The former provides a measure of the `loss' due to using outer bounds rather than the identified set, whereas the latter yields a measure of the `loss' due to imprecision in the estimation of confidence bands. Figures \ref{fig:slogit_uncertainty} and \ref{fig:rcslogit_uncertainty} provide this information for the static logit and random coefficient static logit model analyses of Figures \ref{fig:sim_logit_pb} and \ref{fig:sim_randc}, respectively.

\begin{figure}[H]
    \centering
    \begin{minipage}{0.5\textwidth}
        \centering
        \includegraphics[width=1\textwidth]{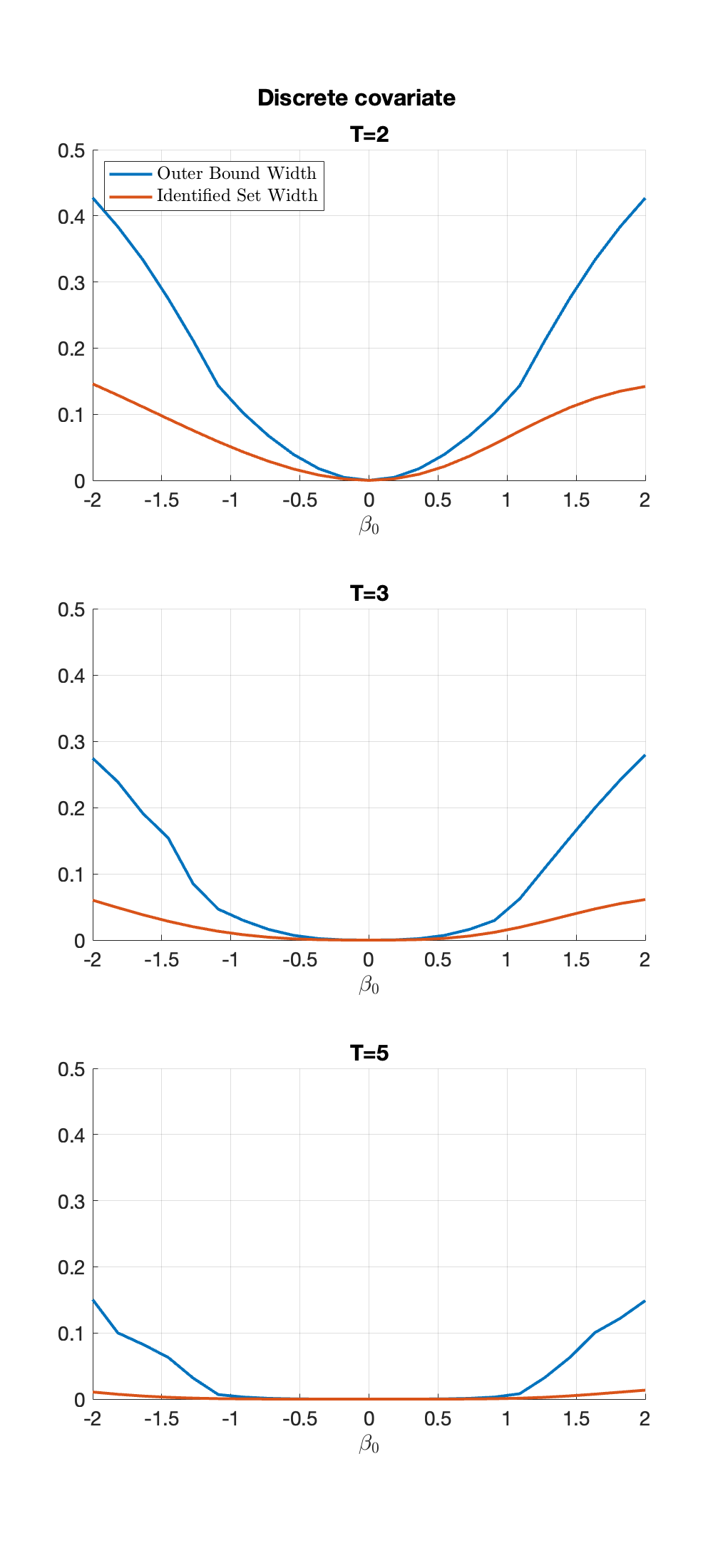}
    \end{minipage}\hfill
    \begin{minipage}{0.5\textwidth}
        \centering
        \includegraphics[width=1\textwidth]{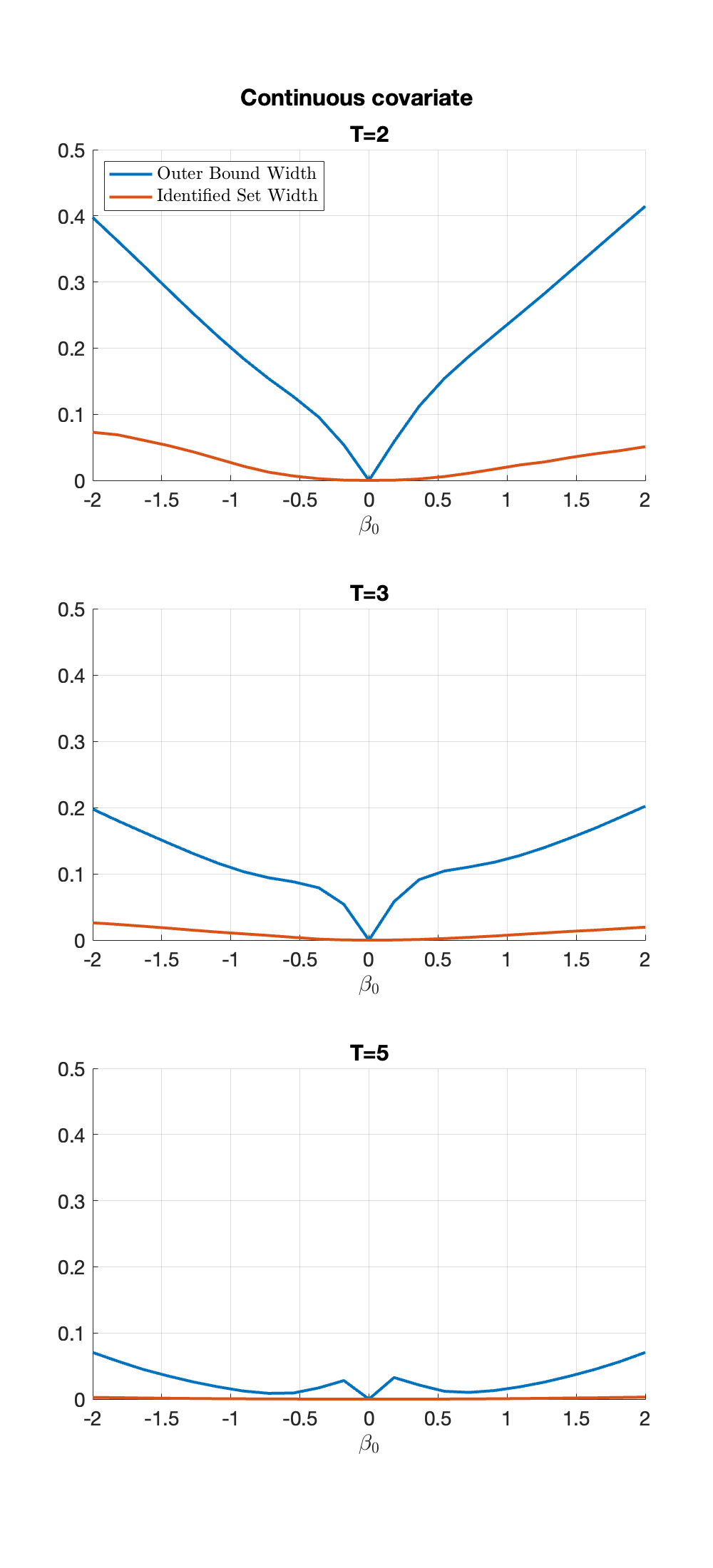}
    \end{minipage}
    \caption{Widths of the outer bounds and identified sets for the static logit model analysis of Figure \ref{fig:idset_logit}.}
    \label{fig:slogit_width}
\end{figure}

\begin{figure}[H]
    \centering
    \includegraphics[width=0.7\linewidth]{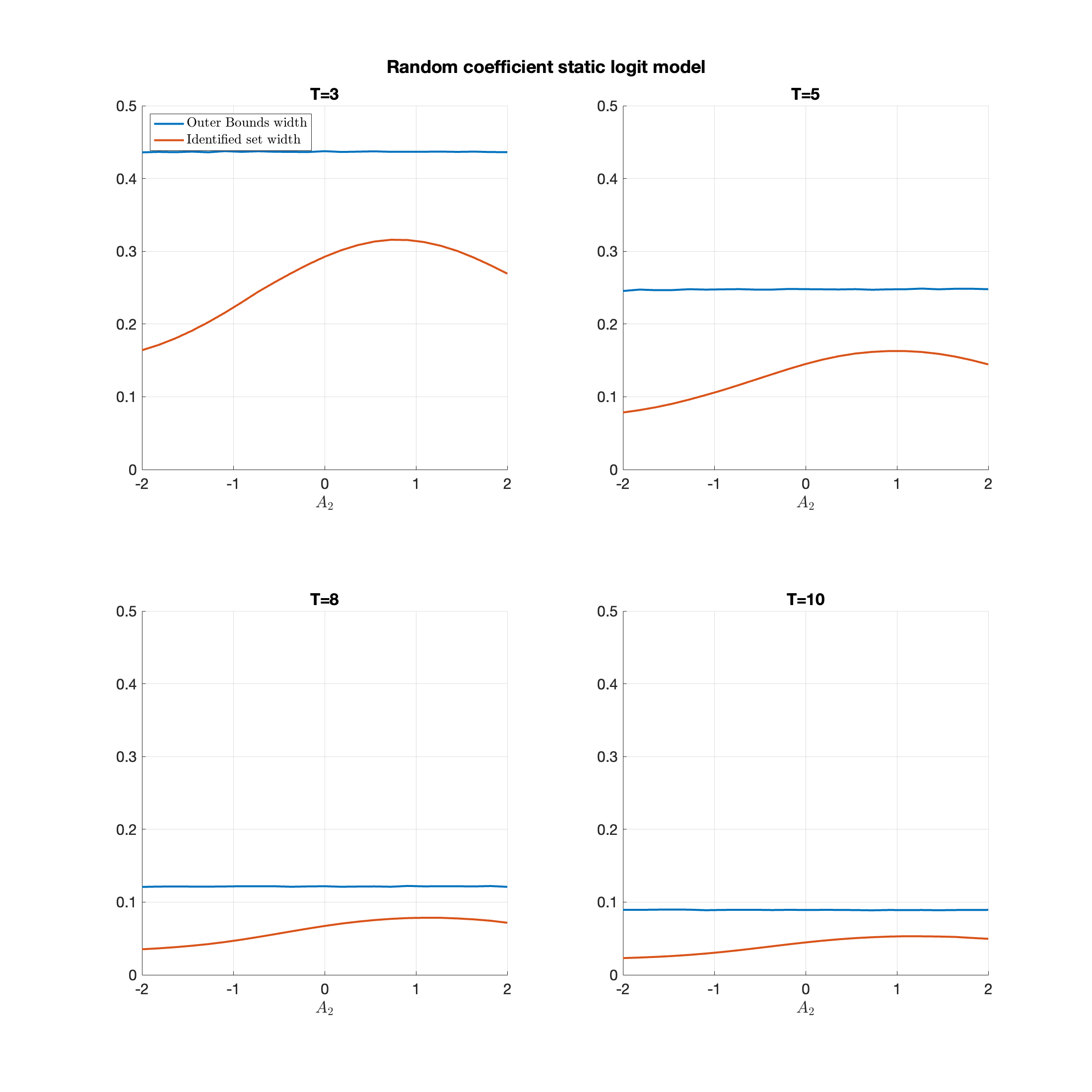}
    \caption{Widths of the outer bounds and identified sets for the random coefficient static logit model analysis of Figure \ref{fig:idset_rc}.}
    \label{fig:rcslogit_width}
\end{figure}

\begin{figure}[H]
    \centering
    \includegraphics[width=0.5\linewidth]{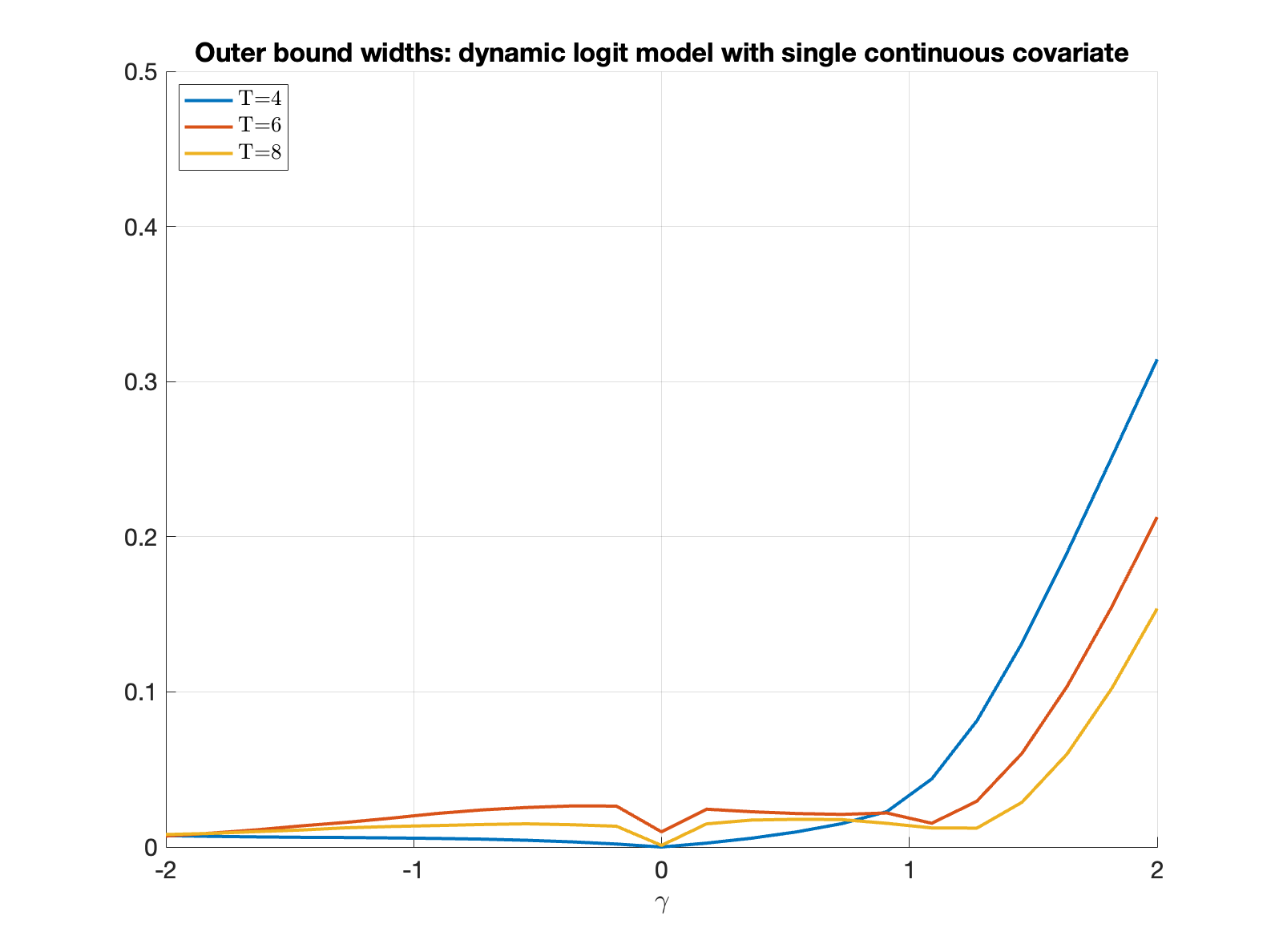}
    \caption{Widths of the outer bounds for the dynamic logit model analysis of Figure \ref{fig:idset_dlogit}.}
    \label{fig:dlogit_width}
\end{figure}

\begin{figure}[H]
\begin{center}
    \includegraphics[width=0.9\linewidth]{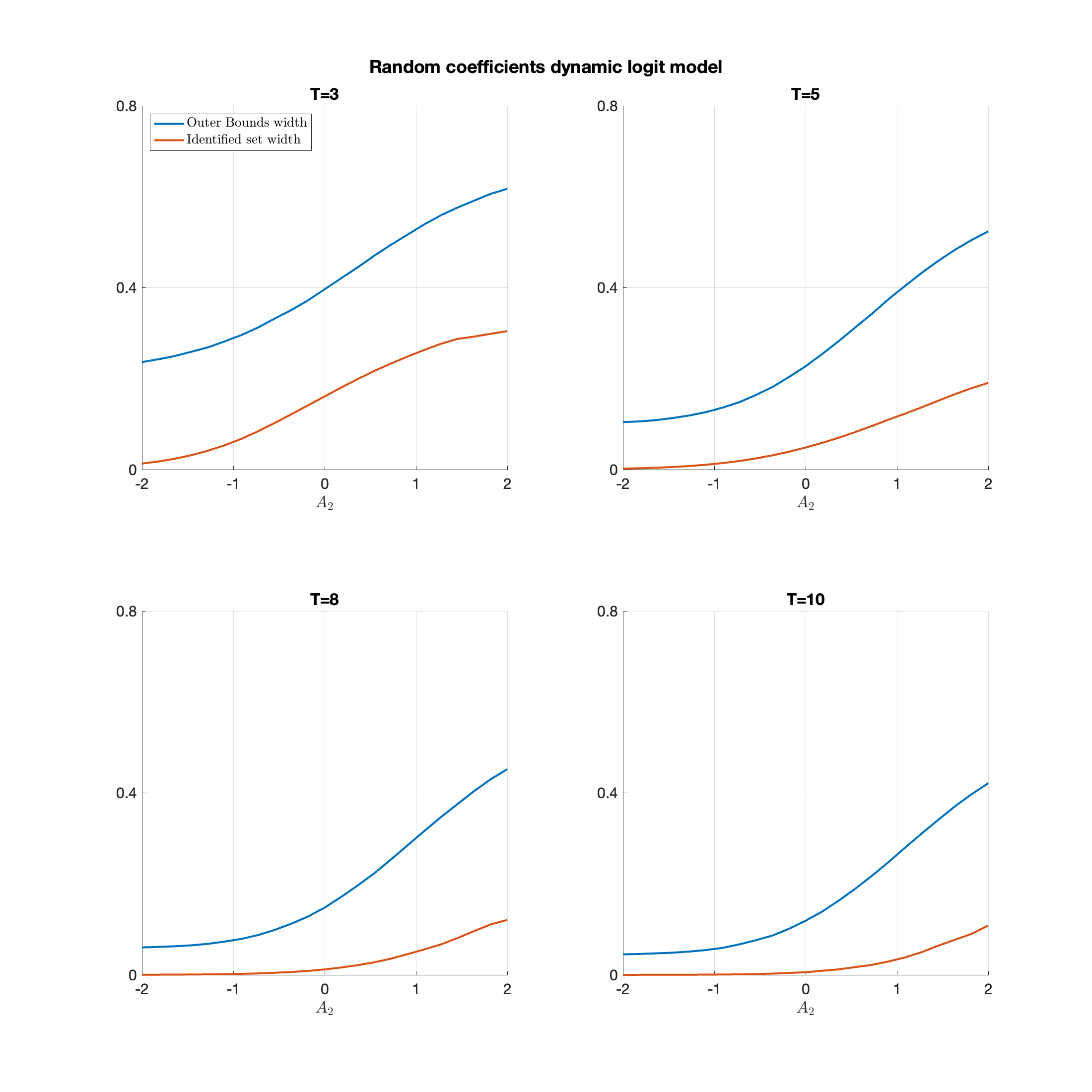}
    \caption{Widths of the outer bounds and identified sets for the random coefficient dynamic logit model analysis of Figure \ref{fig:idset_rcd}.}
    \label{fig:rcdlogit_width}
\end{center}
\end{figure}

\begin{figure}[H]
    \centering
    \begin{minipage}{0.5\textwidth}
        \centering
        \includegraphics[width=1\textwidth]{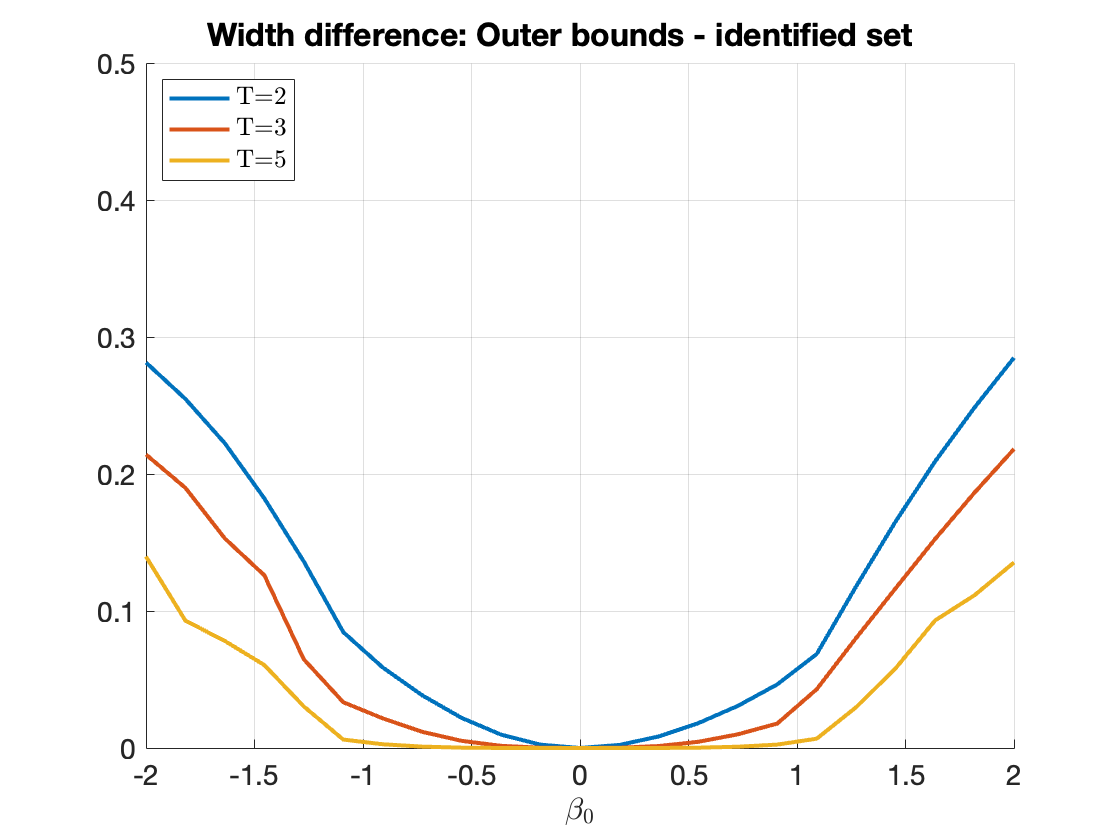}
    \end{minipage}\hfill
    \begin{minipage}{0.5\textwidth}
        \centering
        \includegraphics[width=1\textwidth]{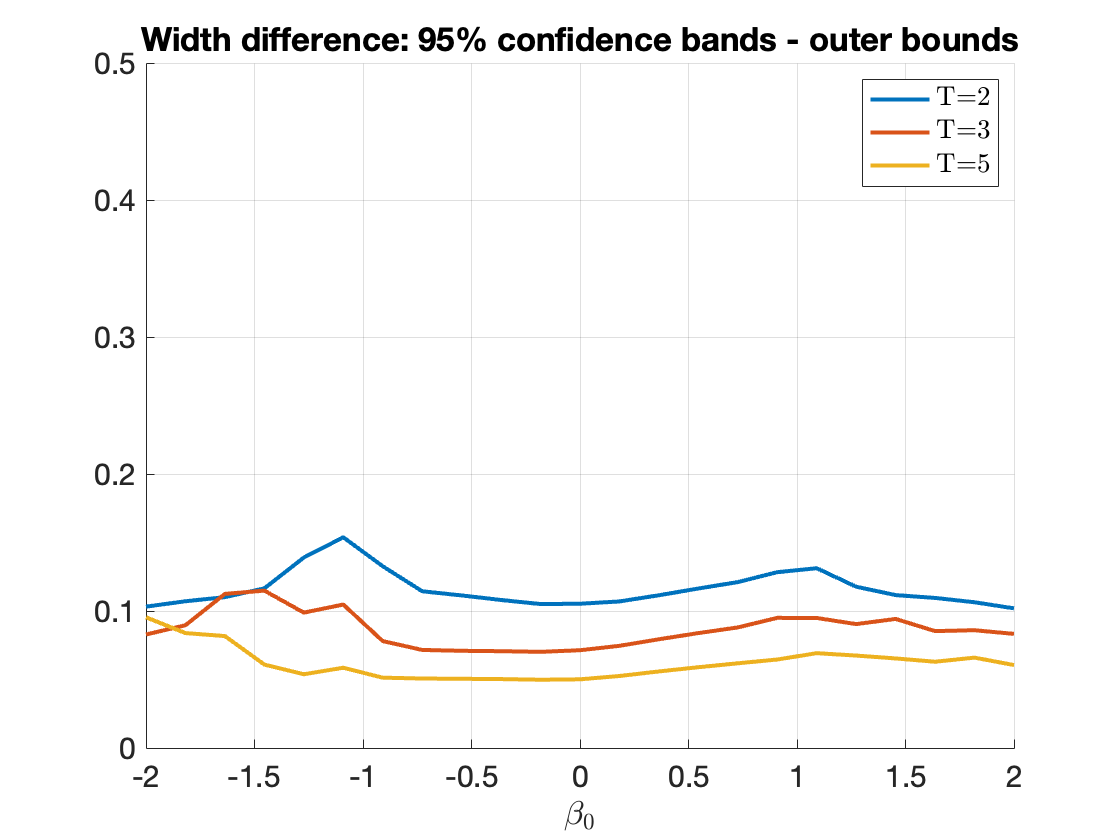}
    \end{minipage}
    \caption{Width differences between the outer bounds and the identified set (left panel) and the confidence bands and the outer bounds (right panel) for the static logit analysis of Figure \ref{fig:sim_logit_pb}, reported for $T\in\{2,3,5\}$.}
    \label{fig:slogit_uncertainty}
\end{figure}

\begin{figure}[H]
    \centering
    \begin{minipage}{0.5\textwidth}
        \centering
        \includegraphics[width=1\textwidth]{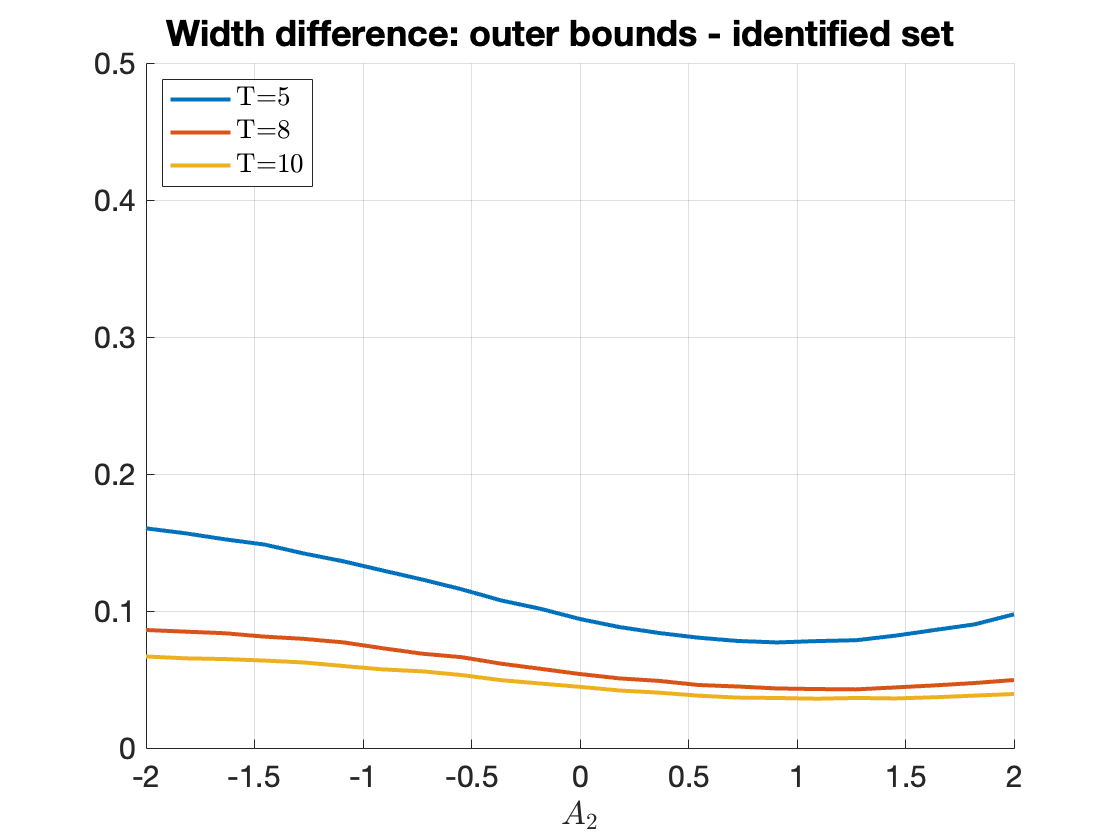}
    \end{minipage}\hfill
    \begin{minipage}{0.5\textwidth}
        \centering
        \includegraphics[width=1\textwidth]{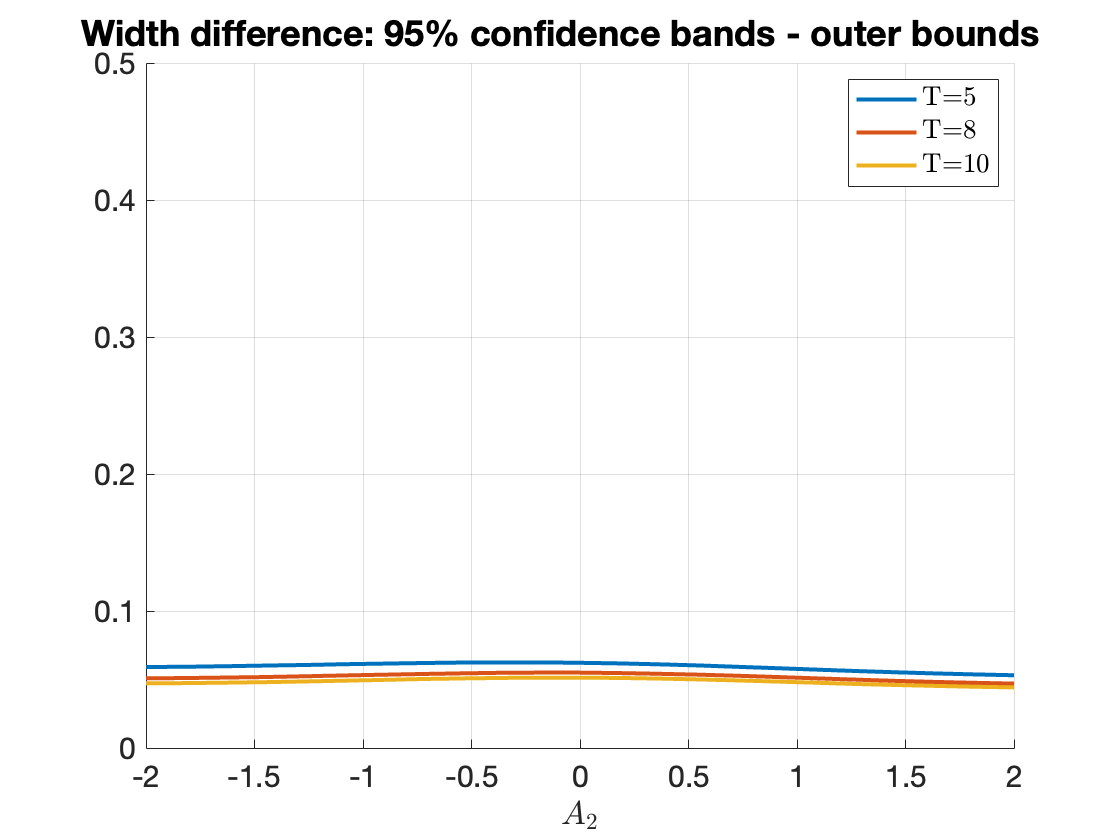}
    \end{minipage}
    \caption{Width differences between the outer bounds and the identified set (left panel) and the confidence bands and the outer bounds (right panel) for the random coefficient static logit analysis of Figure \ref{fig:sim_randc}.}
    \label{fig:rcslogit_uncertainty}
\end{figure}

\newpage
\section{Identified set calculations}\label{sec:archident}
\subsection{Identified set calculations underlying Figures \ref{fig:comb1} and \ref{fig:idanalysis_probit}}\label{sec:idmanual1}
To construct Figures \ref{fig:comb1} and \ref{fig:idanalysis_probit}, we obtain estimates of the identified sets for the average effects. In both cases, our objective is to compare inference for the set-identified average effect when one uses outer bounds vs when one uses the identified set. Therefore, in order to abstract away from estimation/identification of $\beta_0$, in the logit case we use the true $\beta_0$ whereas in the probit case we use the population identified set for $\beta_0$.

In Section \ref{sec:t10p1} we go through the algorithm for obtaining the identified set for $\beta_0$. This is followed by Section \ref{sec:aeidrecipe} where we provide the algorithm for estimating the identified set for the average effect when the identified set for $\beta_0$ is already given. This algorithm can also be used for the case where $\beta_0$ is point-identified: in that case, the algorithm has to be run only once, for the true value of $\beta_0$.

\subsubsection{Obtaining the identified set for $\beta_0$}\label{sec:t10p1}
The analysis is conducted for a grid of values of $\beta_0$ between -2 and 2; call this \texttt{b0grid}. In the following, we explain how the identified set for $\beta_0$ is obtained for a given value $\mathrm{\mathbf{b}}$ on this grid:
\begin{enumerate}
    \item Given $\mathrm{\mathbf{b}}$, we first generate a grid of values around $\mathrm{\mathbf{b}}$, which we will use to construct the identified set; call this grid $\texttt{idgrid}(\mathrm{\mathbf{b}})$. This grid is generated using quantiles of the normal distribution $N(\mathrm{\mathbf{b}},\sigma^2)$. $\sigma^2$ is user-specified and can be used to make the grid tighter/looser around $\mathrm{\mathbf{b}}$. 
    
    \item To check whether a value on $\texttt{idgrid}(\mathrm{\mathbf{b}})$ belongs to the identified set, we use the following algorithm:
    
    \begin{enumerate}
        \item Let $y\in \mathcal Y$ and $x\in \mathcal X$ be some chosen realizations of the outcome and covariate variables. 
        
        \item We discretize $A$ as follows: first we obtain a grid of $Q+1$ equidistant values between some user-specified lower and upper bounds; call this \texttt{agrid} which consists of $a_0,\ldots,a_Q$. Next, for the given $x\in \mathcal X$ obtain the average of the entries of $x$. Calling this $\overline x$, we next calculate the cdf of $N(\overline x - 1/2,1)$ at every value on \texttt{agrid}. Let the resulting set of values be called $\texttt{agridcdf}(x)$ with elements $\texttt{agridcdf}_0(x),\ldots,\texttt{agridcdf}_Q(x)$.
        Then, the conditional probabilities for the discretized $A$ are  given by the $Q \times 1$ vector $P_{A|x}$:
        \begin{align*}
            P_{A|x}
            =
            \begin{bmatrix}
                P(a_1|x)
                \\
                \vdots
                \\
                P(a_Q|x)
            \end{bmatrix}
            =
            \begin{bmatrix}
                \texttt{agridcdf}_{1}(x)- \texttt{agridcdf}_{0}(x)
                \\
                \vdots
                \\
                \texttt{agridcdf}_{Q}(x)- \texttt{agridcdf}_{Q-1}(x)
            \end{bmatrix}
            .
        \end{align*}

        \item Let $P_{Y|x}$ be the $|\mathcal Y|\times 1$ vector with the $r^{th}$ entry given by $P(y_r|x)$ where $y_r$ is the $r^{th}$ element of $\mathcal Y$. Let also the $|\mathcal Y| \times Q$ matrix $P_{Y|x,A}(\beta)$ be defined such that its row $r$ and column $q$ entry is given by $P(y_r|x,a_q;\beta)$.

        Let $x_c$ be the $c^{th}$ element of $\mathcal X$. Then, for the case where the true parameter value is $\mathbf{b}$, the model probabilities conditional on $x_c$ are given by the $|\mathcal Y| \times 1$ vector
        \begin{align*}
            P_{Y|x_c} = P_{Y|x_c,A}(\mathbf{b}) \,\, P_{A|x_c}.
        \end{align*}
        Consequently, the 
        $|\mathcal Y | \times |\mathcal X |$ matrix 
        \begin{align*}
            P_{Y|X} 
            =
            \begin{bmatrix}
                P_{Y|x_1} 
                &
                \cdots
                &
                P_{Y|x_{|\mathcal X|}} 
            \end{bmatrix},
        \end{align*}
        yields the full collection of model probabilities.

        \item We can now check whether a particular value $\beta$ on $\texttt{idgrid}(\mathrm{\mathbf{b}})$ belongs to the identified set. For a given $x\in \mathcal X$, consider the minimization problem
        \begin{align*}
            \Delta_Y(x, \beta)
            =
            \min_{\mathbf{p}_1(x),\ldots,\mathbf{p}_Q(x)}
            \left| \left| 
                P_{Y|x} 
                -
                P_{Y|x,A}(\beta) \, 
                \begin{bmatrix}
                    \mathbf{p}_1(x)
                    \\
                    \vdots
                    \\
                    \mathbf{p}_Q(x)
                \end{bmatrix}
            \right| \right|    
            ,
        \end{align*}
        subject to
        \begin{align}
            \sum_{q=1}^Q \mathbf{p}_q(x) = 1
            \qquad
            \text{and}
            \qquad
            \mathbf{p}_q(x) \geq 0 
            \quad
            \forall q=1,\ldots,Q.
            \label{T10r}
        \end{align}
        Let \texttt{tol} be some pre-specified tolerance value. Then, we retain a $\beta \in \texttt{idgrid}(\mathrm{\mathbf{b}})$ in the identified set if
        \begin{align*}
            \Delta_Y(x,\beta)<\texttt{tol} 
            \qquad
            \text{for all } 
            x\in \mathcal X.
        \end{align*}
        The default \texttt{tol} is set to $10^{-6}$. If this condition is violated, then this $\beta$ does not belong to the identified set.
        
    \end{enumerate}
    Running the above algorithm for all $\beta \in \texttt{idgrid}$ yields the identified set for $\mathbf{b}$.

    \item Going through these steps for every $\mathbf{b}\in \texttt{b0grid}$ yields the required information for constructing the identified set for $\beta_0$ across the grid of points \texttt{b0grid}.

\end{enumerate}

\subsubsection{Obtaining the estimated identified set for the average effect}\label{sec:aeidrecipe}

Here we obtain the identified set for the average effect separately for each $\mathbf{b}\in\texttt{b0grid}$. Therefore, the following steps have to be repeated for each $\mathbf{b}\in \texttt{b0grid}$. 
\begin{enumerate}
    \item We first generate a grid of \texttt{Qae} equidistant points
    between the upper and lower bounds of the identified set for $\beta_0$ when $\beta_0=\mathbf b$. Call this set of grid points \texttt{beta\_id\_grid}. \textbf{If $\beta_0$ is point-identified}, \texttt{beta\_id\_grid} for a given $\mathrm{\mathbf{b}}$ simply collapses to $\mathrm{\mathbf{b}}$.

    \item Given a grid of values for $A_i$ given by $a_1,\ldots,a_Q$, for every $\beta \in \texttt{beta\_id\_grid}$ we go through the following steps (if $\beta_0$ is point-identified, then $\texttt{beta\_id\_grid}=\mathrm{\mathbf{b}}$ so the following steps are executed only once, for $\beta=\mathrm{\mathbf{b}}$):
    \begin{enumerate}
        \item We first estimate some population probabilities by using their sample analogues (empirical probabilities). In particular, for every $x \in \mathcal X$ we obtain $\widehat P (x)$, and the $|\mathcal Y| \times 1$ vector $\widehat P _{Y|x}$ with its row $i$ entry given by $\widehat P (y_i | x)$ where $i=1,\dots,|\mathcal Y|$.

        \item Next, for each $x\in \mathcal X$ we obtain the 
        $|\mathcal Y| \times Q$ matrix $ P _{Y|x,A}(\beta)$ 
        with the row $i$ and column $q$ entry given by $ P(y_i|x,a_q;\beta)$ where $i=1,\ldots,|\mathcal Y|$ and 
        $q=1,\ldots,Q$. 
    
        \item We then obtain, for each $x\in \mathcal X$,
        \begin{align*}
            \widetilde P_{Y|x} (\beta)
            =
            P_ {Y|x,A}(\beta) 
            \begin{bmatrix}
                \widetilde{\mathbf{p}} _1 (x)
                \\
                \vdots
                \\
                \widetilde{\mathbf{p}} _Q (x)
            \end{bmatrix}
        \end{align*}
        where
        \begin{align*}
            \begin{bmatrix}
                \widetilde{\mathbf{p}} _1 (x)
                \\
                \vdots
                \\
                \widetilde{\mathbf{p}} _Q (x)
            \end{bmatrix}
            = 
            \arg\min_{\widetilde{\mathbf p} _1,\ldots,\widetilde{\mathbf p} _Q} \left|\left|
            \widehat P _{Y|x}
            -
            P_ {Y|x,A} (\beta)
            \begin{bmatrix}
                \widetilde{\mathbf{p}} _1
                \\
                \vdots
                \\
                \widetilde{\mathbf{p}} _Q
            \end{bmatrix}
            \right|\right|.
        \end{align*}
    
        \item In the next stage, defining
        the $1\times Q$ row vector $m_{x,A}(\beta)$ which contains $m(x,a_q,\beta)$ in its column $q$ entry, we solve, for every $x\in \mathcal X$,
        \begin{align*}
            \max_{\mathbf p _1 (x),\ldots,\mathbf p _Q (x)}
            m_{x,A}(\beta) 
            \begin{bmatrix}
                \mathbf p _1 (x)
                \\
                \vdots
                \\
                \mathbf p _Q (x)
            \end{bmatrix}
        \end{align*}
        subject to the constraints
        \begin{align*}
            \widetilde P _{Y|x}(\beta)
            =
            P_ {Y|x,A}(\beta) 
            \begin{bmatrix}
                \mathbf p _1 (x)
                \\
                \vdots
                \\
                \mathbf p _Q (x)
            \end{bmatrix}
            ,
            \qquad
            \sum_{q=1}^Q \mathbf p _q(x) 
            =
            1,
        \end{align*}
        and
        \begin{align*}
            \mathbf p _q(x) \geq 0 \qquad \forall q=1,\ldots,Q.
        \end{align*}
    
        \item We similarly obtain the solution to the minimization problem
        \begin{align*}
            \min_{\mathbf p _1 (x),\ldots,\mathbf p _Q (x)}
            m_{x,A}(\beta) 
            \begin{bmatrix}
                \mathbf p _1 (x)
                \\
                \vdots
                \\
                \mathbf p _Q (x)
            \end{bmatrix}
        \end{align*}
        subject to the same constraints. Let $U_x(\beta)$ and $L_x (\beta)$ be the solutions to these problems, respectively.
        
        \item Once we obtain these solutions for all $x\in \mathcal X$ we calculate
        \begin{align*}
            L(\beta) = \sum_{x\in \mathcal X} L_x(\beta) \widehat P (x)
            \qquad
            \text{and}
            \qquad
            U(\beta) = \sum_{x\in \mathcal X} U_x(\beta) \widehat P (x).
        \end{align*}

    \end{enumerate}

    \item The final estimated bounds on the average effect are given by
    \begin{align*}
        L = \min_{\beta \in \texttt{beta\_id\_grid}} L(\beta)
        \qquad
        \text{and}
        \qquad
        U = \max_{\beta \in \texttt{beta\_id\_grid}} U(\beta).
    \end{align*}
        
\end{enumerate}

\subsection{Identified set calculations underlying Figures \ref{fig:idset_logit}, \ref{fig:idset_rc}, \ref{fig:idset_dlogit} and \ref{fig:idset_rcd}}\label{sec:idmanual2}

The identified sets in Figures \ref{fig:idset_logit}, \ref{fig:idset_rc}, \ref{fig:idset_dlogit} and \ref{fig:idset_rcd}---where the true parameter value $\beta_0$ is taken as given---are the `population' identified sets, in the sense that they are based on the true choice probabilities implied by the model structure. To be more precise, given the value of $\beta_0$ the algorithm is based on the following steps:
\begin{enumerate}
    \item First, we obtain the discretized support of $A_i$, given by $\left\{ a_1,\ldots,a_Q\right\}$ and then obtain $P(a_q|x)$
    for all $x\in \mathcal X$ and $a_q$, $q=1,\ldots,Q$. The discretization of the support of $A_i$ and the calculation of $P(a_q|x)$ are based on the DGP and so differ from example to example. 
    
    As an example, here we illustrate the case of the static logit model with a discrete covariate; to recall from equation \eqref{eq:idlogitdisc}, this model is given by
    \begin{align*}
        Y_{it} = 1\left\{ X_{it}\beta + A_i \geq \varepsilon_{it} \right\}, \quad A_i \sim N(0,1),  \quad X_{it} = 1\left\{ A_i \geq \eta_{it} \right\}, \quad \eta_{it} \sim N(0,1).
    \end{align*}
    The discretized support of $A_i$ is given by 
    $\{a_1,\ldots,a_Q\}$ where $a_q=\Phi^{-1}(q/(Q+1))$. The conditional probability $P(x|a_q)$ is then obtained as
    \begin{align*}
        P(x|a_q) = P(x_1|a_q) \times \ldots \times P(x_T|a_q),
    \end{align*}
    where $x_1,\ldots,x_T$ are the individual entries of the $T\times 1$ vector $x$, and
    \begin{align*}
        P(x_t|a_q) 
        = 
        1\{x_t=1\} \Phi(a_q) + 1\{x_t=0\} (1-\Phi(a_q)),
    \end{align*}
    which mimics the DGP for the discrete covariate. Next, due to the way the support $\mathcal A$ is discretized, we have 
    $P(a_q)=1/Q$ for all $q=1,\ldots,Q$. From these, we obtain $P(x, a_q) = P(x|a_q)/Q$, 
    $P(x) = \sum_{q=1}^Q P(x,a_q) $, and finally
    \begin{align*}
        P(a_q|x) = \frac{P(a_q,x)}{P(x)}.
    \end{align*}

    \item For any $x\in \mathcal X$ and $y\in \mathcal Y$, 
    the choice probabilities are then given by
    \begin{align*}
        P(y|x) = \sum_{q=1}^Q P(y|x,a_q;\beta_0) \times P(a_q|x).
    \end{align*}

    \item Let $P_{Y|x}$ be the $|\mathcal Y|\times 1$ vector with the $r^{th}$ entry given by $P(y_r|x)$ where $y_r$ is the $r^{th}$ element of $\mathcal Y$. Let also the $|\mathcal Y| \times Q$ matrix $P_{Y|x,A}$ be defined such that its row $r$ and column $q$ entry is given by $P(y_r|x,a_q;\beta_0)$. Define also the $1\times Q$ row vector $m_{x,A}$ which contains $m(x,a_q,\beta_0)$ in its column $q$ entry. All these terms are obtained for all $x \in \mathcal X$.

    \item\label{1b} We solve the following optimization problem
        \begin{align}
                \max_{\mathbf{p}_1(x),\ldots,\mathbf{p}_Q(x)} m_{x,A}
                \begin{bmatrix}
                    \mathbf{p}_1(x)
                    \\
                    \vdots
                    \\
                    \mathbf{p}_Q(x)
                    \label{opto1}
                \end{bmatrix}
        \end{align}
        subject to the conditions
        \begin{align}
            P_{Y|x} = P_{Y|x,A}
            \begin{bmatrix}
                \mathbf{p}_1(x)
                \\
                \vdots
                \\
                \mathbf{p}_Q(x)
            \end{bmatrix}
            ,
            \qquad
            \sum_{q=1}^{Q} \mathbf{p}_q(x) =1,
            \label{optr1}
        \end{align}
        and
        \begin{align}
            \mathbf{p}_q(x)
            \geq
            0
            \qquad
            \forall
            q=1,\ldots,Q .
            \label{optr2}
        \end{align}
        Call the value of the objective function at the solution
        \begin{align*}
            U(x).
        \end{align*}
        Similarly, solve, with respect to the same conditions as in \eqref{optr1} and \eqref{optr2}, the optimization problem
        \begin{align}
                \min_{\mathbf{p}_1(x),\ldots,\mathbf{p}_Q(x)} m_{x,A}
                \begin{bmatrix}
                    \mathbf{p}_1(x)
                    \\
                    \vdots
                    \\
                    \mathbf{p}_Q(x)
                    \label{opto2}
                \end{bmatrix}
        \end{align}
        and call the value of the objective function at the solution
        \begin{align*}
            L(x).
        \end{align*}

        \item Repeating this algorithm for all $x\in \mathcal X$ yields the set $\{(L(x),U(x)):x\in \mathcal X\}$. The final lower and upper bounds of the identified set for the average effect are then given by
        \begin{align*}
            L = \sum_{x\in \mathcal X } L(x) P(x)
            \qquad
            \text{and}
            \qquad
            U = \sum_{x\in \mathcal X } U(x) P(x),
        \end{align*}
        respectively.

\end{enumerate}

\end{appendix}

\end{document}